\documentclass[sigconf, nonacm]{acmart}

\usepackage{cleveref}
\usepackage{quiver}
\usepackage{caption}
\usepackage{subcaption}
\usepackage{mathtools}
\usepackage{graphicx}
\graphicspath{ {./imgs/} }

\usepackage[colorinlistoftodos]{todonotes}
\usepackage{listings}
\usepackage{xcolor}

\usepackage{bm} 

\usepackage[small,nohug,heads=vee]{diagrams}
\diagramstyle[labelstyle=\scriptstyle]

\definecolor{light-gray}{RGB}{230,230,230}

\newcommand{\update}[1]{#1}
\newcommand{\required}[1]{#1}

\newcommand{\eat}[1]{}

\newcommand{\cegis}{\textsf{CEGIS}}
\newcommand{\eqsat}{\textsf{EQSAT}}
\newcommand{\zzz}{\texttt{\texttt{z3}}}
\newcommand{\egraph}{\textsf{e-graph}}

\newcommand{\eclass}{\textsf{e-class}}
\newcommand{\eclasses}{\textsf{e-classes}}
\newcommand{\enode}{\textsf{e-node}}
\newcommand{\enodes}{\textsf{e-nodes}}
\newcommand{\prem}{\textsf{PreM}}

\newcommand{\rai}{relational\underline{AI}}

\newcommand{\defeq}{\stackrel{\text{def}}{=}}

\newcommand{\B}{\mathbb B} 
\newcommand{\N}{\mathbb N} 
\newcommand{\R}{\mathbb R} 

\newcommand{\set}[1]{\{#1\}}                    
\newcommand{\setof}[2]{\{{#1}\mid{#2}\}}        

\newcommand{\lfp}{\text{\sf lfp}}

\newcommand{\one}{\bar 1}
\newcommand{\zero}{\bar 0}

\newcommand{\cd}{\text{ :- }}

\newcommand{\name}{\text{\sf Datalog}^\circ}
\newcommand{\trop}{\text{\sf Trop}}
\newcommand{\totrop}[1]{[#1]_{\infty}^{0}}

\newcommand{\bA}{\mathbf{A}}
\newcommand{\bB}{\mathbf{B}}

\newcommand{\bv}{\mathbf{v}}

\newcommand{\property}{Property B}

\setcopyright{acmcopyright}
\copyrightyear{2021}
\acmYear{2021}
\acmDOI{10.1145/1122445.1122456}

\begin{document}

\title{Optimizing Recursive Queries with Program Synthesis}


\author{Yisu Remy Wang}
\affiliation{%
 \institution{University of Washington}
 \institution{\rai{}}
 \country{USA}
}

\author{Mahmoud Abo Khamis}
\affiliation{%
 \institution{\rai{}}
 \country{USA}
}

\author{Hung Q. Ngo}
\affiliation{%
 \institution{\rai{}}
 \country{USA}
}

\author{Reinhard Pichler}
\affiliation{%
 \institution{TU Wien}
 \country{Austria}
}

\author{Dan Suciu}
\affiliation{%
 \institution{University of Washington}
 \institution{\rai{}}
 \country{USA}
}


\begin{abstract}
Most work on query optimization has concentrated on loop-free queries.
However, data science and machine learning workloads today typically involve recursive or
iterative computation. In this work, we propose a novel framework for optimizing
recursive queries using methods from program synthesis. In particular, we introduce a simple
yet powerful optimization rule called the ``FGH-rule'' which aims to find a faster way to
evaluate a recursive program. The solution is found by making use of powerful tools, such as
a program synthesizer,
an SMT-solver, and an equality saturation system.
We demonstrate the strength of the optimization by showing
that the  FGH-rule can lead to speedups up to 4 orders of magnitude on
three, already optimized Datalog systems.
\end{abstract}






\maketitle

\section{Introduction}

\label{sec:intro}


Most database systems are designed to support primarily non-recursive
(loop-free) queries.  Their optimizers are based on the rule-driven,
cost-based Volcano architecture, designed specifically for optimizing
non-recursive query plans.  However, most data science and machine
learning workloads today involve some form of recursion or iteration.
Examples include finding the connected components of a graph,
computing the page rank, computing the network centrality, minimizing
an objective function using gradient descent, etc.  The importance of
supporting recursive queries has been noted by system designers.
Some modern data analytics systems, like Spark or Tensorflow, support
for-loops. The SQL standard defines a limited
form of recursive queries, using the
\texttt{with} construct, and some popular engines, like Postgres or
SQLite, do support this restricted form of recursion.

Datalog is a language
designed specifically for recursive queries, and it is gaining in
popularity~\cite{DBLP:conf/datalog/AlvaroMCHMS10,
DBLP:reference/db/RoscoeL18,
10.1145/3452021.3458815,
10.1145/1989323.1989456,
BigDatalog,
DBLP:journals/pvldb/FanZZAKP19,
DBLP:journals/tkde/SeoGL15,
10.14778/2824032.2824052,
francislandau-vieira-eisner-2020-wrla}.
But the optimization problem for recursive queries is
much less studied.  A datalog program consists of multiple rules,
defining several, mutually recursive relations, and one
distinguished relation name which is the output of the program.  The
effect of the program consist of repeatedly applying the rules,
sometimes called the {\em body} of the program, until a fixpoint is
reached, then it returns the output relation.
Datalog engines typically optimize the loop body, without optimizing
the actual loop.  The few systems that do (for example Soufflé) apply
only limited optimization techniques, like magic set optimization and
semi-naive evaluation, which are mainly restricted to positive Datalog
queries.

In this paper we describe a new query optimization framework for
recursive queries.  Our framework replaces a recursive program with
another, equivalent recursive program, whose body may be quite
different, and thus focuses on optimizing the recursive program as a
whole, not on optimizing its body in isolation; the latter can be done
separately, using standard query optimization techniques.  Our
optimization is based on a novel rewrite rule for recursive programs,
called the FGH-rule, which we implement using {\em program synthesis},
a technique developed in the programming languages and verification
communities. We introduce a new method for inferring loop invariants,
which extends the reach of the FGH-rule, and also show how to use
global constraints on the data for semantic optimizations using the
FGH-rule.  We explain these points in some details next.

{\bf The FGH-Rule} At the core of our approach is a novel, yet very
simple rewrite rule, called the FGH-rule (pronounced {\em fig-rule}),
which can be used to prove that two recursive programs are equivalent,
even when their loop bodies are quite different.  We show that the
FGH-rule can express previously known optimizations for Datalog,
including magic sets and semi-naive evaluation, and also a wide range
of new optimizations.  The optimized program is often significantly
more efficient than the original program, and sometimes can have a
strictly lower asymptotic complexity.  We implemented a
source-to-source optimizer using the FGH-rule, evaluated its
effectiveness on several Datalog systems, and observed speedups of up
to $4$ orders of magnitude (Sec.~\ref{sec:eval}).

For a taste of the FGH-optimization, consider the following example,
from~\cite{DBLP:journals/tplp/ZanioloYDSCI17,DBLP:conf/amw/ZanioloYIDSC18}:
compute the connected components of an undirected graph $E(x,y)$.  The
Datalog program in Fig.~\ref{fig:cc} (a) achieves this by first
computing the transitive closure relation $TC(x,y)$, then computing a
$\min$-aggregate query
assigning to every node $x$ the smallest label $L[y]$ of all nodes $y$ reachable from $x$.
In contrast, the optimized program in Fig.~\ref{fig:cc} (b)
computes directly the CC label of every node $x$ as the  minimum of its own label
and the smallest CC label of its neighbors,
using a single recursive rule with
$\min$-aggregation.
The space complexity of the transitive closure is
$O(n^2)$, which, in practice, is prohibitively expensive on large graphs.
On the other hand, the optimized query has space complexity $O(n)$.

\begin{figure}
  \centering
\fcolorbox{black}{light-gray}{\parbox{0.45\textwidth}{
\begin{align*}
  TC(x,y) \cd &  [x=y] \vee \exists z(E(x,z) \wedge TC(z,y)) \\
  CC[x] \cd & \min_y \setof{L[y]}{TC(x,y)}
\end{align*}
}}
\newline\null\hfill(a)\hfill\null

\fcolorbox{black}{light-gray}{\parbox{0.45\textwidth}{
\begin{align*}
  CC[x] \cd & \min(L[x], \min_y \setof{CC[y]}{E(x,y)})
\end{align*}
}}
\newline\null\hfill(b)\hfill\null

\caption{Unoptimized (a) and optimized (b) Datalog program for the connected components of an undirected graph.}
  \label{fig:cc}
\end{figure}

{\bf Pattern Matching vs.\ Query Synthesis} Applying the FGH-rule is
an instance of {\em query rewriting using views}.  In that problem we
are given a set of view expressions and a query, and the task is to
rewrite the query to use the view expressions rather than the base
relations. This problem has been extensively studied in the
literature~\cite{DBLP:journals/vldb/Halevy01}, and today's database
systems perform it using pattern
matching~\cite{DBLP:conf/sigmod/GoldsteinL01}.  This is a form of
transformational synthesis, where every candidate query rewriting is
guaranteed to be correct, because it is obtained by applying a limited
set of manually crafted rules (patterns), which are guaranteed to be
correct.  However, the FGH-rule often requires exploring a very large
space, which cannot be covered by a limited set of rules.  In this
paper we propose to use {\em counterexample-guided inductive
  synthesis} (\cegis) for this purpose, which is a technique designed
for program
sketching~\cite{DBLP:conf/asplos/Solar-LezamaTBSS06,DBLP:conf/tacas/TorlakJ07}.
When applied to our context, we call this technique {\em query
  synthesis}.  Unlike pattern matching, query synthesis explores a
much larger space, by examining rewritings that are not necessarily
correct, and need to be checked for correctness by a verifier (\zzz\ in
our system).
The verifier also produces a small counterexample
database for each rejected candidate, and these counterexamples are
collected by the synthesizer and used to produce only candidate
rewritings that pass all the previous counterexamples, which
significantly prunes the search space of the synthesizer.  We report
in Sec.~\ref{sec:eval} synthesis times of less than 1 second, even for
complex queries that use global constraints and require inferring loop
invariants.



{\bf Monotone Queries and Semiring Semantics} Datalog is, by
definition, restricted to monotone queries.  This ensures that every
query has a well-defined semantics, namely the least fixpoint of its
immediate consequence operator.  Existing optimizations for Datalog,
like semi-naive evaluation and magic set rewriting, apply mainly to
monotone queries. Even stratified negation can (if at all)
only be handled by imposing appropriate restrictions~\cite{DBLP:journals/corr/abs-1909-08246}.
But queries that contain aggregates or negation
(expressed in SQL via subqueries) are not monotone, and most systems
that support recursion prohibit the combination of aggregates and
recursion.  This has two shortcomings: it limits what kind of queries
the user can express, and also prevents many of our FGH-rewritings.
For example, the simple computation of connected components in
Fig.~\ref{fig:cc} (a) can be expressed in PostgreSQL, or in SQLite, or
in Soufflé,
because the first rule uses
only recursion and the second rule uses only aggregation.  However,
none of these systems accepts the query in Fig.~\ref{fig:cc} (b),
because it combines recursion and aggregation.\footnote{Prior
  work~\cite{DBLP:conf/pods/GangulyGZ91,DBLP:journals/tkde/SeoGL15}
  has proposed extending Datalog with $\min$ and $\max$ aggregates by
  explicitly re-defining the semantics of recursive rules with
  aggregates.  Our approach keeps the standard least fixpoint
  semantics, but generalizes the semiring.}  In order to express such
queries, in this paper we propose an extension of Datalog, following
the approach in~\cite{DBLP:conf/pods/GreenKT07}, where the relations
are interpreted over {\em ordered semirings}.

A semiring is an
algebraic structure with two operations, $\oplus, \otimes$.
Traditional Datalog corresponds to the Boolean semiring, where these
two operators are $\vee, \wedge$, while the query in Fig.~\ref{fig:cc}
(b) is over the Tropical semiring, where the two operators are
$\min, +$ (reviewed in Sec.~\ref{sec:background}).
We call this extension of Datalog to ordered semirings $\name$,
pronounced ``Datalogo'', where the circle represents
the semiring.  In $\name$ recursion is still restricted to monotone\footnote{This
monotonicity is over the partial order from the ordered semiring.}
queries, but monotone queries in $\name$ include queries with aggregates, over an
appropriate semiring.  The query in Fig.~\ref{fig:cc} (b) is monotone
over the (ordered) tropical semiring.

{\bf Loop Invariants} One difficulty in reasoning about loops in
programming languages is the need to discover loop invariants.  Some
(but not all) applications of the FGH-rule also require the discovery
of loop invariants.  We describe a novel technique for inferring loop
invariants for $\name$ programs, by combining symbolic execution with
equality saturation, and using a verifier.  We execute symbolically
the recursive program for a very small number of iterations (five in
our system), obtain query expressions for the IDBs (the recursive
predicates), and construct all identities satisfied by the IDBs.
Then, we retain only candidates that hold at each iteration, and check
each candidate for correctness using the SMT solver.  By inferring and
using loop invariants we show that we can significantly improve some
instances of magic-set optimizations from the literature: we call the
new optimization {\em beyond magic}.


{\bf Constraints and Semantic Optimizations} Optimizations that are
conditioned on certain constraints on the database are known as {\em
semantic optimizations}~\cite{DBLP:journals/debu/RamakrishnanS94}.
SQL optimizers routinely use key constraints and foreign key
constraints to optimize queries.  More powerful optimizations can
be performed using the chase and back-chase
framework~\cite{DBLP:conf/vldb/DeutschPT99,DBLP:conf/sigmod/PopaDST00},
and these include optimizations under inclusion constraints, or
conditional functional dependencies, or tuple generating constraints.
However, all constraints that are useful for optimizing non-recursive
queries are {\em local}.  In contrast, the FGH-rule optimizes
recursive queries, and therefore it can also exploit {\em global}
constraints.  For example, suppose the database represents a graph,
and the global constraint states that the graph is a tree.  This
global constraint does not help optimize non-recursive queries, but
can be used to great advantage to optimize some recursive queries; we
give details in Sec.~\ref{subsec:constraints}.

{\bf Equality Saturation Systems} Throughout our optimizer we need to
manage symbolic expressions of queries, and their equivalence classes,
as defined by a set of rules.  We uses for this purpose a
state-of-the-art Equality Saturation System (\eqsat),
EGG~\cite{DBLP:journals/pacmpl/WillseyNWFTP21}.  We show how to use
\eqsat\ for checking equality under constraints, inferring loop
invariants, and ``denormalization'' (which is essentially query
rewriting using views).

\textbf{Related Work} Our work was partially inspired by the \prem\
condition, described by Zaniolo et
al.~\cite{DBLP:journals/tplp/ZanioloYDSCI17}, which, as we shall
explain, is a special case of the FGH-rule.  Unlike our system, their
implementation required the programmer to check the \prem\ manually,
then perform the corresponding optimization.
%
Seveal prior systems  leveraged SMT-solvers to reason about
query languages~\cite{
  DBLP:conf/icfem/VeanesGHT09,
  DBLP:conf/cidr/ChuWWC17,
  DBLP:conf/cav/GrossmanCIRS17,
  DBLP:conf/sosp/SchlaipferRLS17,
  DBLP:journals/pacmpl/0001DLC18};
but none of these  consider recursive queries.
%
Datalog synthesizers have been described in~\cite{
  DBLP:conf/cp/AlbarghouthiKNS17,
  DBLP:conf/sigsoft/SiLZAKN18,
  DBLP:conf/ijcai/SiRHN19,
  DBLP:journals/pvldb/WangSCPD20,
  DBLP:journals/pacmpl/RaghothamanMZNS20}.
Their setting is different from ours:
the specification is given by input-output examples, and the
synthesizer needs to produce a program that matches all examples.
A design choice that we made, and which sets us further aside from the
previous systems, is to use an existing \cegis\ system, Rosette; thus,
we do not aim to  improve the \cegis\ system itself, but optimize the
way we use it.

%

{\bf Contributions} In summary, the main contribution of this paper
consists of a new, principled and powerful method for optimizing
recursive queries.  We make the following specific contributions:
\begin{itemize}
\item We introduce a simple optimization rule for recursive queries,
  called the FGH-rule (Sec.~\ref{sec:fgh}).
\item We show how the FGH-rule captures known optimizations (magic
  sets, $\prem$, semi-naive), (Sec.~\ref{subsec:simple:examples}),  some
  new optimizations (Sec.~\ref{subsec:loop:invariants}), and
  optimizations under global constraints
  (Sec.~\ref{subsec:constraints}).
\item We present our novel framework for query optimization via
the FGH-rule  (Sec.~\ref{sec:optimization}).
\item We describe how an SMT solver
(Sec.~\ref{sec:verification}) and a
\cegis\ system (Sec.~\ref{subsec:synthesis})
can be profitably integrated into our FGH-optimizer.
\item We describe how to use an \eqsat\ system for various tasks in
  the FGH optimizer: loop-invariant inference, denormalization,
  and checking equivalence under constraints (Sec.~\ref{sec:semantic-opt}).
\end{itemize}



\section{Background}

\label{sec:background}


{\bf Datalog} \update{A {\em relation} of arity $k$ is a finite subset
  of $D^k$, where $D$ is a fixed domain.  The abbreviations} EDB and
IDB stand for {\em Extensional Database} and {\em Intensional
  Database}, and represent the base relations and the computed
relations respectively.  A {\em Datalog rule} has the form:
\begin{align*}
  R_0(\texttt{vars}) \cd & R_1(\texttt{vars}_1) \wedge \cdots \wedge R_m(\texttt{vars}_m)
\end{align*}
where $R_0$ is an IDB, and $R_1, \ldots, R_m$ are IDBs or EDBs.  The
rule is {\em safe} if every variable occurs in at least some predicate
in the body, and the rule is {\em linear} if its body contains at most
one IDB.  A {\em Datalog program} consists of a set of possibly
mutually recursive rules.  Usually, only a subset of the IDB
predicates are returned to the user, and we will call them the {\em
  answer} IDBs.  The {\em Immediate Consequence Operator}, ICO, is the
mapping on the IDB predicates that consists of one application of all
the Datalog rules.  The {\em semantics} of a Datalog program is given
by the least fixpoint of its ICO.  The {\em naive evaluation
  algorithm} consists of repeatedly applying the ICO until the IDBs no
longer change.

\update{
  In this paper we will combine multiple rules with the same head into
  a single rule by OR-ing their bodies, and writing explicitly all
  existential quantifiers.  This is a common convention used in the
  literature, see e.g.,~\cite{DBLP:journals/jlp/Fitting91}.  For
  example the following datalog program, which computes the transitive
  closure of a relation $E$,
\begin{align*}
  TC(x,y) \cd & E(x,y) \\
  TC(x,y) \cd & E(x,z) \wedge TC(z,y)
\end{align*}
becomes $TC(x,y) \cd E(x,y) \vee \exists z (E(x,z) \wedge TC(z,y))$.
}

{\bf (Pre-)Semirings} A {\em pre-semiring} is a tuple
$\bm S = (S, \oplus, \otimes, \zero, \one)$ where \update{$\oplus$ is
  commutative,} both $\oplus, \otimes$ are associative, have
identities $\zero$ and $\one$ respectively, and $\otimes$ distributes
over $\oplus$.  \update{When $\otimes$ is commutative, then we call
  $\bm S$ a {\em commutative} pre-semiring.  All pre-semirings in this
  paper are commutative, and we will simply refer to them as
  pre-semirings.}  When the equality $x \otimes \zero = \zero$ holds
for all $x$, then it is called a {\em semiring}.  An {\em ordered}
pre-semiring is a pre-semiring with a partial order $\preceq$, where
both $\oplus, \otimes$ are monotone operations.  When the partial
order is defined by $x\preceq y$ iff $\exists z, x\oplus z = y$ then
it is called the {\em natural order}.  Examples of ordered
(pre-)semirings are the Booleans
$\B = (\set{0,1}, \vee, \wedge, 0, 1)$, the closed natural numbers
$\N^\infty = (\N \cup \set{\infty},+,*,0,1)$, the tropical semiring
$\trop = (\N \cup \set{\infty}, \min, +, \infty,0)$, the reversed
tropical semiring $\trop^r=(\N, \max, +, 0, 0)$, the lifted naturals
and lifted reals $\N_\bot = (\N \cup \set{\bot}, +, *, 0, 1)$,
$\R_\bot = (\R \cup \set{\bot}, +, *, 0, 1)$, where
$\bot + x = \bot * x = \bot$.  The structures $\B, \N^\infty, \trop$
are semirings, the others are pre-semirings.  $\B$,
$\N^\infty, \trop$, and $\trop^r$ are naturally ordered.  Confusingly
(!!), the order relation on $\trop$ is the reverse one: $\infty$ is
the smallest, and $0$ is the largest element.  The order relation in
$\N_\bot$ and $\R_\bot$ is given by $\bot \preceq x$ for all $x$: they
are ordered pre-semirings but not naturally ordered.%
\footnote{Note that we define $\trop$ and $\trop^r$ over the natural
  numbers rather than the reals. The motivation for this slight
  deviation from the standard definition of these semirings will
  become clear in Section \ref{sec:verification}: the support of
  integer theories by the SMT-solver \zzz.}


{\bf $\bm S$-relations} An $\bm S$-relation $R$ is a function that
associates to each tuple $t \in D^k$ a value in the semiring,
$R[t] \in \bm S$.  In this context, $\bm S$ is called the {\em value
  space} of the relation $R$, while the domain $D$ of its attributes
is called the {\em key space}.  $\bm S$-relations were first
introduced\footnote{Under the name $K$-relations.} by Green et
al.~\cite{DBLP:conf/pods/GreenKT07} in order to model data provenance.
A $\B$-relation is a set, an $\N^\infty$-relation is a bag (with
possibly infinite multiplicities), an $\R_\bot$-relation is a tensor
(with possibly undefined entries).

{\bf Queries} Consider a relational schema $R_1, R_2, \ldots$ over a
\update{pre-}semiring $\bm S$.  A positive (relational algebra) {\em
  query} is a relational algebra expression using selections,
projections, joins, and unions (no difference operator in the positive
fragment).  \update{The most common definition of the relational
  algebra restricts the predicates used in selections to equality
  predicates, $x=y$.  In this paper we
  follow~\cite{DBLP:conf/pods/GreenKT07} and allow arbitrary
  predicates $p(x,y,\ldots)$ over the value space, including
  disequality $x\neq y$, inequality $x < y$, or any other interpreted
  predicate.}  Green~\cite{DBLP:conf/pods/GreenKT07} showed that
positive relational algebra extends naturally to an arbitrary
semi\-ring $\bm S$.  When $\bm S$ is the Boolean semiring, then this
coincides with the set semantics of relational algebra, and when
$\bm S$ is the semiring of natural numbers, then it coincides with bag
semantics.

{\bf Normal Forms} Alternatively, a query can be described using
rules, as follows.  A {\em sum-product} query is an expression
\begin{align}
  T(x_1, \ldots, x_k) &\cd \bigoplus_{x_{k+1}, \ldots,x_p \in D} A_1 \otimes  \cdots \otimes A_m
                        \label{eq:t:monomial}
\end{align}
where each $A_u$ is a {\em relational atom} of the form
$R_i(x_{t_{1_{i}}}, \ldots, x_{t_{k_i}})$, or some interpreted predicate such as
$x_i > 5x_j+3$.  The variables $x_1, \ldots, x_k$ are called {\em
  free variables}, or {\em head variables}, and the others are called
{\em bound variables}.
A {\em sum-sum-product} query has the form:
\begin{align}
  Q(x_1, \ldots, x_k) &\cd T_1(x_1, \ldots, x_k) \oplus \cdots \oplus T_q(x_1, \ldots, x_k)
\label{eq:sum:sum:product}
\end{align}
where $T_1, T_2, \ldots, T_q$ are sum-product expressions with the
same head variables $x_1, \ldots, x_k$.  When the semiring is
$\B, \N^\infty$ \update{and the interpreted predicates are restricted
  to equality predicates}, then these queries are (Unions of)
Conjunctive Queries (UCQs) under set semantics, or under bag
semantics; \update{when the semiring is $\R_\bot$, then the
  sum-products} are tensor expressions, sometimes called {\em Einsum
  expressions}~\cite{einsum:rocktaschel}.
Every positive relational algebra query $Q$ can be converted into a
sum-sum-product expression, which we call the {\em normal form} of
$Q$.

{\bf Datalog$^{\mathbf{o}}$} \update{Let $\bm S$ be an ordered
  pre-semiring.} A $\name$ program consists of a set of (possibly
recursive) sum-sum-product rules~\eqref{eq:sum:sum:product}
\update{over $\bm S$-relations.} \update{We allow two extensions to
  the expressions~\eqref{eq:t:monomial}
  and~\eqref{eq:sum:sum:product}}: the summation
in~\eqref{eq:t:monomial} may be restricted by some Boolean predicate,
and we also allow \update{an atom $A$ in~\eqref{eq:t:monomial} to be
  an {\em interpreted function}.  One important interpreted function
  is the cast operator $[-]_{\zero}^{\one} : \B \rightarrow \bm S$,
  which maps $0$ to $\zero$ and $1$ to $\one$ and therefore, for any
  predicate $P$, $[P]_{\zero}^{\one}$ is an atom in the pre-semiring
  $\bm S$. For example, $[x < y]_{\zero}^{\one}$ is $\zero \in \bm S$
  when $x \geq y$ and $\one \in \bm S$ when $x < y$; when $\zero,\one$
  are clear from the context, we drop them and write simply $[x<y]$.}
We treat interpreted functions in a similar way to negation in
standard Datalog, and require a program to be {\em stratified}, such
that the interpreted functions are applied only to EDBs or to IDBs
defined in earlier strata. \update{This implies that the ICO of that
  stratum is a monotone function in the IDBs defined by that stratum,
  and its semantics is defined as its least fixpoint.}  Abo Khamis et
al.~\cite{khamis21:_conver_datal_pre_semir} proved that any $\name$
program over the semirings discussed in this section (except for
$\N^\infty$ and $\trop^r$) converges in polynomial time in the size of
the input database.

\begin{example}
  Consider the body of the rule in Fig.~\ref{fig:cc}(b).
  The relations $L, CC$ are over the tropical semiring, while $E$ is
  over the Boolean semiring.  Formally, its body is a sum-sum-product
  expression, with a Boolean predicate:
  \begin{align*}
    & L[x] \oplus \bigoplus_y \setof{CC[y]}{E(x,y)}
  \end{align*}
  \update{Here the summation $\bigoplus_y$ is restricted to those
    values $y$ that satisfy the predicate $E(x,y)$.  Equivalently, we
    can} rephrase it as:
  \begin{align*}
    & L[x] \oplus \bigoplus_y \left(CC[y] \otimes [E(x,y)]_\infty^0\right)
  \end{align*}
  where $[-]_\infty^0$ is the cast operator from $\B$ to $\trop$; it
  maps $0,1$ to $\infty,0$ respectively.  Alternatively, suppose that
  we represent a label $v=L[x]$ using a standard, Boolean-valued
  relation $L(x,v)$, where $x$ is a key, and $v$ is the numerical
  value (label).  Then, instead of the atom $L[x]$ we would write
  $\bigoplus_v \setof{v}{L(x,v)}$, or
  $\bigoplus_v \left(v \otimes [L(x,v)]_\infty^0\right)$.  Here $v$ is
  considered to be an atom.
\end{example}


\section{The FGH-Rule}

\label{sec:fgh}

\update{
  In this section we introduce a simple rewrite rule that allows us to
  rewrite an iterative program to another, possibly more efficient
  program.  Then, we illustrate how this rule, when applied to $\name$
  programs, can express several known optimizations in the literature,
  as well as some new ones.
}

Consider an iterative program that repeatedly applies a function $F$
until some termination condition is satisfied, then applies a function $G$ that
returns the final answer $Y$:
\begin{align}
  & X \leftarrow  X_0 \nonumber \\
  & \texttt{loop } X \leftarrow F(X) \texttt{ end loop} \label{eq:f}\\
  & Y \leftarrow G(X) \nonumber
\end{align}
We call this an FG-program.  The FGH-rule (pronounced {\em FIG-rule})
provides a sufficient condition for the final answer $Y$ to be
computed by the alternative program, called the GH-program:
\begin{align}
  & Y \leftarrow  G(X_0) \nonumber \\
  & \texttt{loop } Y \leftarrow H(Y) \texttt{ end loop} \label{eq:h}
\end{align}

\begin{theorem}[The FGH-Rule]\label{th:fgh} If the following identity
  holds:
  \begin{align}
    G(F(X)) = & \ H(G(X)) \label{eq:fgh}
  \end{align}
  then the FG-program~\eqref{eq:f} is equivalent to the
  GH-program~\eqref{eq:h}.
\end{theorem}
\begin{proof}
  Let $X_0, X_1, X_2, \ldots$ denote the intermediate values of the
  FG-program, and $Y_0, Y_1, Y_2, \ldots$ those of the GH-program.
  By the FGH-rule,  the following diagram commutes, proving
  the claim:

\[\begin{tikzcd}
	{X_0} & {X_1} & {X_2} & \cdots & {X_n} \\
	{Y_0} & {Y_1} & {Y_2} & \cdots & {Y_n}
	\arrow["F", from=1-1, to=1-2]
	\arrow["G"', from=1-1, to=2-1]
	\arrow["H", from=2-1, to=2-2]
	\arrow["G"', from=1-2, to=2-2]
	\arrow["F", from=1-2, to=1-3]
	\arrow["H", from=2-2, to=2-3]
	\arrow["G"', from=1-3, to=2-3]
	\arrow["F", from=1-3, to=1-4]
	\arrow["H", from=2-3, to=2-4]
	\arrow["G"', from=1-5, to=2-5]
	\arrow["F", from=1-4, to=1-5]
	\arrow["H", from=2-4, to=2-5]
\end{tikzcd}\]
\end{proof}

In this paper we will apply the FGH-rule to optimize $\name$ programs.
In this context, $F$ is the ICO of the $\name$ program, $X$ is the
tuple of all its IDB predicates, and $Y$ are the answer-IDB
predicates.  We will also make the natural assumption that $G$ maps
the initial state $X_0$ of the IDBs of the program~\eqref{eq:f} to the
initial state $Y_0$ of~\eqref{eq:h}.  For example, if both programs
are traditional Datalog programs, then the initial state consists of
all IDBs being the empty set, which we denote, with some abuse, by
$X_0 = \emptyset$, even when $X$ consists of several mutually
recursive IDBs.  Similarly, $Y_0=\emptyset$. Typically, $G$ is a
conjunctive query, which maps $\emptyset$ to $\emptyset$, and in that
case the theorem implies that, if Eq.~\eqref{eq:fgh} holds, then the
following $\name$ programs $\Pi_1, \Pi_2$ return the same answer $Y$:
\begin{align}
\Pi_1: &&  X \cd & F(X) & \Pi_2: && Y \cd & H(Y) \nonumber \\
     &&  Y \cd & G(X) \label{eq:p1:p2}
\end{align}

More generally, however, the theorem does not care about the
termination condition of the FG-programs~\eqref{eq:f}. It only assumes
that the GH-program is executed the same number of iterations as the
FG-program.  However, it follows immediately that, if $F$ reaches a
fixpoint, then so does $H$:

\begin{corollary} \label{cor:number:iterations}
  If the FG-program reaches a fixpoint after $n$ steps (meaning:
  $X_n=X_{n+1}$) then the GH-program also reaches a fixpoint after $n$
  steps ($Y_n = Y_{n+1}$).  The converse fails: the GH-program may
  converge much faster than the FG-program.
\end{corollary}


In summary, the optimization proceeds as follows.  Given an FG-program
defined by the query expressions $F$ and $G$, find a new query
expression $H$ such that the identity $G\circ F = H \circ G$ holds,
then replace the FG-program with the GH-program.  We will describe
this process in detail in Sec.~\ref{sec:optimization}.  In the
remainder of this section we present several examples showing that the
FGH-rule can express several known optimizations, like magic set
rewriting, and new optimizations, like semantic optimizations using
global constraints.

\subsection{Simple Examples}

\label{subsec:simple:examples}

\begin{example}[Connected Components]\label{ex:fgh:cc} Consider the
computation of the connected components of a graph, which
is a well-known target of query optimization in the literature, 
see e.g.,~\cite{DBLP:conf/amw/ZanioloYIDSC18}.
The program is given in Fig.~\ref{fig:cc} (a), and its optimized version
  in Fig.~\ref{fig:cc} (b).  The three transformations $F,G,H$ are as follows:
{\footnotesize
  \begin{align*}
    F(TC) \defeq & TC' &&\mbox{ where}&  TC'(x,y) \defeq & [x=y] \vee \exists z(E(x,z) \wedge TC(z,y))\\
    G(TC) \defeq & CC  &&\mbox{ where}&  CC[x] \defeq & \min_y \setof{L[y]}{TC(x,y)}\\
    H(CC) \defeq & CC' &&\mbox{ where}&  CC'[x] \defeq & \min(L[x], \min_y \setof{CC[y]}{E(x,y)})
  \end{align*}
}
To check the FGH-rule, we compute $CC_1 \defeq G(F(TC)) = G(TC')$,
then compute
$CC_2 \defeq H(G(TC)) = H(CC)$, both shown in Fig.~\ref{fig:cc1:cc2}, and
observe that it becomes identical to $CC_1$ after renaming the
variables $y',y$ to $y,z$ respectively.
\end{example}

\begin{figure}
\fcolorbox{black}{light-gray}{\parbox{0.2\textwidth}{
\footnotesize
  \begin{align*}
    CC_1[x]  \defeq & \min_y \setof{L[y]}{TC'(x,y)}\\
    = & \min_y \setof{L[y]}{[x=y] \vee \exists z(E(x,z) \wedge TC(z,y))}\\
    = & \min(L[x], \min_y\setof{L[y]}{\exists z(E(x,z) \wedge TC(z,y))})\\
    = & \min(L[x], \min_{y,z}\setof{L[y]}{E(x,z) \wedge TC(z,y)})\\
    CC_2[x] \defeq & \min(L[x], \min_y \setof{CC[y]}{E(x,y)}) \\
    =  & \min(L[x], \min_y \setof{\min_{y'} \setof{L[y']}{TC(y,y')}}{E(x,y)})\\
    = & \min(L[x], \min_{y',y}\setof{L[y']}{E(x,y) \wedge TC(y,y')})
  \end{align*}
}}
  \caption{Computing $CC_1$ and $CC_2$ from Example~\ref{ex:fgh:cc}.}
  \label{fig:cc1:cc2}
\end{figure}

%
%

\begin{example}[$\prem$ Property] Zaniolo et
  al.~\cite{DBLP:journals/tplp/ZanioloYDSCI17} define the {\em
    Pre-mappability} rule ($\prem$), and prove that, under this
  rule, one Datalog program with ICO $F$ is equivalent to
  another program with a simpler ICO.
  The $\prem$ property is a
  restricted form of the FGH-rule, more precisely it asserts that the
  identity $G(F(X))=G(F(G(X)))$ holds.  In this case one can simply
  define $H$ as $H(X) = G(F(X))$, and the FGH-rule holds.  The $\prem$
  rule is more restricted than the FGH-rule, in two ways.  First, the
  types of the IDBs of the F-program and the H-program must be the
  same.  Second, the new query $H$ is uniquely
  defined, namely $H \defeq G\circ F$.  While this simplifies the
  optimizer significantly, it also limits the type of optimizations
  that are possible under $\prem$.
\end{example}

\begin{example}[Simple Magic]\label{ex:simple:magic} The simplest
  application of magic set
  optimization~\cite{DBLP:conf/pods/BancilhonMSU86,DBLP:conf/sigmod/MumickP94,DBLP:conf/sigmod/MumickFPR90}
  converts {\em transitive closure} to {\em reachability}.  More
  precisely, it rewrites this program:
  \begin{align}
\Pi_1: &&    TC(x,y) \cd & [x=y] \vee \exists z(TC(x,z) \wedge E(z,y))\nonumber\\
       &&    Q(y) \cd & TC(a,y) \label{eq:simple:magic:nonopt}
  \end{align}
  where $a$ is some constant, into this program:
  \begin{align}
\Pi_2: &&   Q(y) \cd & [y=a] \vee \exists z(Q(z) \wedge E(z,y)) \label{eq:simple:magic:opt}
  \end{align}
  This is a powerful optimization, because it reduces the run time from
  $O(n^2)$ to $O(n)$.  Several Datalog systems support some form of
  magic set optimizations.  We check
  that~\eqref{eq:simple:magic:nonopt} is equivalent
  to~\eqref{eq:simple:magic:opt} by verifying the FGH-rule.  The
  functions $F, G, H$ are shown in Fig.~\ref{fig:simple:magic}.  One
  can verify that $G(F(TC))=H(G(TC))$, for any relation $TC$.  Indeed,
  after converting both expressions to normal form, we obtain
  $G(F(TC))=H(G(TC)) = P$, where:
  \begin{align*}
    P(y) \defeq & [y=a] \vee \exists z(TC(a,z) \wedge E(z,y))
  \end{align*}
  We prove in the full version of this paper that, given a sideways information
  passing strategy (SIPS)~\cite{DBLP:journals/jlp/BeeriR91} every magic set
  optimization~\cite{DBLP:journals/jlp/BalbinPRM91} over a Datalog program can be
  proven correct using a sequence of applications of the FGH-rule.
\end{example}

\begin{figure}
\fcolorbox{black}{light-gray}{\parbox{0.45\textwidth}{
\footnotesize
  \begin{alignat*}{3}
    F(TC) \defeq & TC' &&\mbox{ where } & TC'(x,y) \defeq &[x=y] \vee \exists z(TC(x,z) \wedge E(z,y)) \\
    G(TC) \defeq & Q   &&\mbox{ where } & Q(y) \defeq & TC(a,y) \\
    H(Q) \defeq & Q'   &&\mbox{ where } & Q'(y) \defeq & [y=a] \vee \exists z(Q(z) \wedge E(z,y))
  \end{alignat*}
}}
\caption{Expressions $F,G,H$ in Example~\ref{ex:simple:magic}.}
  \label{fig:simple:magic}
\end{figure}

\begin{example}[Generalized Semi-Naive Evaluation] The na\"ive evaluation
    algorithm for (positive) Datalog re-discovers each fact from step $t$ again at
    steps $t+1, t+2, \ldots$ The {\em semi-naive algorithm} aims at avoiding
  this, by computing only the new facts.  We {\em generalize} the semi-naive
  evaluation from the Boolean semiring to any ordered pre-semiring
  $\bm S$, and prove its correctness using the FGH-rule.  We require
  $\bm S$ to be a complete distributive lattice and $\oplus$ to be
  idempotent, and define the ``minus'' operation as:
  $b \ominus a \defeq \bigwedge \setof{c}{b \preceq a \oplus c}$, then prove using
  the FGH-rule that the following two programs are equivalent:
  \begin{align*}
    \begin{array}{ll|lll}
      \Pi_1: \hspace{-8mm} & & \Pi_2: \hspace{-8mm} && \\
      & X_0 := \emptyset;               && Y_0 := \emptyset;  & \Delta_0 := F(\emptyset)  \ominus \emptyset;\ \ \ \   // = F(\emptyset) \\
      & \texttt{loop } X_t := F(X_{t-1}); && \texttt{loop } & Y_t := Y_{t-1} \oplus \Delta_{t-1};\\
      &                                 &&              & \Delta_t := F(Y_{t}) \ominus Y_{t};
    \end{array}
  \end{align*}
  To prove their equivalence, we define $G(X)\defeq (X, F(X)\ominus X)$,
  $H(X,\Delta) \defeq (X \oplus \Delta, F(X\oplus \Delta )\ominus (X\oplus \Delta))$, and then we prove
  that $G(F(X)) = H(G(X))$ by exploiting the fact that $\bm S$ is a
  complete distributive lattice. In practice, we compute the
  difference $\Delta_t = F(Y_{t})\ominus Y_{t} =
  F(Y_{t-1}\oplus \Delta_{t-1}) \ominus F(Y_{t-1})$
  using an efficient differential rule
  that computes
  $\delta F(Y_{t-1}, \Delta_{t-1}) =
  F(Y_{t-1}\oplus \Delta_{t-1}) \ominus F(Y_{t-1})$, where $\delta F$ is
  an {\em incremental update} query for $F$, i.e., it satisfies the
  identity $F(Y) \oplus \delta F(Y, \Delta) = F(Y\oplus \Delta)$.

  Thus, semi-naive query evaluation generalizes from standard Datalog
  over the Booleans to $\name$ over any complete distributive lattice
  with idempotent $\oplus$, and, moreover, is a special case of the
  FGH-rule.  However, the semi-naive program (more precisely, function
  $H$) is no longer monotone, while our synthesizer (described in
  Sec.~\ref{subsec:synthesis}) is currently restricted to infer only
  monotone functions $H$.  For that reason we do not synthesize the
  semi-naive algorithm; instead we apply it using pattern-matching as
  the last optimization step.
\end{example}

\subsection{Loop Invariants}

\label{subsec:loop:invariants}

More advanced uses of the FGH-rule require a loop-invariant,
$\phi(X)$.  By refining Theorem~\ref{th:fgh} with a loop invariant we
obtain the following corollary:

\begin{corollary}\label{cor:invariant}
  Let $\phi(X)$ be any predicate satisfying the
  following three conditions:
  \begin{align}
    & \phi(X_0) \label{eq:invariant:0}\\
    & \phi(X) \Rightarrow  \phi(F(X)) \label{eq:invariant:1}\\
    & \phi(X) \Rightarrow \left(G(F(X)) =  H(G(X))\right) \label{eq:invariant:2}
  \end{align}
  then the FG-program~\eqref{eq:f} is equivalent to the
  GH-program~\eqref{eq:h}.
\end{corollary}

To prove the corollary, we consider the restriction of the function
$F$ to values $X$ that satisfy $\phi$.
Conditions~\eqref{eq:invariant:0} and~\eqref{eq:invariant:1} state
that $\phi$ is a loop invariant for the FG-program~\eqref{eq:f}, while
condition~\eqref{eq:invariant:2} is the FGH-rule applied to the
restriction of $F$ to $\phi$.

\begin{example}[Beyond Magic] \label{ex:more:magic} By using
  loop-invariants, we can perform optimizations that are more powerful
  than standard magic set rewritings.
  For a simple illustration, consider the following program:
\begin{align}
\Pi_1: &&  TC(x,y) \cd & [x=y] \vee \exists z(E(x,z) \wedge TC(z,y))\label{eq:magic:nonopt}\\
      &&  Q(y) \cd & TC(a,y) \nonumber
\end{align}
which we want to optimize to:
\begin{align}
\Pi_2: &&  Q(y) \cd & [y=a] \vee \exists z(Q(z) \wedge E(z,y)) \label{eq:magic:opt}
\end{align}
Unlike the simple magic program in Example~\ref{ex:simple:magic}, here
rule~\eqref{eq:magic:nonopt} is right-recursive.  As shown in
\cite{DBLP:journals/jlp/BeeriR91}, the magic set optimization using
the standard sideways information passing
optimization~\cite{DBLP:books/aw/AbiteboulHV95} yields a program that
is more complicated than our program~\eqref{eq:magic:opt}.  Indeed,
consider a graph that is simply a directed path
$a_0 \rightarrow a_1 \rightarrow \dots \rightarrow a_n$ with
$a = a_0$.  Then, even with magic set optimization, the
right-recursive rule~\eqref{eq:magic:nonopt} needs to derive {\em
  quadratically many\/} facts of the form $T(a_i, a_j)$ for
$i \leq j$, whereas the optimized program~\eqref{eq:magic:opt} can be
evaluated in linear time.  Note also that the FGH-rule cannot be
applied directly to prove that the program~\eqref{eq:magic:nonopt} is
equivalent to~\eqref{eq:magic:opt}.  To see this, denote by
$P_1 \defeq G(F(TC))$ and $P_2 \defeq H(G(TC))$, and observe that
$P_1, P_2$ are defined as:
\begin{align*}
  P_1(y) \defeq & [y=a] \vee \exists z(E(a,z) \wedge TC(z,y)) \\
  P_2(y) \defeq & [y=a] \vee \exists z(TC(a,z) \wedge E(z,y))
\end{align*}
In general, $P_1 \neq P_2$.  The problem is that the FGH-rule requires
that $G(F(TC))=H(G(TC))$ for {\em every} input $TC$, not just the
transitive closure of $E$.  However, the FGH-rule {\em does} hold if
we restrict $TC$ to relations that satisfy the following
loop-invariant $\phi(TC)$:
\begin{align}
  \exists z_1 (E(x,z_1)\wedge TC(z_1,y))\Leftrightarrow \exists z_2 (TC(x,z_2)\wedge E(z_2,y)) \label{eq:invariant:tc}
\end{align}
If $TC$ satisfies this predicate, then it follows immediately that
$P_1=P_2$, allowing us to optimize the program~\eqref{eq:magic:nonopt}
to~\eqref{eq:magic:opt}.  It remains to prove that $\phi$ is indeed an
invariant for the function $F$.  The base case~\eqref{eq:invariant:0}
holds because both sides of~\eqref{eq:invariant:tc} are empty when
$TC=\emptyset$.  It remains to check
$\phi(TC)\Rightarrow \phi(F(TC))$.  Let us denote $TC'\defeq F(TC)$,
then we need to check that, if~\eqref{eq:invariant:tc} holds, then the
predicate $\Psi_1(x,y) \defeq \exists z_1 (E(x,z_1)\wedge TC'(z_1,y))$
is equivalent to the predicate
$\Psi_2(x,y) \defeq \exists z_2 (TC'(x,z_2)\wedge E(z_2,y))$.  We
expand both predicates in Fig.~\ref{fig:more:magic:1:2}, where we
renamed $z$ to $z_2$ in the last line of $\Psi_1$, and renamed $z$ to
$z_1$ in $\Psi_2$.  Their equivalence follows from the
assumption~\eqref{eq:invariant:tc}.
\end{example}

\begin{figure}
\fcolorbox{black}{light-gray}{\parbox{0.45\textwidth}{
\footnotesize
\begin{align*}
\Psi_1(x,y)
  \equiv & \exists z_1 \left(E(x,z_1)\wedge \left(
                [z_1=y] \vee \exists z(E(z_1,z) \wedge TC(z,y))
                \right)\right)\\
  \equiv & \exists z_1 \left(E(x,z_1)\wedge [z_1=y]  \vee E(x,z_1)\wedge \exists z(E(z_1,z) \wedge TC(z,y))\right)\\
  \equiv & E(x,y) \vee \exists z_1 (E(x,z_1) \wedge \exists z(E(z_1,z) \wedge TC(z,y)))\\
  \equiv & E(x,y) \vee \exists z_1 (E(x,z_1) \wedge \exists z_2(E(z_1,z_2) \wedge TC(z_2,y)))\\
  \Psi_2(x,y) \equiv &  \exists z_2  \left(
                \left([x=z_2] \vee \exists z(E(x,z) \wedge TC(z,z_2))\right)
                \wedge E(z_2,y)
                \right) \\
  \equiv & E(x,y) \vee \exists z, z_2 (E(x,z) \wedge TC(z,z_2) \wedge E(z_2,y))\\
  \equiv & E(x,y) \vee \exists z (E(x,z) \wedge \exists z_2(TC(z,z_2) \wedge E(z_2,y)))\\
  \equiv & E(x,y) \vee \exists z_1 (E(x,z_1) \wedge \exists z_2(TC(z_1,z_2) \wedge E(z_2,y)))
\end{align*}
}}
\caption{Predicates $\Psi_1$ and $\Psi_2$ from Example~\ref{ex:more:magic}.}
  \label{fig:more:magic:1:2}
\end{figure}

%
%

\subsection{Semantic Optimization Under Constraints}

\label{subsec:constraints}

{\em Semantic optimization} refers to optimization rules that hold
when the database satisfies certain
constraints~\cite{DBLP:journals/debu/RamakrishnanS94}.  For example,
most database systems today can optimize key/foreign-key joins by
simply removing the join when the table containing the key is not
used anywhere else in the query.

A priori knowledge on the structure of the underlying data
may often provide additional potential for optimization. For instance,
in \cite{DBLP:conf/pods/BancilhonMSU86}, the counting and reverse counting
methods are presented to further optimize the same-generation program if
it is known that the underlying graph is acyclic.
We present a principled way of exploiting such a priori knowledge.
As we show here, recursive queries
have the potential to use {\em global} constraints on the data during
semantic optimization; for example, the query optimizer may exploit
the fact that the graph is a tree, or the graph is connected.

Let $\Gamma$ denote a set of constraints on the EDBs.  Then, the
FGH-rule~\eqref{eq:fgh} needs to be be checked only for EDBs that
satisfy $\Gamma$.  We illustrate this with an example:

\begin{example}[Semantic Optimization] \label{ex:fgh:constraints}
  Consider a hierarchy of subparts consisting of two relations:
  $\texttt{SubPart}(x,y)$ indicates that $y$ is a subpart of $x$, and
  $\texttt{Cost}[x] \in \N$ represents the cost of the part $x$.  We
  want to compute, for each $x$, the total cost $Q[x]$ of all its
  subparts, sub-subparts, etc.  Since the hierarchy can, in
  general, be a DAG, we first need to compute the transitive closure, before
  summing up the costs of all
  subparts,  sub-subparts, etc:
  \begin{align}
\Pi_1: &&    S(x,y) \cd & [x=y] \vee \exists z\left(S(x,z) \wedge \texttt{SubPart}(z,y)\right)\label{eq:cost:nonopt}\\
       &&    Q[x] \cd & \sum_y \setof{\texttt{Cost}[y]}{S(x,y)} \nonumber
  \end{align}
  The first rule, defining the $S$ predicate, is over the $\B$
  semiring, while the second rule, defining $Q$, is over the $\N_\bot$
  semiring.  Consider now the case when our subpart hierarchy is a
  tree.  Then, we can compute the total cost much more efficiently,
  using the following program:
  \begin{align}
\Pi_2: &&    Q[x] \cd & \texttt{Cost}[x] + \sum_z \setof{Q[z]}{\texttt{SubPart}(x,z)} \label{eq:cost:opt}
  \end{align}
  Optimizing the program~\eqref{eq:cost:nonopt} to~\eqref{eq:cost:opt}
  is an instance of {\em semantic optimization}, since this only holds
  if the database instance is a tree.  We do this in three steps. We
  define the constraint $\Gamma$ stating that the data is a tree;
  using $\Gamma$ we infer a loop-invariant $\Phi$ of the program
  $\Pi_1$; using $\Gamma$ and $\Phi$ we prove the FGH-rule, concluding
  that $\Pi_1$ is equivalent to $\Pi_2$.

%
%
  The constraint $\Gamma$ is the conjunction of the following statements:
  \begin{align}
    & \forall x_1, x_2, y \left(\texttt{SubPart}(x_1,y) \wedge \texttt{SubPart}(x_2,y) \Rightarrow x_1=x_2\right) \label{eq:gamma:key}\\
    & \forall x,y\left(\texttt{SubPart}(x,y) \Rightarrow T(x,y)\right) \label{eq:gamma:et}\\
    & \forall x,y,z\left(T(x,z) \wedge \texttt{SubPart}(z,y) \Rightarrow T(x,y)\right)\label{eq:gamma:tc}\\
    & \forall x,y\left(T(x,y) \Rightarrow x\neq y\right) \label{eq:gamma:neq}
  \end{align}
  The first asserts that $y$ is a key in $\texttt{SubPart}(x,y)$.  The
  last three are an Existential Second Order Logic (ESO) statement: they assert that
  there exists some relation $T(x,y)$ that contains
  $\texttt{SubPart}$, is transitively closed, and irreflexive.
  Next, we infer the following loop-invariant of the program $\Pi_1$:
%
  \begin{align}
\Phi: S(x,y) \Rightarrow [x=y] \vee T(x,y) \label{eq:invariant:tree}
  \end{align}
  Finally, we check the FGH-rule, under the assumptions
  $\Gamma, \Phi$.  Denote by $P_1 \defeq G(F(S))$ and
  $P_2 \defeq H(G(S))$.  To prove $P_1=P_2$ we simplify $P_1$ using
  the assumptions $\Gamma, \Phi$, as shown in
  Fig.~\ref{fig:fgh:constraints}. We explain each step.  Line 2-3 are
  inclusion/exclusion.  \update{Line 4 uses the fact that} the term on
  line 3 is $=0$, because the loop invariant implies:
{\footnotesize
  \begin{align*}
     S(x,z) \wedge \texttt{SubPart}(z,y)& \Rightarrow  ([x=z] \vee T(x,z)) \wedge  \texttt{SubPart}(z,y) & \mbox{\update{by~\eqref{eq:invariant:tree}}}\\
&  \equiv \texttt{SubPart}(x,y) \vee (T(x,z) \wedge \texttt{SubPart}(z,y))\\
& \Rightarrow T(x,y) \vee T(x,y) & \mbox{\update{by~\eqref{eq:gamma:tc}}}\\
& \equiv  T(x,y) \\
& \Rightarrow  x\neq y & \mbox{\update{by~\eqref{eq:gamma:neq}}}
  \end{align*}
}
\update{Line 5} follows from the fact that $y$ is a key in
$\texttt{SubPart}(z,y)$.  A direct calculation of $P_2 = H(G(S))$
results in the same expression as line 5 of
Fig.~\ref{fig:fgh:constraints}, proving that $P_1=P_2$.
\end{example}

\begin{figure}
\fcolorbox{black}{light-gray}{\parbox{0.45\textwidth}{
\footnotesize
  \begin{align*}
   P_1[x] = & \sum_y \setof{\texttt{Cost}[y]}{[x=y] \vee \exists z\left(S(x,z) \wedge \texttt{SubPart}(z,y)\right)}\\
    = & \texttt{Cost}[x]+\sum_y \setof{\texttt{Cost}[y]}{\exists z\left(S(x,z) \wedge \texttt{SubPart}(z,y)\right)}\\
      & \ \ \ - \sum_y\setof{\texttt{Cost}[y]}{[x=y] \wedge \exists z\left(S(x,z) \wedge \texttt{SubPart}(z,y)\right)}\\
    = & \texttt{Cost}[x]+\sum_y \setof{\texttt{Cost}[y]}{\exists z\left(S(x,z) \wedge \texttt{SubPart}(z,y)\right)}\\
    = & \texttt{Cost}[x]+\sum_y \sum_z \setof{\texttt{Cost}[y]}{\left(S(x,z) \wedge \texttt{SubPart}(z,y)\right)}
  \end{align*}
}}
  \caption{Transformation of $P_1 \defeq G(F(S))$ in Example~\ref{ex:fgh:constraints}.}
  \label{fig:fgh:constraints}
\end{figure}

%
%
%

\section{Architecture of FGH-Optimization}\label{sec:optimization}

\begin{figure}
  \null\mbox{\parbox{\linewidth}{\hspace{-5mm}\includegraphics[width=1.3\linewidth]{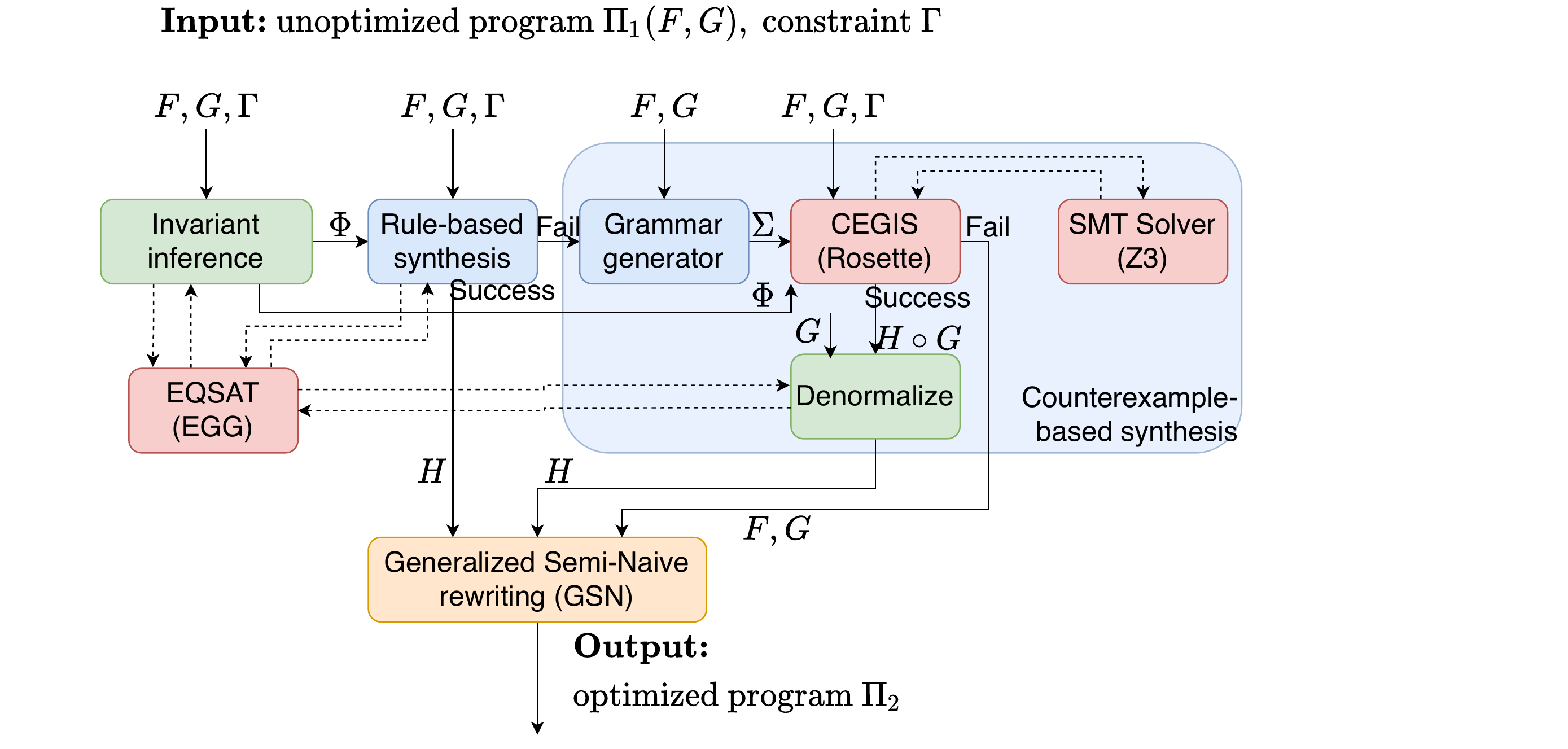}}}
  \caption{The architecture of the FGH-optimizer.  The input is the
    unoptimized program $\Pi_1$, consisting of the functions $F, G$ and
    the database constraint $\Gamma$. The output consists of the
    optimized program $\Pi_2$, see Eq.~\eqref{eq:p1:p2}.  Blue boxes are
    described in Section~\ref{subsec:synthesis} and the green boxes in
    Section~\ref{sec:semantic-opt}.  The yellow box
    (generalized semi-naive optimization) is described in
    Section~\ref{subsec:simple:examples}.  The red boxes represent
    three state-of-the-art systems: Rosette is a \cegis\ system~\cite{DBLP:conf/asplos/Solar-LezamaTBSS06,DBLP:conf/tacas/TorlakJ07,DBLP:conf/oopsla/TorlakB13}, \zzz\ is
    an SMT solver~\cite{10.1007/978-3-540-78800-3_24}, and EGG is an \eqsat\ system~\cite{DBLP:journals/pacmpl/WillseyNWFTP21}.}\label{fig:arch}
\end{figure}

In the rest of the paper we describe our synthesis-based
FGH-optimizer, whose architecture is shown in Fig.~\ref{fig:arch}.  We
optimize one stratum at a time.  We denote by $\Pi_1$ one stratum of
the input program, denote by $X$ its recursive IDBs, by $Y$ its output
IDBs, and by $F,G$ the ICO and the output operator respectively; see
Eq.~\eqref{eq:p1:p2}.  The optimizer also takes as input a database
constraint, $\Gamma$.  The optimizer starts by inferring the loop
invariant $\Phi$; this is discussed in Sec.~\ref{sec:semantic-opt}.
Next, the optimizer needs to find $H$ such that
$\Gamma \wedge \Phi \models \left(G(F(X)) = H(G(X))\right)$.  To
reduce clutter we will often abbreviate this to
$\Gamma \models \left(G(F(X)) = H(G(X))\right)$, assuming that
$\Gamma$ incorporates $\Phi$.  The optimizer makes two attempts at
synthesizing $H$: it first tries using a simpler rule-based
synthesizer, and, if that fails, then it tries the state-of-the-art
Counterexample-Guided Inductive Synthesis (\cegis).  This is described
in Sec.~\ref{subsec:synthesis}.  Finally, $H$ (or the original program
if the FGH-optimization failed) is further transformed using
generalized semi-naive optimization, as we already described in
Sec.~\ref{subsec:simple:examples}.  Notice that stratification ensures
that no interpreted functions are applied to the IDBs $X$; they can
still be applied to the EDBs, or occur in predicates.

The FGH-optimization is an instance of {\em query rewriting using
  views}~\cite{DBLP:journals/vldb/Halevy01,DBLP:conf/sigmod/GoldsteinL01}.
Denoting by $Q \defeq G(F(X))$ and $V \defeq G(X)$, one has to rewrite
the query $Q$ using the view(s) $V$, in other words $Q = H(V)$.  This
is a {\em total} rewriting, in the sense that $H$ is no longer allowed
to refer to the IDBs $X$.  This problem is NP-complete for UCQs with
set semantics~\cite{DBLP:conf/pods/LevyMSS95}, in NP for UCQs with bag
semantics\footnote{This follows from the fact that, under bag semantics, two UCQ queries are
  equivalent iff they are
  isomorphic.~\cite{DBLP:conf/icdt/Green09,DBLP:journals/pvldb/WangHSHL20}.},
and undecidable for realistic SQL queries that include aggregates and
arithmetic~\cite{DBLP:conf/sigmod/GoldsteinL01}.
Systems that support query rewriting using views are rule-based, and
apply a set of hand crafted, predefined patterns; our first attempt to
synthesize $H$ is also rule-based.  Such synthesizers usually cannot take advantage
of database constraints, but we will show in
Sec.~\ref{sec:semantic-opt} how to exploit the constraint $\Gamma$ in
the rule-based synthesizer.  However, rule-based rewriting explores a
limited space, which is insufficient for many FGH-optimizations.
In a seminal paper~\cite{DBLP:conf/asplos/Solar-LezamaTBSS06}
Solar{-}Lezama proposed an alternative to rule-based transformation,
called {\em Counterexample-Guided Inductive Synthesis}, \cegis: the
synthesizer produces potentially incorrect candidates, and an SMT
solver verifies their correctness.  In the FGH-optimizer we use a
program synthesizer, Rosette~\cite{DBLP:conf/oopsla/TorlakB13}, to synthesize $H$.

At a conceptual level, program synthesis has two abstract steps: {\em
  generate} $H$, and {\em verify} $G(F(X))=H(G(X))$.  While the
verifier is not used explicitly, it is used implicitly in the
synthesizer, and we describe it in Sec.~\ref{sec:verification}.  Then we
describe the synthesizer in Sec.~\ref{subsec:synthesis}.




\section{Verification}
\label{sec:verification}

\update{
  We introduced the FGH-rule in Sec.~\ref{sec:fgh} and showed several
  examples.  In order to apply the rule, one needs to check the
  identity~\eqref{eq:fgh}, $F(G(X))=G(H(X))$. In this section we
  describe how we verify this identity.
}
This step is implicit in both boxes {\em Rule-based Synthesis} and
\cegis\ in Fig~\ref{fig:arch}.
\update{
  The identity can be checked in one of two ways: by applying a
  predefined set of identity rules (as currently done by most query
  optimizers), or by using an SMT solver.
}

\subsection{Rule-based  Test}\label{sec:axiomatic-rewrite}

%
Let $P_{1} = G(F(X))$, $P_{2} = H(G(X))$.  To check $P_1=P_2$, the
{\em rule-based test} first normalizes both expressions into a
sum-sum-product expression (Eq.~\eqref{eq:sum:sum:product})
via the semiring axioms, then checks if the expressions are isomorphic:
if yes, then $P_{1} = P_{2}$, otherwise we assume $P_{1}\neq P_{2}$.
The treatment of  a constraint $\Gamma$ will be discussed in Sec.~\ref{sec:semantic-opt}.
This test can be visualized as follows:
\begin{align}
  P_{1} \xrightarrow[]{\text{axioms}} \text{normalize}(P_{1})
  \simeq
  \text{normalize}(P_{2}) \xleftarrow[]{\text{axioms}} P_{2}
\label{eq:equality:isomorphism}
\end{align}
where $\simeq$ denotes isomorphism.  The Rule-based test is sound.
When both $P_1, P_2$ are over the
$\N^\infty$ semiring and have no interpreted functions then it
is also
complete~\cite{DBLP:conf/icdt/Green09,DBLP:journals/pvldb/WangHSHL20}.
This simple test motivates the need for a complete set of axioms that
allows any semiring expression to be normalized.  The axioms include
standard semiring axioms, and axioms about summations and free
variables \texttt{fv}.  For example, in order to
prove $CC_1 = CC_2$ in Example~\ref{ex:fgh:cc} (with
semiring notation in Figure~\ref{fig:cc1:semiring}) one needs
all three axioms below:
\begin{align}
  \bigoplus_x \bigoplus_y \left(\cdots\right) = \ &  \bigoplus_{x,y}\left( \cdots \right)  \label{eq:axiom:plus:plus} \\
  A \otimes \bigoplus_{x} B = \ & \bigoplus_{x} A \otimes B \text{ when } x \not \in \texttt{fv}(A) \label{eq:axiom:pullmul}\\
  \bigoplus_x \left(A(x) \otimes [x=y]\right) = \ & A(y)  \label{eq:axiom:plus:equal}
\end{align}
\begin{figure}
\fcolorbox{black}{light-gray}{\parbox{0.2\textwidth}{
\footnotesize
  \begin{align*}
    CC_{1}[x] & = \bigoplus_{y} L[y] \otimes \big([x=y]_{\infty}^{0} \oplus \bigoplus_{z} [E(x,z)]_{\infty}^{0} \otimes [TC(z,y)]_{\infty}^{0}\big) \\
    CC_{2}[x] & = L[x] \oplus \bigoplus_{y} \big( \bigoplus_{y'} L[y'] \otimes [TC(y,y')]_{\infty}^{0}\big) \otimes [E(x,y)]_{\infty}^{0}
  \end{align*}
}}
  \caption{$CC_{1}$ and $CC_{2}$ in semiring notation; their normal
    forms are isomorphic.}
  \label{fig:cc1:semiring}
\end{figure}

\subsection{SMT Test} \label{subsubsec:smt:test} When the
expressions $P_1, P_2$ are over a semiring other than $\N^\infty$, or
they contain interpreted functions, then the rule-based test is
insufficient and we use an SMT solver for our verifier.
We still normalize the
expressions using our axioms, because today's solvers cannot reason
about bound/free variables (as needed in
axioms~\eqref{eq:axiom:plus:plus}-\eqref{eq:axiom:plus:equal}).  The
{\em SMT test} is captured by the following figure:
%
%
\begin{align}
  P_{1} \xrightarrow[]{\text{axioms}} \text{normalize}(P_{1})
  \xleftrightarrow{\text{SMT}} &
  \text{normalize}(P_{2}) \xleftarrow[]{\text{axioms}} P_{2} \label{eq:equality:smt}
\end{align}

\begin{example}[APSP100]\label{ex:sssp100}
  Consider a labeled graph $E$ where $E[x,y]$ represents the cost of
  the edge $x,y$.  The following query over $\trop$
  computes the all-pairs shortest
  path up to length of 100:
{\footnotesize
  \begin{align}
    D[x,y] \cd & \texttt{if } x=y \texttt{ then } 0 \texttt{ else } \min_{z} \left(D[x,z] +E[z,y]\right) \nonumber \\
    Q[x,y] \cd & \min(D[x,y],100)  \label{eq:apss:opt}
  \end{align}
}
%
%
%
  The program is inefficient because it first computes the full path
  length, only to cap it later to 100.  By using the FGH-rule we get:
{\footnotesize
  \begin{align}
    Q[x,y] \cd & \texttt{if } x=y \texttt{ then } 0 \texttt{ else } \min\left(\min_{z} \left(Q[x,z]+E[z,y]\right),100\right) \label{eq:apss:opt:opt}
  \end{align}
}
We show how to verify
that~\eqref{eq:apss:opt:opt} is equivalent to~\eqref{eq:apss:opt}.
Denote by $P_1 \defeq G(F(D))$ and $P_2 \defeq H(G(D))$ (where $F,G,H$
are the obvious functions in the two programs defining $Q$).  After we
de-sugar, convert to semiring expressions, and normalize,
they become:
{\footnotesize
\begin{align*}
  P_{1}[x,y] & =  \left(0 \otimes \totrop{x=y}\right) \oplus \left(\bigoplus_{z} D[x,z] \otimes E[z,y]\right)\oplus 100 \\
  P_{2}[x,y] 
             & = \left(0 \otimes \totrop{x=y}\right) \oplus \left(\bigoplus_{z} D[x,z] \otimes E[z,y]\right)  \oplus \left(100 \otimes\bigoplus_{z} E[z,y]\right)\oplus 100
\end{align*}
}
In the normalized expressions we push the summations past the joins,
i.e., we apply rule~\eqref{eq:axiom:pullmul} from right to left, thus
we write $100 \otimes \bigoplus (\cdots)$ instead of
$\bigoplus (100 \otimes \cdots)$: we give the rationale below.  At
this point, the normalized $P_{1}$ and $P_{2}$ are not isomorphic, yet
they are equivalent if they are interpreted in $\trop$.
We explain below in detail how the solver can check that.
In this particular semiring,
the
identity
$
  100 = \left(100 \otimes \bigoplus_{z} E[z,y]\right)  \oplus 100
$
holds since it becomes
$
100 = \min(100 + \min_{z} E[z,y],100)
$
with $E[z,y] \geq 0$,
once we replace the uninterpreted operators
$\oplus, \otimes$ with $\min, +$.
%
\end{example}

{\bf Implementation} We describe how we implemented the SMT test
$\Gamma \models P_1 = P_2$ using a solver, now also taking
the database constraint $\Gamma$ into account, where
$P_1, P_2$ are the expressions $G\circ F$
and $H \circ G$. We used the \zzz\ solver~\cite{10.1007/978-3-540-78800-3_24}, but our
discussion applies to other solvers as well. We need to normalize
$P_1, P_2$ before using the solver, because solvers require all axioms
to be expressed in First Order Logic. They cannot encode the
axioms~\eqref{eq:axiom:plus:plus}-\eqref{eq:axiom:plus:equal}, because
they are referring to free variables, which is a
meta-logical condition not expressible in First Order Logic.
Once normalized, we encode the equality as a first-order logic
formula, and assert its negation, asking the solver to check if
$\Gamma \wedge \left(P_1 \neq P_2\right)$ is satisfiable.  The solver
returns $\textsf{UNSAT}$, a counterexample, or
$\textsf{UNKNOWN}$.  $\textsf{UNSAT}$ means the
identity holds.  When it returns a counterexample, then the identity
fails, and the counterexample is given as input to the synthesizer
(Sec.~\ref{subsec:synthesis}).  $\textsf{UNKNOWN}$ means that it could
neither prove nor disprove the equivalence and we
assume $P_1 \neq P_2$.  For the theory of reals with $+, *$, despite its
decidability,  \zzz\ often timed out in our experiments.
We therefore used the theory of integers, and
\zzz\ never timed out or returned $\textsf{UNKNOWN}$ in our experiments.

We encode every $\bm S$-relation
$R(x_1, \dots, x_n)$ as an uninterpreted function
$R : \N \times \cdots \times \N \rightarrow \bm S$, where $\bm S$ is
the {\em interpreted} semiring, i.e., $\B$, $\trop$, $\N^\infty$, etc.
We represent natural numbers as
integers with nonnegativity assertions, and represent
the sets $\N^\infty, \N_\bot, \R_\bot$ as union types.
Operators supported by the solver, like $+, *, \min, -$,
are entered unchanged; we treat other operators
as uninterpreted functions.  Unbounded aggregation, like
$\bigoplus_{x} e(x)$, poses a challenge:  there is no
such  operation in any SMT theory.
Here we use the fact that $P_1$ and $P_2$ are normalized
sum-sum-product~expressions:
\begin{align*}
P_1 = & \left(\bigoplus_{x_1} e_1\right) \oplus \left(\bigoplus_{x_2}  e_2\right) \oplus \cdots
&
P_2 = & \left(\bigoplus_{x_1'} e_1'\right) \oplus \left(\bigoplus_{x_2'} e_2'\right) \oplus \cdots
\end{align*}
Assume first that each $x_i$ is a single variable. We
ensure that all the variables $x_1, x_2, \ldots$ in $P_1$ are
distinct, by renaming them if necessary.  Next, we replace each
expression $\bigoplus_{x_i} e_i$ with $u(x_i,e_i)$ where $u$ is an
uninterpreted function.  Finally, we ask the solver to check
\begin{align}
\Gamma \models \left(u(x_1,e_1) \oplus u(x_2,e_2) \oplus \cdots =  u(x_1',e_1') \oplus u(x_2',e_2') \oplus \cdots\right)
\nonumber
\end{align}
This procedure is sound, because if the identity $u(x,e) = u(x',e')$
holds, then $x=x'$ (they are the same variable) and $e=e'$, which
means that $\bigoplus_x e = \bigoplus_{x'} e'$.  Moreover,
when synthesizing $P_2$, we will
ensure that the generator includes the variables $x_1, x_2, \ldots$
present in $P_1$ to achieve a limited form of completeness,
see
Sec.~\ref{subsec:synthesis}. Finally, if a
summation is over multiple variables, we simply nest the uninterpreted
function, i.e., write $\bigoplus_{x,y} e$ as $u(x,u(y,e))$.

\begin{example}  We now finish Example~\ref{ex:sssp100}.
After introducing the
  uninterpreted functions described above, we obtain:
  \begin{align*}
    P_1 = & \min(0+w(x,y), u(z,D[x,z]+E[z,y]), 100) \\
    P_2 = & \min(0+w(x,y), u(z,D[x,z]+E[z,y]), 100+u(z,E[z,y]), 100)
  \end{align*}
  where $w(x,y)$ is an uninterpreted function representing
  $[x=y]_\infty^0$, and $u$ is our uninterpreted function encoding
  summation.  The solver proves that the two expressions are equal,
  given that $w \geq 0$ and $u \geq 0$.  Notice that it was critical
  to factorize the term 100: had we not done that, then the expression
  $100+u(z,E[z,y])$ would be $u(z,100 + E[z,y])$ and the identity
  $P_1=P_2$ no longer holds.
\end{example}

\required{
  {\bf Discussion} Readers unfamiliar with First Order Logic may be
  puzzled by our statement that the identity $u(x,e)=u(x',e')$ holds
  iff $x=x'$ and $e = e'$.  In order to explain this, it helps to
  first review the basic definitions of validity and satisfiability in
  logic.  A statement is ``valid'' if it is true for all
  interpretations of its uninterpreted symbols.  For example, the
  equality $f(x)+y=y+f(x)$ is valid over integers, because it holds
  for all function $f$ and all values of $x$ and $y$.  A statement is
  ``satisfiable'' if there exists interpretations of its uninterpreted
  symbols that make the statement true.  A statement is valid iff its
  negation is not satisfiable.  In our case, the statement
  $u(x,e)=u(x',e')$ is valid if the equality is true for all possible
  interpretations of $u, x, x'$.  For example, suppose we asked the
  solver to check whether $u(x,2(x+1)) = u(y,2y+2)$ is valid.  To
  answer this question, we negate the statement and ask the \zzz\
  solver whether the negation is satisfiable:
  $u(x,2(x+1))\neq u(y,2y+2)$.  One can easily satisfy this with pen
  and paper, e.g.,  $x=1, y=2, u(a,b)=a+b$, then $u(x,2(x+1))=5$,
  $u(y,2y+2)=8$.  \zzz\ also answers ``yes'', and provides the
  following example for the inequality\footnote{Please refer to the
    documentation of \zzz\ for how models for uninterpreted functions
    are constructed.}:
  \[x=0, y=38, u(a,b)= \text{if } a=38\wedge b=78 \text{ then } 6 \text{ else } 4\]
  Therefore, the identity $u(x,2(x+1)) = u(y,2y+2)$ is not valid.  In
  contrast, suppose we asked the solver whether
  $u(x,2(x+1)) = u(x,2x+2)$ is valid.  Its negation is
  $u(x,2(x+1)) \neq u(x,2x+2)$, and \zzz\ returns $\textsf{UNSAT}$,
  which means that the identity is valid.  In general, the identity
  $u(x,e)=u(x',e')$ is valid iff $x=x'$ and $e=e'$.
}

\section{Synthesis}

\label{subsec:synthesis}


\update{
We have seen in Sec.~\ref{sec:verification} how to use an SMT solver
to check the identity $G(F(X))=H(G(X))$.
}
We are now ready to discuss the core of the FGH-optimizer: given only
the query expressions $F, G$, find an expression $H$ such that the
identity $G(F(X)) = H(G(X))$ holds; recall that we denote these
expressions by $P_1, P_2$. As for verification, this can be done by
using only rewriting, or using program synthesis with an SMT solver.
We are also given a database constraint $\Gamma$, and we assume that
we have already added to it the loop invariant $\Phi$.

\subsection{Rule-based Synthesis}

\label{subsec:rule:based:synthesis}

The optimizer first attempts to synthesize $H$ using rule-based
rewriting.  This process is akin to our initial verifier that relies
only on normalization and isomorphism checking.
\begin{align}
  P_{1} \xrightarrow[]{\text{axioms}} \text{normalize}(P_{1}) \xrightarrow[]{\text{axioms}} P_{2}
  \label{eq:synthesis:isomorphism}
\end{align}
There is no obvious way to ``denormalize'' an expression, since many
expressions can share the same normal form.  We used for this purpose
an equality saturation system (\eqsat), which we also used for multiple tasks
of the FGH-optimizer, see Fig~\ref{fig:arch}.  We describe \eqsat\ in Sec.~\ref{sec:semantic-opt}.




\subsection{Counterexample-based Synthesis}

The rule-based synthesis~\eqref{eq:synthesis:isomorphism} explores
only correct rewritings $P_2$, but its space is limited by the
hand-written axioms. The alternative approach, pioneered in the
programming language
community~\cite{DBLP:conf/asplos/Solar-LezamaTBSS06}, is to synthesize
candidate programs $P_2$ from a much larger space, then using an SMT
solver to verify their correctness.  This technique, called
Counterexample-Guided Inductive Synthesis, or \cegis, can find
rewritings $P_2$ even in the presence of interpreted functions,
because it exploits the theory of the underlying domain.  As a first
attempt it can be described as follows (we will revise it below):
\begin{align}
  P_{1} \xrightarrow[]{\text{axioms}} \text{normalize}(P_{1}) \xrightarrow[]{\cegis} P_{2}
  \label{eq:synthesis:cegis}
\end{align}

\subsubsection{Brief Overview of \cegis} We give a brief overview of
the \cegis\ system,
Rosette~\cite{DBLP:conf/tacas/TorlakJ07,DBLP:conf/oopsla/TorlakB13},
that we used in our optimizer.  Understanding its working is
important in order to optimize its usage for FGH-optimization.  The
input to Rosette consists of a {\em specification} and a {\em
  grammar}, and the goal is to synthesize a program defined by the
grammar and that satisfies the specification.
The main loop is implemented with a pair of {\em dueling} SMT-solvers,
the {\em generator} and the {\em checker}.  In our setting, the inputs
are the query $P_1$, the database constraint $\Gamma$, and a small
grammar $\Sigma$ (described below).  \update{The specification is
  $\Gamma \models (P_1 = P_2)$, where $P_2$ is defined by the grammar
  $\Sigma$.}  The generator generates syntactically correct programs
$P_2$, and the verifier checks $\Gamma \models (P_1 = P_2)$.  In the
most naive attempt, the generator could blindly generate candidates
$P_2, P_2', P_2'', \ldots$, until one is found that the verifier
accepts.  This is hopelessly inefficient.  The first optimization in
\cegis\ is that the verifier returns a small counterexample database
instance $D$ for each unsuccessful candidate $P_2$, i.e.,
$P_1(D)\neq P_2(D)$.  When considering a new candidate $P_2$, the
generator checks that $P_1(D_i) = P_2(D_i)$ holds for all previous
counterexamples $D_1, D_2, \ldots$, by simply evaluating the queries
$P_1,P_2$ on the small instance $D_i$.  This significantly reduces the
search space of the generator.

\cegis\
applies a second optimization, where it uses the SMT solver itself to
generate the next candidate $P_2$, as follows.  It requires a fixed
recursion depth for the grammar $\Sigma$; in other words we can assume
w.l.o.g. that $\Sigma$ is non-recursive.  Then it associates a
symbolic Boolean variable $b_1, b_2, \ldots$ to each choice of the
grammar.  The grammar $\Sigma$ can be viewed now as a BDD (binary
decision diagram) where each node is labeled by a choice variable
$b_j$, and each leaf by a completely specified program $P_2$.  The
search space of the generator is now completely defined by the choice
variables $b_j$, and Rosette uses the SMT solver to generate values
for these Boolean variables such that the corresponding program $P_2$
satisfies $P_1(D_i) = P_2(D_i)$, for all counterexample instances
$D_i$.  This significantly speeds up the choice of the next candidate
$P_2$.

\subsubsection{Using Rosette} To use Rosette, we need to define the
specification and the grammar. A first attempt is to simply define
some grammar for $H$, with the specification
$\Gamma \models \left(G(F(X)) = H(G(X))\right)$.  This does not work,
since Rosette uses the SMT solver to check the identity: as explained
in Sec.~\ref{subsubsec:smt:test}, modern SMT solvers have limitations
that require us to first normalize $G(F(X))$ and $H(G(X))$ before
checking their equivalence.
Even if we modify Rosette
to normalize $H(G(X))$ during verification,
there is still no obvious way to incorporate normalization
into the program generator driven by the SMT solver.
Instead,
we define a grammar $\Sigma$ for $\text{normalize}(H(G(X)))$ rather
than for $H$, and then specify:
\[
\Gamma \models \text{normalize}(G(F(X))) = \text{normalize}(H(G(X)))
\]
Then, we denormalize the result
returned by Rosette, in order to extract $H$, using the {\em
  denormalization} module in Fig.~\ref{fig:arch}, described in
Sec.~\ref{sec:semantic-opt}.  In summary, our \cegis-approach for
FGH-optimization can be visualized as follows:
\begin{align}
  P_{1} \xrightarrow[]{\text{axioms}} \text{normalize}(P_{1}) \xrightarrow[]{\cegis} \text{normalize}(P_{2}) \xrightarrow[]{\text{axioms}} P_{2}
  \label{eq:synthesis:cegis2}
\end{align}
The choice of the grammar $\Sigma$ is critical for the FGH-optimizer.
If it is too restricted, then the optimizer will be limited too, if it
is too general, then the optimizer will take a prohibitive amount of
time to explore the entire space.  We briefly describe our design at a
high level.  Recall that $X$ denotes multiple IDBs, and the query
$G(X)$ may also return multiple intermediate relations. In our system
$G(X)$ is restricted to return a single relation, so we will assume
that $Y = G(X)$ is a single IDB.  The expression $G$ is known to us,
and is a sum-sum-product expression, see
Eq.~\eqref{eq:sum:sum:product},
\begin{align*}
  G(X) = & G_1(X) \oplus \cdots \oplus G_m(X)
\end{align*}
where each $G_i(X)$ is a sum-product expression,
Eq.~\eqref{eq:t:monomial}, using the IDBs $X$ and/or the EDBs.

To generate $\texttt{normalize}(H(G(X)))$, \update{we group its
  sum-products by the number of occurrences of $Y$:}
\begin{align*}
  \texttt{normalize}(H(Y)) = &  H^{(0)} \oplus H^{(1)}(Y) \oplus
  \cdots \oplus H^{(k_{\max})}(Y)
\end{align*}
where $H^{(k)}$ is a sum-sum-product
$H^{(k)} = Q_1 \oplus Q_2 \oplus \cdots$ s.t. each $Q_i$ contains
exactly $k$ occurrences of $Y$, and an arbitrary number of EDBs (it
may not contain the IDBs $X$).  We choose $k_{\max}$ as the largest
number of recursive IDBs $X$ that occur in any rule of the original
program $F(X)$, e.g.,  if the original program was linear, then
$k_{\max} \defeq 1$.  We obtain:
{\footnotesize
\begin{align*}
  & \texttt{normalize}(H(G(X))) = \\
  & \ \ \ \   H^{(0)} \oplus \texttt{normalize}(H^{(1)}(G(X))) \oplus \cdots \oplus \texttt{normalize}(H^{(k_{\max})}(G(X)))
\end{align*}
}
The grammar $\Sigma$ is shown in Fig.~\ref{fig:grammar}.  The start
symbol, $A$, generates a sum matching the expression above.  $A_0$
generates $H^{(0)}$, which is a sum of sum-product terms without any
occurrence of $Y$.  Recall from Sec.~\ref{subsubsec:smt:test} that the
expression $u(z,Q)$ denotes $\bigoplus_z Q$.  $E$ is one of the EDBs,
and $Z$ is a non-terminal for which we define rules
$Z \rightarrow z_1 | z_2 | \cdots | z_m | z_1' | z_2' \cdots$ where
$z_1, \ldots, z_m$ are variables that already occur in
$\texttt{normalize}(G(F(X)))$, and $z_1', z_2', \ldots$ is some fixed
set of fresh variable names. $A_k$ generates
$\texttt{normalize}(H^{(k)}(G(X)))$, which is a sum of sum-products,
each with exactly $k$ occurrence of $Y$.  \update{As stated in
  Fig.~\ref{fig:grammar}, the rules for $A_k$ are incorrect.  For
  example consider $A_1$:}
the $m$ non-terminals $A_{11}, \ldots, A_{1m}$ \update{should have
  identical derivations, instead of being expanded independently}.
For example, assume $G = G_1 \oplus G_2$ (thus $m=2$) and we want $H$
to be one of $E_1 \otimes Y$ or $E_2 \otimes Y$ or $E_3 \otimes Y$.
Then, $\texttt{normalize}(H(G(X)))$ can be one of the following three
expressions $E_1\otimes G_1 \oplus E_1 \otimes G_2$ or
$E_2 \otimes G_1 \oplus E_2 \otimes G_2$ or
$E_3 \otimes G_1 \oplus E_3 \otimes G_2$.  However, the grammar
\update{$A_1 \rightarrow A_{11}\oplus A_{12}$} also generates
incorrect expressions $E_1 \otimes G_1 \oplus E_2 \otimes G_2$,
\update{because $A_{11}, A_{12}$ can choose independently the IDB
  $E_1, E_2$, or $E_3$}.  We fix this by exploiting the choice
variables in Rosette: we simply use the same variables in
$A_{11}, A_{12}, \ldots$ ensuring that all these non-terminals make
exactly the same choices.  We note that our current system is
restricted to linear programs, hence $k_{\max}=1$.

\begin{figure}
\fcolorbox{black}{light-gray}{\parbox{0.2\textwidth}{
\footnotesize
%
\begin{align*}
A \rightarrow \ & \parbox{15mm}{\mbox{$A_0 \oplus A_1 \oplus \cdots \oplus A_{k_{\max}},$}} \\
A_0 \rightarrow \ & \parbox{15mm}{\mbox{$Q_0 \mid Q_0 \oplus A_0,\ \ Q_0 \rightarrow \ u(Z,Q_0) \mid Q_0 \otimes Q_0 \mid E(Z,Z,\cdots, Z),$}}\\
A_1\rightarrow  & \parbox{15mm}{\mbox{$A_{11} \oplus \cdots \oplus A_{1m},\ \ A_2\rightarrow  \  A_{211} \oplus \cdots \oplus A_{2mm}, \ \ A_3 \rightarrow A_{3111}\oplus \ldots$}}\\
A_{1i} \rightarrow \ & Q_{1i} \mid Q_{1i} \oplus A_{1i}, & Q_{1i} \rightarrow \ & u(Z,Q_{1i}) \mid Q_{1i} \otimes Q_0 \mid G_i(X), &i=&1,m\\
A_{2ij} \rightarrow \ & Q_{2ij} \oplus A_{2ij}, & Q_{2ij} \rightarrow \ & u(Z,Q_{2ij}) \mid Q_{1i} \otimes G_j(X), &i,j=&1,m\\
A_{3ij\ell} \rightarrow \ & Q_{3ij\ell} \oplus A_{3ij\ell}, &Q_{3ij\ell} \rightarrow \ & u(Z,Q_{3ij\ell}) \mid Q_{2ij} \otimes G_\ell(X), &i,j,\ell=&1,m
  \end{align*}
}}
\caption{Grammar $\Sigma$ for $\texttt{normalize}(H(G(X)))$, for $k_{\max}=3$.}
  \label{fig:grammar}
\end{figure}

\subsubsection{Discussion}\label{sec:grammar}\ \update{Even though our grammar is
  restricted to $k_{\max}=1$,} it is more complex than
Fig~\ref{fig:grammar}, in order to further reduce the search space.
We use more non-terminals to better control which variables $z$ can be
used where, and we also consider the choice of including entire
subexpressions that occur in the original program $P_1$, since they
are often reused in the optimized program. The synthesizer would
require many trials to find them, had we not included them explicitly.

\section{Equality Saturation}
\label{sec:semantic-opt}


\begin{figure}
  \includegraphics[width=0.45\linewidth]{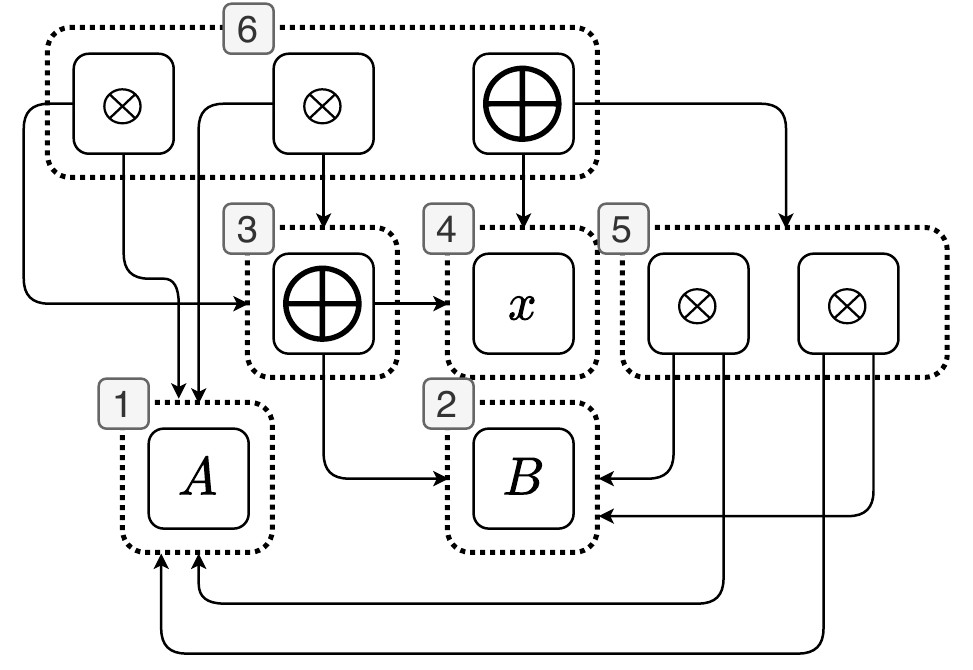}
  \caption{Example \egraph.}\label{fig:eqsat}
\end{figure}

Throughout the FGH-optimizer we need to manipulate expressions, apply
rules, and manage equivalent expressions.  This problem is common to
all query optimizers.  Instead of implementing our own expression
manager, we adopt a state-of-the-art rewriting system dubbed
Equality Saturation (\eqsat).
Specifically, we used
EGG~\cite{DBLP:journals/pacmpl/WillseyNWFTP21} to implement the green
boxes in the architecture shown in Fig.~\ref{fig:arch}.

An \eqsat\ system maintains a data structure called an
\egraph\ that compactly represents a set of expressions, together with
an equivalence relation over this set.  Each \egraph\ consists of a set
of \eclasses, each \eclass\ consists of a set of \enodes, and each
\enode\ is a function symbol with \eclasses\ as children.
Figure~\ref{fig:eqsat} shows an \egraph\ representing the two
expressions in Eq.~\eqref{eq:axiom:pullmul}, their subexpressions, and
other equivalent expressions.  Each \eclass\ (dotted box) represents a class of
equivalent expressions.  For example \eclass\ 5 represents
$A \otimes B$ and $B \otimes A$, which are equivalent by
commutativity.  \eclass\ 6 represents four equivalent expressions
(including the two choices in \eclass\ 5).

The \eqsat\ system maintains separately a collection of {\em rules},
each represented by a pair of patterns.  For example, one rule may
state that $\otimes$ is commutative: $x \otimes y = y \otimes x$.  The
\egraph\ can efficiently add a new expression to its collection,
insert a new rule, and match a given expression against the \egraph.

We describe how we use EGG in the FGH-optimizer.  First, we use it to
extend the Rule-based test (Sec.~\ref{sec:axiomatic-rewrite}) to
account for a constraint $\Gamma$.  By design, the \egraph\ makes it
easy to infer the equivalence $P_1 = P_2$ from a set of rules.
Suppose we want to check such an equivalence conditioned on $\Gamma$.
We may assume w.l.o.g.\ that $\Gamma$ is a logical implication,
$\Delta \Rightarrow \Theta$ since all database constraints are
expressed this way.  We convert it into an equivalence
$\Delta \wedge \Theta = \Delta$, and insert it into the \egraph, then
check for equivalence $P_1 = P_2$.

Second, we use the \egraph\ to denormalize an expression.  More
precisely, recall from Sec.~\ref{subsec:rule:based:synthesis} that we
attempt to synthesize $H$ by denormalizing
$P_1 \defeq \text{normalize}(F(G(X)))$, in other words, writing it in
the form $H(G(X))$.  For that we add $G(X)$ to the \egraph, observe in
which \eclass\ it is inserted, and replace that \eclass\ with a new
node $Y$.  The root of the new \egraph\ represents many equivalent
expressions, and each of them is a candidate for $H$.  We choose the
expression $H$ that has the smallest AST {\em and} does not have any
occurrence of the IDBs $X$.

Finally, we use the \egraph\ to infer the loop invariants.  We do this
by symbolically executing the recursive program $F$ for up to 5
iterations, and compute the symbolic expressions of the IDBs $X$:
$X_0, X_1, \ldots$  Using an \egraph\ we represent all identities
satisfied by these (distinct!) expressions.  The identities that are
satisfied by every $X_i$ are candidate loop invariants: for each of
them we use the SMT solver to check if they satisfy
Eq.~\eqref{eq:invariant:1} from 
Sec.~\ref{subsec:loop:invariants}.


\begin{figure*}
{\footnotesize
  \begin{center}
    \begin{tabular}{|l|l|l|l|l|l|} \hline
      Program & Synthesis Type & Constraint? & Invariant? & Dataset & \required{Size (\# ops)} \\ \hline
      Beyond Magic (BM) & rule-based & No & Yes & twitter, epinions, wiki & \required{6} \\ \hline
      Connected Components (CC) & rule-based & No & No & twitter, epinions, wiki & \required{6} \\ \hline
      Single Source Shortest Path (SSSP) & rule-based & No & No & twitter, epinions, wiki & \required{17} \\ \hline
      Sliding Window Sum (WS) & \cegis\ & No & Yes & Vector of Numbers & \required{15} \\ \hline
      Betweenness Centrality (BC) & \cegis\ & No & No & Erdős–Rényi Graphs & \required{43} \\ \hline
      Graph Radius (R) & \cegis\ & Yes & Yes & Random Recursive Trees & \required{12} \\ \hline
      Multi-level Marketing (MLM) & \cegis\ & Yes & Yes & Random Recursive Trees & \required{6} \\ \hline
    \end{tabular}
  \end{center}
}
  \caption{Experimental Setup}
  \label{fig:setup}
\end{figure*}

\begin{figure*}
  \begin{subfigure}[b]{0.33\textwidth}
    \centering
    \includegraphics[width=\textwidth]{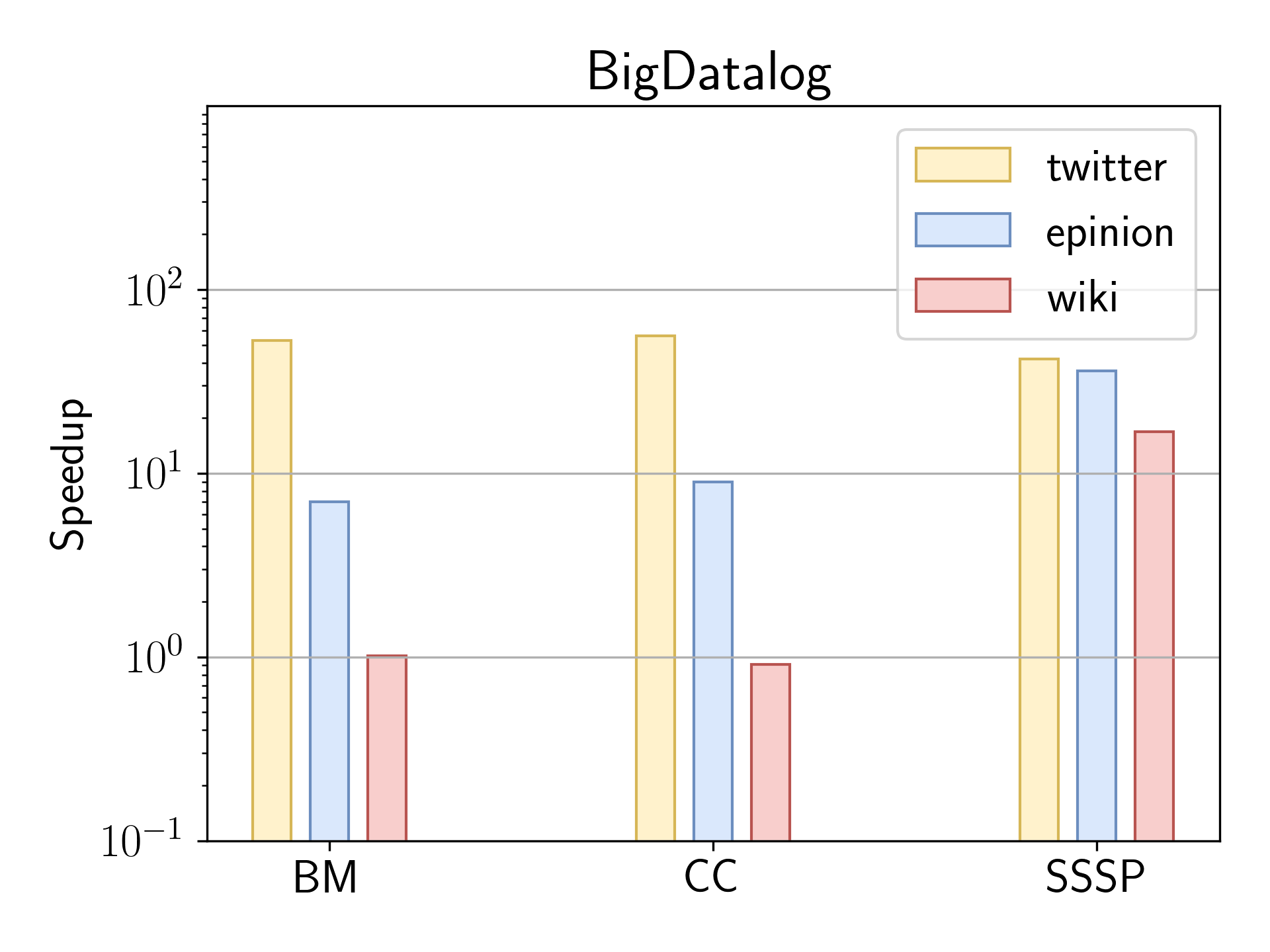}
  \end{subfigure}
  \hfill
  \begin{subfigure}[b]{0.33\textwidth}
    \centering
    \includegraphics[width=\textwidth]{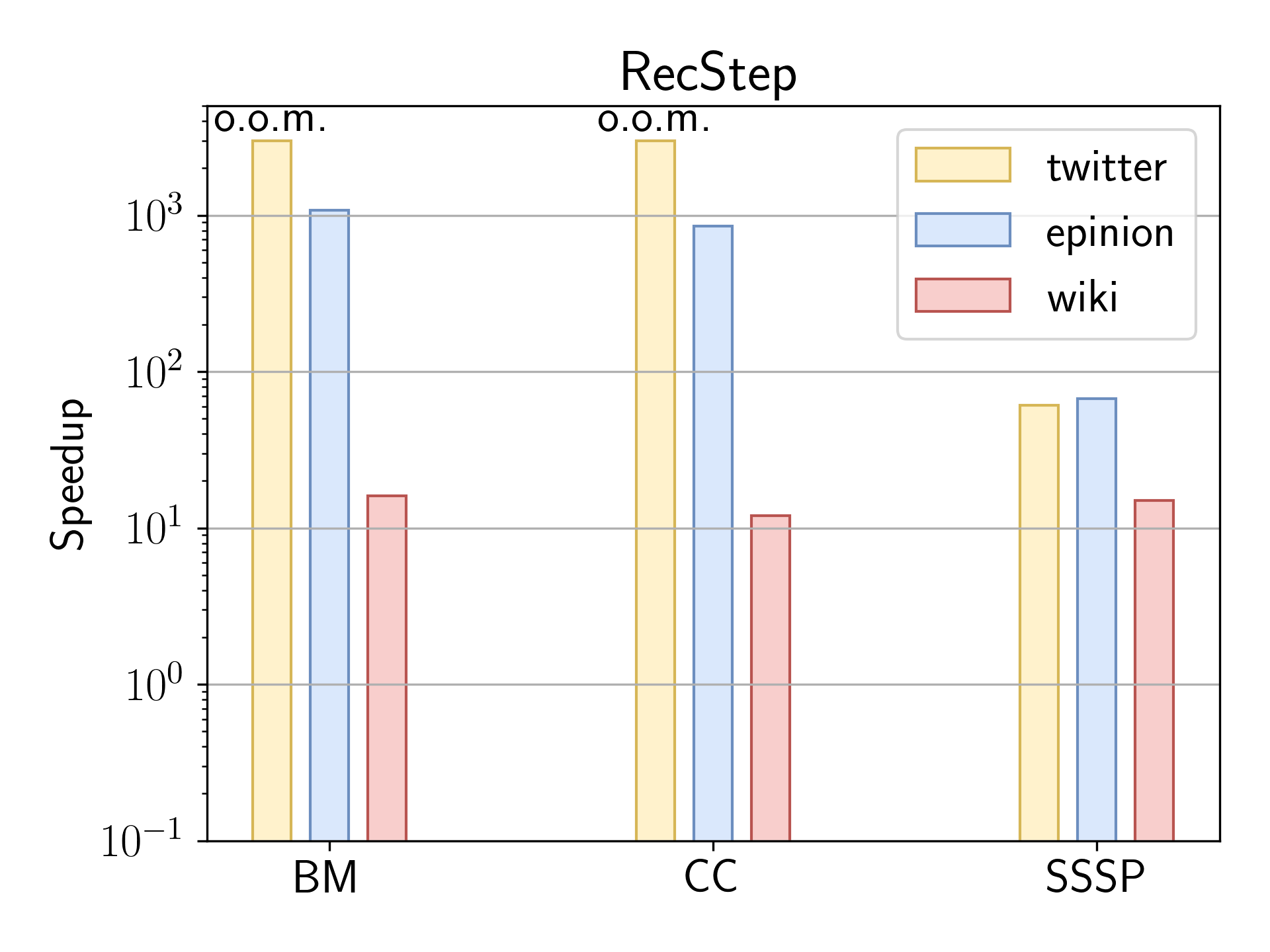}
  \end{subfigure}
  \hfill
  \begin{subfigure}[b]{0.33\textwidth}
    \centering
    \includegraphics[width=\textwidth]{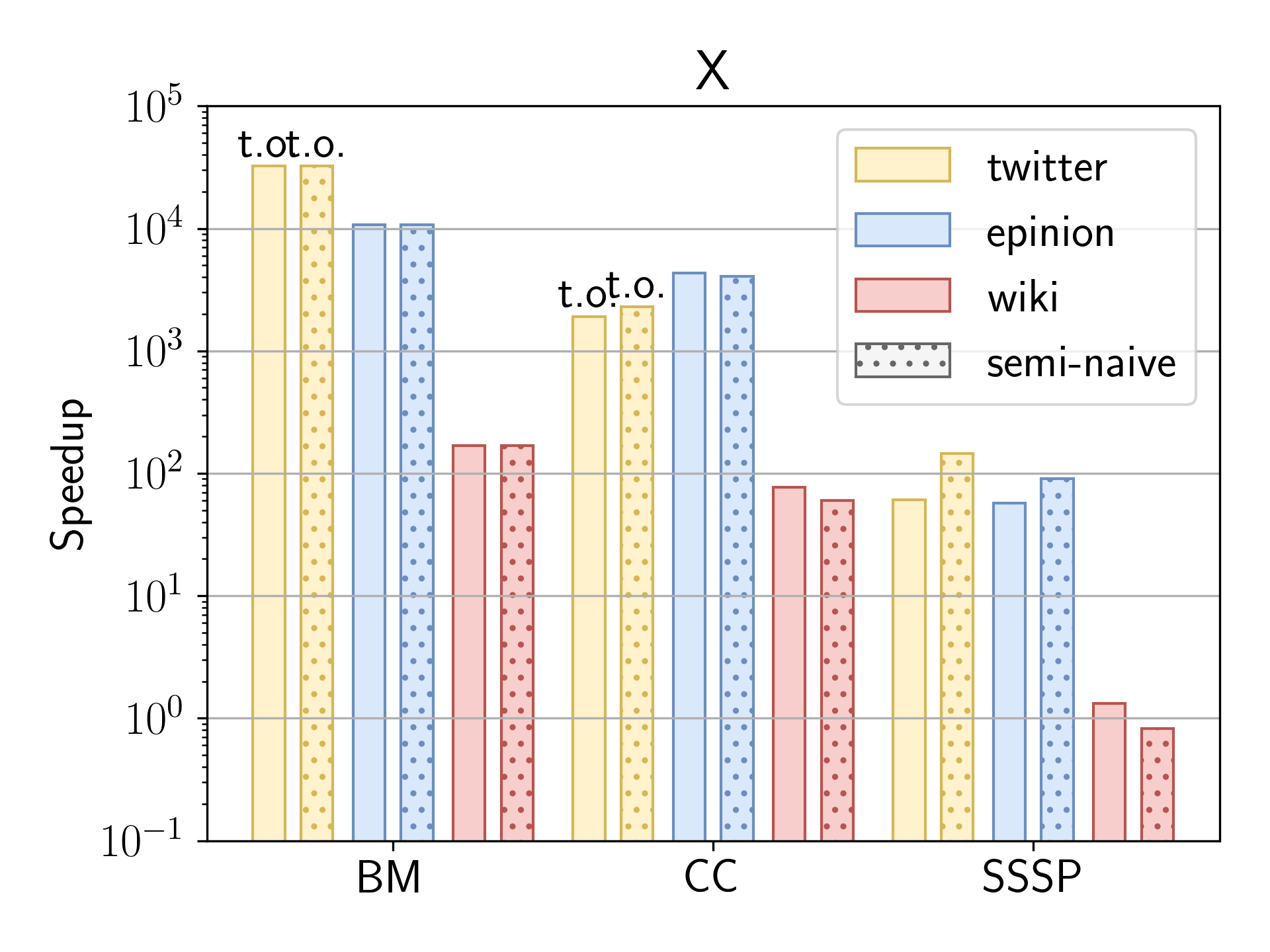}
  \end{subfigure}
  \caption{Speedup of the optimized v.s.\ original program; higher is better; 
  \required{t.o. means the original program timed out after 3 hours, in which case we report the speedup against 3 hours; o.o.m. means the original program ran out of memory.}}\label{fig:eval:eqsat}
\end{figure*}

\begin{figure*}
    \centering
    \includegraphics[width=\textwidth]{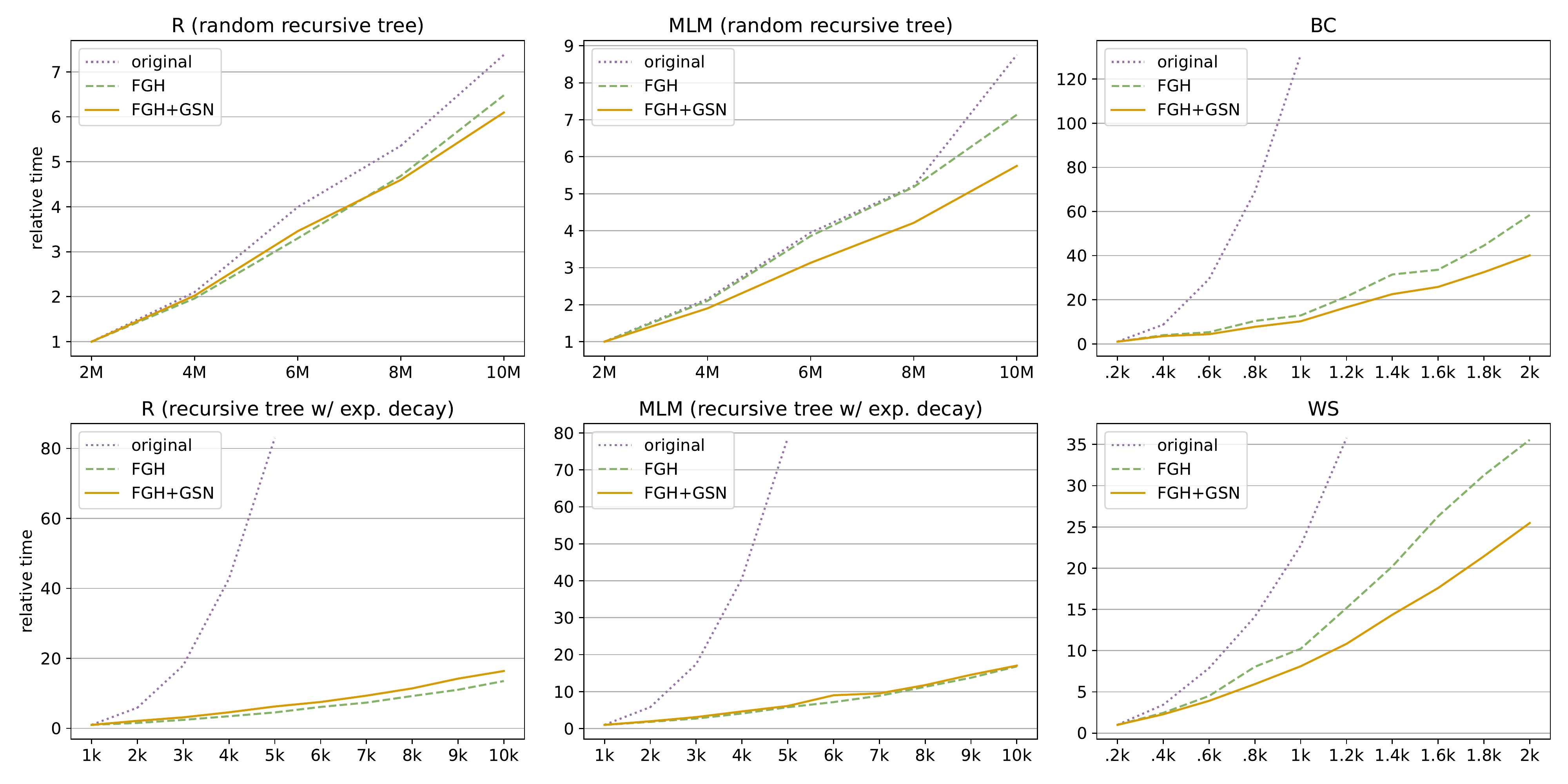}
    \caption{Runtime increase as a function of the data size; lower is  better.}\label{fig:eval:hard}
\end{figure*}

\begin{figure*}
  \centering
{\footnotesize
 \begin{tabular}{ c | c c c | c c c c }
 Program & BM & CC & SSSP & R & MLM & BC & WS \\
 \hline
 Invariance inference & 0.092 & 0 & 0 & 0.129 & 0.132 & 0 & 0\\
 Synthesis & 0.004 & 0.005 & 0.004 & 0.284 & 0.299 & 1.2 & 0.821 \\
 \hline
 Total & 0.096 & 0.005 & 0.004 & 0.413 & 0.431 & 1.2 & 0.821 \\
 Opt. / Exec. (max-min) & .82\% - .16\% & .04\%-.01\% & .24\%-.002\% & .41\%-.07\% & .76\%-.09\% & 6.3\%-.51\% & 7.4\%-.66\%
\end{tabular}
\qquad
\begin{tabular}{ c | c c c c }
 Program & R & MLM & BC & WS \\
 \hline
 Search space & 10 & 20 & 132 & 94 \\
\end{tabular}
}
  \caption{Optimization time in seconds, 
  \required{optimization time over execution time,}
   and  size of the search space. }
  \label{fig:synthesis:time:space}
\end{figure*}

\section{Evaluation}\label{sec:eval}

We implemented a source-to-source FGH-optimizer, based on
Fig.~\ref{fig:arch}.  The input is a program $\Pi_1$, given by $F, G$,
and a database constraint $\Gamma$, and the output is an optimized
program $H$.  We evaluated it on three Datalog systems, and several
programs from benchmarks proposed by prior
research~\cite{BigDatalog, DBLP:journals/pvldb/FanZZAKP19}; we also
propose new benchmarks that perform standard data analysis tasks.  We
did not modify any of the three Datalog engines.  We asked two major
questions:
\begin{enumerate}
\item How effective is our source-to-source optimization, given that
  each system already supports a range of optimizations?
\item How much time does the actual FGH optimization take?
\end{enumerate}

\subsection{Setup}


There is a great number of commercial and open-source Datalog engines in the
wild, but only a few support aggregates in recursion.
We were able to identify five major systems with such support:
SociaLite~\cite{DBLP:journals/tkde/SeoGL15},
Myria~\cite{10.14778/2824032.2824052}, the DeALS family of systems
(DeALS~\cite{DBLP:conf/icde/ShkapskyYZ15},
BigDatalog~\cite{DBLP:conf/sigmod/ShkapskyYICCZ16}, and
RaDlog~\cite{DBLP:conf/sigmod/0001WMSYDZ19}),
RecStep~\cite{DBLP:journals/pvldb/FanZZAKP19}, and
Dyna~\cite{francislandau-vieira-eisner-2020-wrla}.
Prior work~\cite{BigDatalog} reports SociaLite and Myria are consistently slower than
newer systems, so we do not include them in our experiments.
Dyna is designed to experiment with novel language semantics and
not for data analytics, and we were not able to run our benchmarks
without errors using it.
Systems in the DeALS family are similar to each other;
we pick BigDatalog because it is open source
and runs our benchmarks without errors;
we include RecStep for the same reasons.
Both BigDatalog and RecStep are multi-core systems.
Finally, we run experiments on an unreleased commercial system X,
which is single core.
As we shall discuss, X is the only one that supports all features for
our benchmarks.

We conducted all experiments on a server running CentOS 8.3.2011.
The server has a total of 1008GB memory, and
4 Intel Xeon CPU E7-4890 v2 2.80GHz CPUs, 
each with 15 cores and 30 threads. 
We ran seven
benchmarks, shown in Fig~\ref{fig:setup}. BM and CC are
Examples~\ref{ex:more:magic} and ~\ref{ex:fgh:cc}; MLM is basically
Example~\ref{ex:fgh:constraints}.  CC, SSSP and MLM are
from~\cite{BigDatalog}, the others are designed by us.  R and MLM
require a database constraint stating that the data is a tree.  BM,
R, and MLM each have a non-trivial loop invariant that is inferred by the
optimizer.  
\required{Our optimizer requires each program to consist of two rules,
one each for $F$ and $G$, and so a meaningful metric for
program size is the number of semiring operations.
These numbers are listed in the last column of Fig~\ref{fig:setup}.
Our benchmark programs are comparable in size to those used in prior work~\cite{BigDatalog, DBLP:journals/pvldb/FanZZAKP19}.}
All programs are available in our git repository.  The
real-world datasets twitter~\cite{mcauley2012learning}, epinions~\cite{richardson2003trust}, and
wiki~\cite{leskovec2010signed} are from the popular SNAP collection~\cite{snapnets}.  We
follow the setting in~\cite{BigDatalog,DBLP:journals/pvldb/FanZZAKP19}
when generating the synthetic graphs.
We additionally generate random recursive trees with an exponential
decay, modeling the decay of association in multi-level
marketing~\cite{emek2011mechanisms}.
For WS, we input the vector
$[1, \ldots, n]$, since the values of the entries do not affect run time.
In general, we used smaller datasets
than~\cite{BigDatalog,DBLP:journals/pvldb/FanZZAKP19} because some of
our experiments run single-threaded.

\subsection{Run Time Measurement}


For each program-dataset pair, we measure the run times of three
programs: original, with the FGH-optimization, and with the
FGH-optimization and the generalized semi-naive (GSN, for short)
transformation.  We report only the speedups relative to the original
program in Fig.~\ref{fig:eval:eqsat} and~\ref{fig:eval:hard}.  In some
cases the original program timed out our preset limit of 3 hours, 
\required{where we report the speedup against the 3 hours mark.
In some other cases the original program ran out of memory and we mark them with ``o.o.m.'' in the figure.
}  The {\em absolute} runtimes are irrelevant
for our discussion, since we want to report the effect of {\em adding}
our optimizations.  (We also do not have permission to report the
runtimes of X.)  All three systems already perform semi-naive
evaluation on the original program, since that is expressed over the
Boolean semiring.  But the FGH-optimized program is over a different
semiring (except for BM), and GSN has non-stratifiable rules with
negation, which are supported only by system X; we report GSN only for
system X.  While the benchmarks in Fig.~\ref{fig:eval:eqsat} were on
real datasets, those in Fig.~\ref{fig:eval:hard} use synthetic data,
for multiple reasons: we did not have access to a good tree dataset
needed in the R and MLM benchmarks, BC timed out on our real data (BC
is computationally expensive), and WS uses only a simple array.  A
benefit of synthetic data is that we can report how the optimizations
scale with the data size.  Unfortunately, the FGH-optimized programs
in Fig.~\ref{fig:eval:hard} require recursion with \textsf{SUM}
aggregation, which is not supported by BigDatalog or RecStep; this is
in contrast with those in Fig.~\ref{fig:eval:eqsat}, which require
recursion with \textsf{MIN} aggregation which is supported by all
systems.


\subsubsection{Findings} Figure~\ref{fig:eval:eqsat} shows the results
of the first group of benchmarks optimized by the
rule-based synthesizer.  Overall, we observe our optimizer provides
consistent and significant (up to 4 orders of magnitude) speedup across
systems and datasets.  Only a few datapoints indicate the optimization
has little effect: BM and CC on wiki under BigDatalog, and SSSP on
wiki under X.  This is due to the small size of the wiki dataset: both
the optimized and unoptimized programs finish very quickly, so the run
time is dominated by system overhead which cannot be optimized away.
We also note that (under X) GSN speeds up SSSP but slows down CC (note the log scale).
The latter occurs because the $\Delta$-relations for CC are very large,
and as a result the semi-naive evaluation has the same complexity as
the naive evaluation; but the semi-naive program is more complex and
incurs a constant slowdown.  GSN has no effect on BM because the
program is in the boolean semiring, and X already implements the
standard semi-naive evaluation.
Optimizing BM with FGH on BigDatalog sees a significant speedup even though the
systems already implements magic set rewrite,
because the optimization depends on a loop invariant.\footnote{
BigDatalog can optimize the left-recursive version of BM~\eqref{eq:simple:magic:nonopt}
to obtain similar speedup, via the classic magic set rewrite.}
Overall, both the semi-naive and
naive versions of the optimized program are significantly faster than
the unoptimized program.

Figure~\ref{fig:eval:hard} shows the results of the second group of
benchmarks, which required \cegis.  Since we used synthetic data,  we
examined here the asymptotic behavior of the optimization as a
function of the data size.
The most advanced optimization was for BC, which leads essentially to
Brandes' algorithm~\cite{brandes2001faster}: its effect is dramatic.
R and MLM rely on semantic optimization for a tree.  We generated two
synthetic trees, a random recursive tree with expected depth of
$O(\log n)$ and one with exponential decay with expected depth of
$O(n)$.  Since the benefit of the optimization depends on the depth,
we see a much better asymptotic behavior in the second case.  Here,
too, the optimizations were always improving the runtime.

%
%
%


\subsection{Optimization Time and the Size of the Search Space}

\cegis\ can quickly become very expensive if its search space is
large, and, for that reason, we have designed the grammar generator
carefully to reduce the search space without losing generality.
Fig.~\ref{fig:synthesis:time:space} reports the runtime of the
synthesizer (in seconds) for both rule-based synthesis and \cegis, and the size
of the search space.  The rule-based synthesizer runs in
milliseconds, while \cegis\ took over 1s for BC (our hardest
benchmark).  These numbers are close to those demanded by modern query
optimizers, and represent only a tiny portion of the total runtime
of the optimized query. 
\required{
Optimization time takes less than 1\% of the query run time for all benchmarks except for BC and WS on the smallest input data. }
To our surprise, our grammar managed to
narrow the search space considerably, to no more than 132 candidates,
which (in hindsight) explains the low optimization times. 
\required{
The search space can grow rapidly, and even exponentially, 
as the size of the input program grows.
Our optimizer optimizes a single stratum at a time, 
focusing on improving critical ``basic blocks'' of a program.
Our benchmark programs demonstrate a wide range of data analysis computation can
be expressed succinctly using just a few semiring operations, 
and optimization can have a dramatic impact on performance.}

\subsection{Summary}

We conclude that our optimizer can significantly speedup already
optimized  Datalog  systems, either single-core or multi-core.  GSN
can, sometimes, further improve the runtime.  We achieved this using a
rather small search space, which led to fast optimization.

%

%


\section{Conclusion}

We have presented a new optimization method for recursive queries,
which generalizes many previous optimizations described in the
literature.  We implemented it using a \cegis\ and an \eqsat\ system.
Our experiments have shown that this optimization is beneficial,
regardless of what other optimizations a Datalog system supports.  We
discuss here some limitations and future work.

\update{
  Our current implementation is restricted to linear programs, but our
  techniques apply to nonlinear programs as well. Non-linear programs
  require a more complex grammar $\Sigma$; this is likely to increase
  the search space, and possibly increase the optimization time. We
  leave this exploration to future work.
}

Our current optimizer is heuristic-based, and future work needs to
integrate it with a cost model.  This, however, will be challenging,
because very little work exists for estimating the cost of recursive
queries.  
\update{
  This paper applies a simple cost-model. We use the arity of the IDB
  predicate as a proxy for a simple asymptotic cost model, because
  $N^{\text{arity}}$ is the size bound of the output, when $N$ is the
  size of the active domain. This simple cost-model is currently used
  by the commercial DB system mentioned in the paper. If the optimized
  program reduces the arity, then it is assessed to have lower cost.
}

\update{
  Two limitations of our current implementation are the fact that we
  currently do not ``invent'' new IDBs for the optimized query, and do
  not apply the FGH-optimizer repeatedly.  Both would be required in
  order to support more advanced instances of magic set optimizations.
}

Our initial motivation for this work came from a real application,
which consists of a few hundred Datalog rules that were
computationally very expensive, and required a significant amount of
manual optimizations.  Upon close examination, at a very high level,
the manual optimization that we performed could be described,
abstractly, as a {\em sliding window} optimization (WS in
Fig.~\ref{fig:setup}), which is one of the simplest instantiations of
the FGH-rule.  Yet, our current system is far from able to optimize
automatically programs with hundreds of rules: we leave that for
future work.

\begin{acks}
Suciu and Wang were partially supported by NSF IIS 1907997 and NSF IIS 1954222.
Pichler was supported by the Austrian Science Fund (FWF):P30930.
\end{acks}

\balance
\bibliographystyle{ACM-Reference-Format}
\bibliography{main}


\begin{thebibliography}{56}


\ifx \showCODEN    \undefined \def \showCODEN     #1{\unskip}     \fi
\ifx \showDOI      \undefined \def \showDOI       #1{#1}\fi
\ifx \showISBNx    \undefined \def \showISBNx     #1{\unskip}     \fi
\ifx \showISBNxiii \undefined \def \showISBNxiii  #1{\unskip}     \fi
\ifx \showISSN     \undefined \def \showISSN      #1{\unskip}     \fi
\ifx \showLCCN     \undefined \def \showLCCN      #1{\unskip}     \fi
\ifx \shownote     \undefined \def \shownote      #1{#1}          \fi
\ifx \showarticletitle \undefined \def \showarticletitle #1{#1}   \fi
\ifx \showURL      \undefined \def \showURL       {\relax}        \fi
\providecommand\bibfield[2]{#2}
\providecommand\bibinfo[2]{#2}
\providecommand\natexlab[1]{#1}
\providecommand\showeprint[2][]{arXiv:#2}

\bibitem[\protect\citeauthoryear{Abiteboul, Hull, and Vianu}{Abiteboul
  et~al\mbox{.}}{1995}]%
        {DBLP:books/aw/AbiteboulHV95}
\bibfield{author}{\bibinfo{person}{Serge Abiteboul}, \bibinfo{person}{Richard
  Hull}, {and} \bibinfo{person}{Victor Vianu}.}
  \bibinfo{year}{1995}\natexlab{}.
\newblock \bibinfo{booktitle}{\emph{Foundations of Databases}}.
\newblock \bibinfo{publisher}{Addison-Wesley}.
\newblock
\showISBNx{0-201-53771-0}
\urldef\tempurl%
\url{http://webdam.inria.fr/Alice/}
\showURL{%
\tempurl}


\bibitem[\protect\citeauthoryear{Albarghouthi, Koutris, Naik, and
  Smith}{Albarghouthi et~al\mbox{.}}{2017}]%
        {DBLP:conf/cp/AlbarghouthiKNS17}
\bibfield{author}{\bibinfo{person}{Aws Albarghouthi},
  \bibinfo{person}{Paraschos Koutris}, \bibinfo{person}{Mayur Naik}, {and}
  \bibinfo{person}{Calvin Smith}.} \bibinfo{year}{2017}\natexlab{}.
\newblock \showarticletitle{Constraint-Based Synthesis of Datalog Programs}. In
  \bibinfo{booktitle}{\emph{Principles and Practice of Constraint Programming -
  23rd International Conference, {CP} 2017, Melbourne, VIC, Australia, August
  28 - September 1, 2017, Proceedings}} \emph{(\bibinfo{series}{Lecture Notes
  in Computer Science}, Vol.~\bibinfo{volume}{10416})},
  \bibfield{editor}{\bibinfo{person}{J.~Christopher Beck}} (Ed.).
  \bibinfo{publisher}{Springer}, \bibinfo{pages}{689--706}.
\newblock
\showISBNx{978-3-319-66157-5}
\urldef\tempurl%
\url{https://doi.org/10.1007/978-3-319-66158-2\_44}
\showDOI{\tempurl}


\bibitem[\protect\citeauthoryear{Alvaro, Marczak, Conway, Hellerstein, Maier,
  and Sears}{Alvaro et~al\mbox{.}}{2010}]%
        {DBLP:conf/datalog/AlvaroMCHMS10}
\bibfield{author}{\bibinfo{person}{Peter Alvaro}, \bibinfo{person}{William~R.
  Marczak}, \bibinfo{person}{Neil Conway}, \bibinfo{person}{Joseph~M.
  Hellerstein}, \bibinfo{person}{David Maier}, {and} \bibinfo{person}{Russell
  Sears}.} \bibinfo{year}{2010}\natexlab{}.
\newblock \showarticletitle{Dedalus: Datalog in Time and Space}. In
  \bibinfo{booktitle}{\emph{Datalog Reloaded - First International Workshop,
  Datalog 2010, Oxford, UK, March 16-19, 2010. Revised Selected Papers}}
  \emph{(\bibinfo{series}{Lecture Notes in Computer Science},
  Vol.~\bibinfo{volume}{6702})}, \bibfield{editor}{\bibinfo{person}{Oege
  de~Moor}, \bibinfo{person}{Georg Gottlob}, \bibinfo{person}{Tim Furche},
  {and} \bibinfo{person}{Andrew~Jon Sellers}} (Eds.).
  \bibinfo{publisher}{Springer}, \bibinfo{pages}{262--281}.
\newblock
\showISBNx{978-3-642-24205-2}
\urldef\tempurl%
\url{https://doi.org/10.1007/978-3-642-24206-9\_16}
\showDOI{\tempurl}


\bibitem[\protect\citeauthoryear{Balbin, Port, Ramamohanarao, and
  Meenakshi}{Balbin et~al\mbox{.}}{1991}]%
        {DBLP:journals/jlp/BalbinPRM91}
\bibfield{author}{\bibinfo{person}{Isaac Balbin}, \bibinfo{person}{Graeme~S.
  Port}, \bibinfo{person}{Kotagiri Ramamohanarao}, {and}
  \bibinfo{person}{Krishnamurthy Meenakshi}.} \bibinfo{year}{1991}\natexlab{}.
\newblock \showarticletitle{Efficient Bottom-UP Computation of Queries on
  Stratified Databases}.
\newblock \bibinfo{journal}{\emph{J. Log. Program.}} \bibinfo{volume}{11},
  \bibinfo{number}{3{\&}4} (\bibinfo{year}{1991}), \bibinfo{pages}{295--344}.
\newblock
\urldef\tempurl%
\url{https://doi.org/10.1016/0743-1066(91)90030-S}
\showDOI{\tempurl}


\bibitem[\protect\citeauthoryear{Bancilhon, Maier, Sagiv, and Ullman}{Bancilhon
  et~al\mbox{.}}{1986}]%
        {DBLP:conf/pods/BancilhonMSU86}
\bibfield{author}{\bibinfo{person}{Fran{\c{c}}ois Bancilhon},
  \bibinfo{person}{David Maier}, \bibinfo{person}{Yehoshua Sagiv}, {and}
  \bibinfo{person}{Jeffrey~D. Ullman}.} \bibinfo{year}{1986}\natexlab{}.
\newblock \showarticletitle{Magic Sets and Other Strange Ways to Implement
  Logic Programs}. In \bibinfo{booktitle}{\emph{Proceedings of the Fifth {ACM}
  {SIGACT-SIGMOD} Symposium on Principles of Database Systems, March 24-26,
  1986, Cambridge, Massachusetts, {USA}}},
  \bibfield{editor}{\bibinfo{person}{Avi Silberschatz}} (Ed.).
  \bibinfo{publisher}{{ACM}}, \bibinfo{pages}{1--15}.
\newblock
\urldef\tempurl%
\url{https://doi.org/10.1145/6012.15399}
\showDOI{\tempurl}


\bibitem[\protect\citeauthoryear{Beeri and Ramakrishnan}{Beeri and
  Ramakrishnan}{1991}]%
        {DBLP:journals/jlp/BeeriR91}
\bibfield{author}{\bibinfo{person}{Catriel Beeri} {and} \bibinfo{person}{Raghu
  Ramakrishnan}.} \bibinfo{year}{1991}\natexlab{}.
\newblock \showarticletitle{On the Power of Magic}.
\newblock \bibinfo{journal}{\emph{J. Log. Program.}} \bibinfo{volume}{10},
  \bibinfo{number}{3{\&}4} (\bibinfo{year}{1991}), \bibinfo{pages}{255--299}.
\newblock
\urldef\tempurl%
\url{https://doi.org/10.1016/0743-1066(91)90038-Q}
\showDOI{\tempurl}


\bibitem[\protect\citeauthoryear{Brandes}{Brandes}{2001}]%
        {brandes2001faster}
\bibfield{author}{\bibinfo{person}{Ulrik Brandes}.}
  \bibinfo{year}{2001}\natexlab{}.
\newblock \showarticletitle{A faster algorithm for betweenness centrality}.
\newblock \bibinfo{journal}{\emph{Journal of mathematical sociology}}
  \bibinfo{volume}{25}, \bibinfo{number}{2} (\bibinfo{year}{2001}),
  \bibinfo{pages}{163--177}.
\newblock


\bibitem[\protect\citeauthoryear{Chu, Wang, Weitz, and Cheung}{Chu
  et~al\mbox{.}}{2017}]%
        {DBLP:conf/cidr/ChuWWC17}
\bibfield{author}{\bibinfo{person}{Shumo Chu}, \bibinfo{person}{Chenglong
  Wang}, \bibinfo{person}{Konstantin Weitz}, {and} \bibinfo{person}{Alvin
  Cheung}.} \bibinfo{year}{2017}\natexlab{}.
\newblock \showarticletitle{Cosette: An Automated Prover for {SQL}}. In
  \bibinfo{booktitle}{\emph{8th Biennial Conference on Innovative Data Systems
  Research, {CIDR} 2017, Chaminade, CA, USA, January 8-11, 2017, Online
  Proceedings}}. \bibinfo{publisher}{www.cidrdb.org}.
\newblock
\urldef\tempurl%
\url{http://cidrdb.org/cidr2017/papers/p51-chu-cidr17.pdf}
\showURL{%
\tempurl}


\bibitem[\protect\citeauthoryear{de~Moura and Bj{\o}rner}{de~Moura and
  Bj{\o}rner}{2008}]%
        {10.1007/978-3-540-78800-3_24}
\bibfield{author}{\bibinfo{person}{Leonardo de Moura} {and}
  \bibinfo{person}{Nikolaj Bj{\o}rner}.} \bibinfo{year}{2008}\natexlab{}.
\newblock \showarticletitle{Z3: An Efficient SMT Solver}. In
  \bibinfo{booktitle}{\emph{Tools and Algorithms for the Construction and
  Analysis of Systems}}, \bibfield{editor}{\bibinfo{person}{C.~R. Ramakrishnan}
  {and} \bibinfo{person}{Jakob Rehof}} (Eds.). \bibinfo{publisher}{Springer
  Berlin Heidelberg}, \bibinfo{address}{Berlin, Heidelberg},
  \bibinfo{pages}{337--340}.
\newblock
\showISBNx{978-3-540-78800-3}


\bibitem[\protect\citeauthoryear{Deutsch, Popa, and Tannen}{Deutsch
  et~al\mbox{.}}{1999}]%
        {DBLP:conf/vldb/DeutschPT99}
\bibfield{author}{\bibinfo{person}{Alin Deutsch}, \bibinfo{person}{Lucian
  Popa}, {and} \bibinfo{person}{Val Tannen}.} \bibinfo{year}{1999}\natexlab{}.
\newblock \showarticletitle{Physical Data Independence, Constraints, and
  Optimization with Universal Plans}. In \bibinfo{booktitle}{\emph{VLDB'99,
  Proceedings of 25th International Conference on Very Large Data Bases,
  September 7-10, 1999, Edinburgh, Scotland, {UK}}},
  \bibfield{editor}{\bibinfo{person}{Malcolm~P. Atkinson},
  \bibinfo{person}{Maria~E. Orlowska}, \bibinfo{person}{Patrick Valduriez},
  \bibinfo{person}{Stanley~B. Zdonik}, {and} \bibinfo{person}{Michael~L.
  Brodie}} (Eds.). \bibinfo{publisher}{Morgan Kaufmann},
  \bibinfo{pages}{459--470}.
\newblock
\urldef\tempurl%
\url{http://www.vldb.org/conf/1999/P44.pdf}
\showURL{%
\tempurl}


\bibitem[\protect\citeauthoryear{Emek, Karidi, Tennenholtz, and Zohar}{Emek
  et~al\mbox{.}}{2011}]%
        {emek2011mechanisms}
\bibfield{author}{\bibinfo{person}{Yuval Emek}, \bibinfo{person}{Ron Karidi},
  \bibinfo{person}{Moshe Tennenholtz}, {and} \bibinfo{person}{Aviv Zohar}.}
  \bibinfo{year}{2011}\natexlab{}.
\newblock \showarticletitle{Mechanisms for multi-level marketing}. In
  \bibinfo{booktitle}{\emph{Proceedings of the 12th ACM conference on
  Electronic commerce}}. \bibinfo{pages}{209--218}.
\newblock


\bibitem[\protect\citeauthoryear{Fan, Zhu, Zhang, Albarghouthi, Koutris, and
  Patel}{Fan et~al\mbox{.}}{2019}]%
        {DBLP:journals/pvldb/FanZZAKP19}
\bibfield{author}{\bibinfo{person}{Zhiwei Fan}, \bibinfo{person}{Jianqiao Zhu},
  \bibinfo{person}{Zuyu Zhang}, \bibinfo{person}{Aws Albarghouthi},
  \bibinfo{person}{Paraschos Koutris}, {and} \bibinfo{person}{Jignesh~M.
  Patel}.} \bibinfo{year}{2019}\natexlab{}.
\newblock \showarticletitle{Scaling-Up In-Memory Datalog Processing:
  Observations and Techniques}.
\newblock \bibinfo{journal}{\emph{Proc. {VLDB} Endow.}} \bibinfo{volume}{12},
  \bibinfo{number}{6} (\bibinfo{year}{2019}), \bibinfo{pages}{695--708}.
\newblock
\urldef\tempurl%
\url{https://doi.org/10.14778/3311880.3311886}
\showDOI{\tempurl}


\bibitem[\protect\citeauthoryear{Fitting}{Fitting}{1991}]%
        {DBLP:journals/jlp/Fitting91}
\bibfield{author}{\bibinfo{person}{Melvin Fitting}.}
  \bibinfo{year}{1991}\natexlab{}.
\newblock \showarticletitle{Bilattices and the Semantics of Logic Programming}.
\newblock \bibinfo{journal}{\emph{J. Log. Program.}} \bibinfo{volume}{11},
  \bibinfo{number}{1{\&}2} (\bibinfo{year}{1991}), \bibinfo{pages}{91--116}.
\newblock
\urldef\tempurl%
\url{https://doi.org/10.1016/0743-1066(91)90014-G}
\showDOI{\tempurl}


\bibitem[\protect\citeauthoryear{Francis-Landau, Vieira, and
  Eisner}{Francis-Landau et~al\mbox{.}}{2020}]%
        {francislandau-vieira-eisner-2020-wrla}
\bibfield{author}{\bibinfo{person}{Matthew Francis-Landau},
  \bibinfo{person}{Tim Vieira}, {and} \bibinfo{person}{Jason Eisner}.}
  \bibinfo{year}{2020}\natexlab{}.
\newblock \showarticletitle{Evaluation of Logic Programs with Built-Ins and
  Aggregation: {A} Calculus for Bag Relations}. In
  \bibinfo{booktitle}{\emph{13th International Workshop on Rewriting Logic and
  Its Applications}}. \bibinfo{pages}{49--63}.
\newblock


\bibitem[\protect\citeauthoryear{Ganguly, Greco, and Zaniolo}{Ganguly
  et~al\mbox{.}}{1991}]%
        {DBLP:conf/pods/GangulyGZ91}
\bibfield{author}{\bibinfo{person}{Sumit Ganguly}, \bibinfo{person}{Sergio
  Greco}, {and} \bibinfo{person}{Carlo Zaniolo}.}
  \bibinfo{year}{1991}\natexlab{}.
\newblock \showarticletitle{Minimum and Maximum Predicates in Logic
  Programming}. In \bibinfo{booktitle}{\emph{Proceedings of the Tenth {ACM}
  {SIGACT-SIGMOD-SIGART} Symposium on Principles of Database Systems, May
  29-31, 1991, Denver, Colorado, {USA}}},
  \bibfield{editor}{\bibinfo{person}{Daniel~J. Rosenkrantz}} (Ed.).
  \bibinfo{publisher}{{ACM} Press}, \bibinfo{pages}{154--163}.
\newblock
\urldef\tempurl%
\url{https://doi.org/10.1145/113413.113427}
\showDOI{\tempurl}


\bibitem[\protect\citeauthoryear{Goldstein and Larson}{Goldstein and
  Larson}{2001}]%
        {DBLP:conf/sigmod/GoldsteinL01}
\bibfield{author}{\bibinfo{person}{Jonathan Goldstein} {and}
  \bibinfo{person}{Per{-}{\AA}ke Larson}.} \bibinfo{year}{2001}\natexlab{}.
\newblock \showarticletitle{Optimizing Queries Using Materialized Views: {A}
  practical, scalable solution}. In \bibinfo{booktitle}{\emph{Proceedings of
  the 2001 {ACM} {SIGMOD} international conference on Management of data, Santa
  Barbara, CA, USA, May 21-24, 2001}},
  \bibfield{editor}{\bibinfo{person}{Sharad Mehrotra} {and}
  \bibinfo{person}{Timos~K. Sellis}} (Eds.). \bibinfo{publisher}{{ACM}},
  \bibinfo{pages}{331--342}.
\newblock
\urldef\tempurl%
\url{https://doi.org/10.1145/375663.375706}
\showDOI{\tempurl}


\bibitem[\protect\citeauthoryear{Green}{Green}{2009}]%
        {DBLP:conf/icdt/Green09}
\bibfield{author}{\bibinfo{person}{Todd~J. Green}.}
  \bibinfo{year}{2009}\natexlab{}.
\newblock \showarticletitle{Containment of conjunctive queries on annotated
  relations}. In \bibinfo{booktitle}{\emph{Database Theory - {ICDT} 2009, 12th
  International Conference, St. Petersburg, Russia, March 23-25, 2009,
  Proceedings}} \emph{(\bibinfo{series}{{ACM} International Conference
  Proceeding Series}, Vol.~\bibinfo{volume}{361})},
  \bibfield{editor}{\bibinfo{person}{Ronald Fagin}} (Ed.).
  \bibinfo{publisher}{{ACM}}, \bibinfo{pages}{296--309}.
\newblock
\urldef\tempurl%
\url{https://doi.org/10.1145/1514894.1514930}
\showDOI{\tempurl}


\bibitem[\protect\citeauthoryear{Green, Karvounarakis, and Tannen}{Green
  et~al\mbox{.}}{2007}]%
        {DBLP:conf/pods/GreenKT07}
\bibfield{author}{\bibinfo{person}{Todd~J. Green}, \bibinfo{person}{Gregory
  Karvounarakis}, {and} \bibinfo{person}{Val Tannen}.}
  \bibinfo{year}{2007}\natexlab{}.
\newblock \showarticletitle{Provenance semirings}. In
  \bibinfo{booktitle}{\emph{Proceedings of the Twenty-Sixth {ACM}
  {SIGACT-SIGMOD-SIGART} Symposium on Principles of Database Systems, June
  11-13, 2007, Beijing, China}}, \bibfield{editor}{\bibinfo{person}{Leonid
  Libkin}} (Ed.). \bibinfo{publisher}{{ACM}}, \bibinfo{pages}{31--40}.
\newblock
\urldef\tempurl%
\url{https://doi.org/10.1145/1265530.1265535}
\showDOI{\tempurl}


\bibitem[\protect\citeauthoryear{Grossman, Cohen, Itzhaky, Rinetzky, and
  Sagiv}{Grossman et~al\mbox{.}}{2017}]%
        {DBLP:conf/cav/GrossmanCIRS17}
\bibfield{author}{\bibinfo{person}{Shelly Grossman}, \bibinfo{person}{Sara
  Cohen}, \bibinfo{person}{Shachar Itzhaky}, \bibinfo{person}{Noam Rinetzky},
  {and} \bibinfo{person}{Mooly Sagiv}.} \bibinfo{year}{2017}\natexlab{}.
\newblock \showarticletitle{Verifying Equivalence of Spark Programs}. In
  \bibinfo{booktitle}{\emph{Computer Aided Verification - 29th International
  Conference, {CAV} 2017, Heidelberg, Germany, July 24-28, 2017, Proceedings,
  Part {II}}} \emph{(\bibinfo{series}{Lecture Notes in Computer Science},
  Vol.~\bibinfo{volume}{10427})}, \bibfield{editor}{\bibinfo{person}{Rupak
  Majumdar} {and} \bibinfo{person}{Viktor Kuncak}} (Eds.).
  \bibinfo{publisher}{Springer}, \bibinfo{pages}{282--300}.
\newblock
\showISBNx{978-3-319-63389-3}
\urldef\tempurl%
\url{https://doi.org/10.1007/978-3-319-63390-9\_15}
\showDOI{\tempurl}


\bibitem[\protect\citeauthoryear{Gu, Watanabe, Mazza, Shkapsky, Yang, Ding, and
  Zaniolo}{Gu et~al\mbox{.}}{2019}]%
        {DBLP:conf/sigmod/0001WMSYDZ19}
\bibfield{author}{\bibinfo{person}{Jiaqi Gu}, \bibinfo{person}{Yugo~H.
  Watanabe}, \bibinfo{person}{William~A. Mazza}, \bibinfo{person}{Alexander
  Shkapsky}, \bibinfo{person}{Mohan Yang}, \bibinfo{person}{Ling Ding}, {and}
  \bibinfo{person}{Carlo Zaniolo}.} \bibinfo{year}{2019}\natexlab{}.
\newblock \showarticletitle{RaSQL: Greater Power and Performance for Big Data
  Analytics with Recursive-aggregate-SQL on Spark}. In
  \bibinfo{booktitle}{\emph{Proceedings of the 2019 International Conference on
  Management of Data, {SIGMOD} Conference 2019, Amsterdam, The Netherlands,
  June 30 - July 5, 2019}}, \bibfield{editor}{\bibinfo{person}{Peter~A. Boncz},
  \bibinfo{person}{Stefan Manegold}, \bibinfo{person}{Anastasia Ailamaki},
  \bibinfo{person}{Amol Deshpande}, {and} \bibinfo{person}{Tim Kraska}} (Eds.).
  \bibinfo{publisher}{{ACM}}, \bibinfo{pages}{467--484}.
\newblock
\showISBNx{978-1-4503-5643-5}
\urldef\tempurl%
\url{https://doi.org/10.1145/3299869.3324959}
\showDOI{\tempurl}


\bibitem[\protect\citeauthoryear{Halevy}{Halevy}{2001}]%
        {DBLP:journals/vldb/Halevy01}
\bibfield{author}{\bibinfo{person}{Alon~Y. Halevy}.}
  \bibinfo{year}{2001}\natexlab{}.
\newblock \showarticletitle{Answering queries using views: {A} survey}.
\newblock \bibinfo{journal}{\emph{{VLDB} J.}} \bibinfo{volume}{10},
  \bibinfo{number}{4} (\bibinfo{year}{2001}), \bibinfo{pages}{270--294}.
\newblock
\urldef\tempurl%
\url{https://doi.org/10.1007/s007780100054}
\showDOI{\tempurl}


\bibitem[\protect\citeauthoryear{Huang, Green, and Loo}{Huang
  et~al\mbox{.}}{2011}]%
        {10.1145/1989323.1989456}
\bibfield{author}{\bibinfo{person}{Shan~Shan Huang},
  \bibinfo{person}{Todd~Jeffrey Green}, {and} \bibinfo{person}{Boon~Thau Loo}.}
  \bibinfo{year}{2011}\natexlab{}.
\newblock \showarticletitle{Datalog and Emerging Applications: An Interactive
  Tutorial}. In \bibinfo{booktitle}{\emph{Proceedings of the 2011 ACM SIGMOD
  International Conference on Management of Data}} (Athens, Greece)
  \emph{(\bibinfo{series}{SIGMOD '11})}. \bibinfo{publisher}{Association for
  Computing Machinery}, \bibinfo{address}{New York, NY, USA},
  \bibinfo{pages}{1213–1216}.
\newblock
\showISBNx{9781450306614}
\urldef\tempurl%
\url{https://doi.org/10.1145/1989323.1989456}
\showDOI{\tempurl}


\bibitem[\protect\citeauthoryear{Khamis, Ngo, Pichler, Suciu, and Wang}{Khamis
  et~al\mbox{.}}{2021}]%
        {khamis21:_conver_datal_pre_semir}
\bibfield{author}{\bibinfo{person}{Mahmoud~Abo Khamis},
  \bibinfo{person}{Hung~Q. Ngo}, \bibinfo{person}{Reinhard Pichler},
  \bibinfo{person}{Dan Suciu}, {and} \bibinfo{person}{Yisu~Remy Wang}.}
  \bibinfo{year}{2021}\natexlab{}.
\newblock \bibinfo{booktitle}{\emph{{Convergence of Datalog over (Pre-)
  Semirings}}}.
\newblock
\showeprint[arxiv]{2105.14435v1}~[cs.DB]


\bibitem[\protect\citeauthoryear{Leskovec, Huttenlocher, and
  Kleinberg}{Leskovec et~al\mbox{.}}{2010}]%
        {leskovec2010signed}
\bibfield{author}{\bibinfo{person}{Jure Leskovec}, \bibinfo{person}{Daniel
  Huttenlocher}, {and} \bibinfo{person}{Jon Kleinberg}.}
  \bibinfo{year}{2010}\natexlab{}.
\newblock \showarticletitle{Signed networks in social media}. In
  \bibinfo{booktitle}{\emph{Proceedings of the SIGCHI conference on human
  factors in computing systems}}. \bibinfo{pages}{1361--1370}.
\newblock


\bibitem[\protect\citeauthoryear{Leskovec and Krevl}{Leskovec and
  Krevl}{2014}]%
        {snapnets}
\bibfield{author}{\bibinfo{person}{Jure Leskovec} {and} \bibinfo{person}{Andrej
  Krevl}.} \bibinfo{year}{2014}\natexlab{}.
\newblock \bibinfo{title}{{SNAP Datasets}: {Stanford} Large Network Dataset
  Collection}.
\newblock \bibinfo{howpublished}{\url{http://snap.stanford.edu/data}}.
\newblock


\bibitem[\protect\citeauthoryear{Levy, Mendelzon, Sagiv, and Srivastava}{Levy
  et~al\mbox{.}}{1995}]%
        {DBLP:conf/pods/LevyMSS95}
\bibfield{author}{\bibinfo{person}{Alon~Y. Levy}, \bibinfo{person}{Alberto~O.
  Mendelzon}, \bibinfo{person}{Yehoshua Sagiv}, {and} \bibinfo{person}{Divesh
  Srivastava}.} \bibinfo{year}{1995}\natexlab{}.
\newblock \showarticletitle{Answering Queries Using Views}. In
  \bibinfo{booktitle}{\emph{Proceedings of the Fourteenth {ACM}
  {SIGACT-SIGMOD-SIGART} Symposium on Principles of Database Systems, May
  22-25, 1995, San Jose, California, {USA}}},
  \bibfield{editor}{\bibinfo{person}{Mihalis Yannakakis} {and}
  \bibinfo{person}{Serge Abiteboul}} (Eds.). \bibinfo{publisher}{{ACM} Press},
  \bibinfo{pages}{95--104}.
\newblock
\urldef\tempurl%
\url{https://doi.org/10.1145/212433.220198}
\showDOI{\tempurl}


\bibitem[\protect\citeauthoryear{Mascellani and Pedreschi}{Mascellani and
  Pedreschi}{2002}]%
        {DBLP:conf/birthday/MascellaniP02}
\bibfield{author}{\bibinfo{person}{Paolo Mascellani} {and}
  \bibinfo{person}{Dino Pedreschi}.} \bibinfo{year}{2002}\natexlab{}.
\newblock \showarticletitle{The Declarative Side of Magic}. In
  \bibinfo{booktitle}{\emph{Computational Logic: Logic Programming and Beyond,
  Essays in Honour of Robert A. Kowalski, Part {II}}}
  \emph{(\bibinfo{series}{Lecture Notes in Computer Science},
  Vol.~\bibinfo{volume}{2408})}, \bibfield{editor}{\bibinfo{person}{Antonis~C.
  Kakas} {and} \bibinfo{person}{Fariba Sadri}} (Eds.).
  \bibinfo{publisher}{Springer}, \bibinfo{pages}{83--108}.
\newblock
\urldef\tempurl%
\url{https://doi.org/10.1007/3-540-45632-5\_4}
\showDOI{\tempurl}


\bibitem[\protect\citeauthoryear{McAuley and Leskovec}{McAuley and
  Leskovec}{2012}]%
        {mcauley2012learning}
\bibfield{author}{\bibinfo{person}{Julian~J McAuley} {and}
  \bibinfo{person}{Jure Leskovec}.} \bibinfo{year}{2012}\natexlab{}.
\newblock \showarticletitle{Learning to discover social circles in ego
  networks.}. In \bibinfo{booktitle}{\emph{NIPS}}, Vol.~\bibinfo{volume}{2012}.
  Citeseer, \bibinfo{pages}{548--56}.
\newblock


\bibitem[\protect\citeauthoryear{Mumick, Finkelstein, Pirahesh, and
  Ramakrishnan}{Mumick et~al\mbox{.}}{1990}]%
        {DBLP:conf/sigmod/MumickFPR90}
\bibfield{author}{\bibinfo{person}{Inderpal~Singh Mumick},
  \bibinfo{person}{Sheldon~J. Finkelstein}, \bibinfo{person}{Hamid Pirahesh},
  {and} \bibinfo{person}{Raghu Ramakrishnan}.} \bibinfo{year}{1990}\natexlab{}.
\newblock \showarticletitle{Magic is Relevant}. In
  \bibinfo{booktitle}{\emph{Proceedings of the 1990 {ACM} {SIGMOD}
  International Conference on Management of Data, Atlantic City, NJ, USA, May
  23-25, 1990}}, \bibfield{editor}{\bibinfo{person}{Hector Garcia{-}Molina}
  {and} \bibinfo{person}{H.~V. Jagadish}} (Eds.). \bibinfo{publisher}{{ACM}
  Press}, \bibinfo{pages}{247--258}.
\newblock
\urldef\tempurl%
\url{https://doi.org/10.1145/93597.98734}
\showDOI{\tempurl}


\bibitem[\protect\citeauthoryear{Mumick and Pirahesh}{Mumick and
  Pirahesh}{1994}]%
        {DBLP:conf/sigmod/MumickP94}
\bibfield{author}{\bibinfo{person}{Inderpal~Singh Mumick} {and}
  \bibinfo{person}{Hamid Pirahesh}.} \bibinfo{year}{1994}\natexlab{}.
\newblock \showarticletitle{Implementation of Magic-sets in a Relational
  Database System}. In \bibinfo{booktitle}{\emph{Proceedings of the 1994 {ACM}
  {SIGMOD} International Conference on Management of Data, Minneapolis,
  Minnesota, USA, May 24-27, 1994}},
  \bibfield{editor}{\bibinfo{person}{Richard~T. Snodgrass} {and}
  \bibinfo{person}{Marianne Winslett}} (Eds.). \bibinfo{publisher}{{ACM}
  Press}, \bibinfo{pages}{103--114}.
\newblock
\urldef\tempurl%
\url{https://doi.org/10.1145/191839.191860}
\showDOI{\tempurl}


\bibitem[\protect\citeauthoryear{Popa, Deutsch, Sahuguet, and Tannen}{Popa
  et~al\mbox{.}}{2000}]%
        {DBLP:conf/sigmod/PopaDST00}
\bibfield{author}{\bibinfo{person}{Lucian Popa}, \bibinfo{person}{Alin
  Deutsch}, \bibinfo{person}{Arnaud Sahuguet}, {and} \bibinfo{person}{Val
  Tannen}.} \bibinfo{year}{2000}\natexlab{}.
\newblock \showarticletitle{A Chase Too Far?}. In
  \bibinfo{booktitle}{\emph{Proceedings of the 2000 {ACM} {SIGMOD}
  International Conference on Management of Data, May 16-18, 2000, Dallas,
  Texas, {USA}}}, \bibfield{editor}{\bibinfo{person}{Weidong Chen},
  \bibinfo{person}{Jeffrey~F. Naughton}, {and} \bibinfo{person}{Philip~A.
  Bernstein}} (Eds.). \bibinfo{publisher}{{ACM}}, \bibinfo{pages}{273--284}.
\newblock
\urldef\tempurl%
\url{https://doi.org/10.1145/342009.335421}
\showDOI{\tempurl}


\bibitem[\protect\citeauthoryear{Raghothaman, Mendelson, Zhao, Naik, and
  Scholz}{Raghothaman et~al\mbox{.}}{2020}]%
        {DBLP:journals/pacmpl/RaghothamanMZNS20}
\bibfield{author}{\bibinfo{person}{Mukund Raghothaman},
  \bibinfo{person}{Jonathan Mendelson}, \bibinfo{person}{David Zhao},
  \bibinfo{person}{Mayur Naik}, {and} \bibinfo{person}{Bernhard Scholz}.}
  \bibinfo{year}{2020}\natexlab{}.
\newblock \showarticletitle{Provenance-guided synthesis of Datalog programs}.
\newblock \bibinfo{journal}{\emph{Proc. {ACM} Program. Lang.}}
  \bibinfo{volume}{4}, \bibinfo{number}{{POPL}} (\bibinfo{year}{2020}),
  \bibinfo{pages}{62:1--62:27}.
\newblock
\urldef\tempurl%
\url{https://doi.org/10.1145/3371130}
\showDOI{\tempurl}


\bibitem[\protect\citeauthoryear{Ramakrishnan and Srivastava}{Ramakrishnan and
  Srivastava}{1994}]%
        {DBLP:journals/debu/RamakrishnanS94}
\bibfield{author}{\bibinfo{person}{Raghu Ramakrishnan} {and}
  \bibinfo{person}{Divesh Srivastava}.} \bibinfo{year}{1994}\natexlab{}.
\newblock \showarticletitle{Semantics and Optimization of Constraint Queries in
  Databases}.
\newblock \bibinfo{journal}{\emph{{IEEE} Data Eng. Bull.}}
  \bibinfo{volume}{17}, \bibinfo{number}{2} (\bibinfo{year}{1994}),
  \bibinfo{pages}{14--17}.
\newblock
\urldef\tempurl%
\url{http://sites.computer.org/debull/94JUN-CD.pdf}
\showURL{%
\tempurl}


\bibitem[\protect\citeauthoryear{Richardson, Agrawal, and Domingos}{Richardson
  et~al\mbox{.}}{2003}]%
        {richardson2003trust}
\bibfield{author}{\bibinfo{person}{Matthew Richardson}, \bibinfo{person}{Rakesh
  Agrawal}, {and} \bibinfo{person}{Pedro Domingos}.}
  \bibinfo{year}{2003}\natexlab{}.
\newblock \showarticletitle{Trust management for the semantic web}. In
  \bibinfo{booktitle}{\emph{International semantic Web conference}}. Springer,
  \bibinfo{pages}{351--368}.
\newblock


\bibitem[\protect\citeauthoryear{Rockt\"aschel}{Rockt\"aschel}{[n.d.]}]%
        {einsum:rocktaschel}
\bibfield{author}{\bibinfo{person}{Tim Rockt\"aschel}.}
  \bibinfo{year}{[n.d.]}\natexlab{}.
\newblock \bibinfo{title}{Einsum is all you need - {Einstein} summation in deep
  learning}.
\newblock
  \bibinfo{howpublished}{\url{https://rockt.github.io/2018/04/30/einsum}}.
\newblock


\bibitem[\protect\citeauthoryear{Roscoe and Loo}{Roscoe and Loo}{2018}]%
        {DBLP:reference/db/RoscoeL18}
\bibfield{author}{\bibinfo{person}{Timothy Roscoe} {and}
  \bibinfo{person}{Boon~Thau Loo}.} \bibinfo{year}{2018}\natexlab{}.
\newblock \showarticletitle{Declarative Networking}.
\newblock In \bibinfo{booktitle}{\emph{Encyclopedia of Database Systems, Second
  Edition}}, \bibfield{editor}{\bibinfo{person}{Ling Liu} {and}
  \bibinfo{person}{M.~Tamer {\"{O}}zsu}} (Eds.). \bibinfo{publisher}{Springer}.
\newblock
\showISBNx{978-1-4614-8266-6}
\urldef\tempurl%
\url{https://doi.org/10.1007/978-1-4614-8265-9\_1220}
\showDOI{\tempurl}


\bibitem[\protect\citeauthoryear{Schlaipfer, Rajan, Lal, and Samak}{Schlaipfer
  et~al\mbox{.}}{2017}]%
        {DBLP:conf/sosp/SchlaipferRLS17}
\bibfield{author}{\bibinfo{person}{Matthias Schlaipfer},
  \bibinfo{person}{Kaushik Rajan}, \bibinfo{person}{Akash Lal}, {and}
  \bibinfo{person}{Malavika Samak}.} \bibinfo{year}{2017}\natexlab{}.
\newblock \showarticletitle{Optimizing Big-Data Queries Using Program
  Synthesis}. In \bibinfo{booktitle}{\emph{Proceedings of the 26th Symposium on
  Operating Systems Principles, Shanghai, China, October 28-31, 2017}}.
  \bibinfo{publisher}{{ACM}}, \bibinfo{pages}{631--646}.
\newblock
\showISBNx{978-1-4503-5085-3}
\urldef\tempurl%
\url{https://doi.org/10.1145/3132747.3132773}
\showDOI{\tempurl}


\bibitem[\protect\citeauthoryear{Seo, Guo, and Lam}{Seo et~al\mbox{.}}{2015}]%
        {DBLP:journals/tkde/SeoGL15}
\bibfield{author}{\bibinfo{person}{Jiwon Seo}, \bibinfo{person}{Stephen Guo},
  {and} \bibinfo{person}{Monica~S. Lam}.} \bibinfo{year}{2015}\natexlab{}.
\newblock \showarticletitle{SociaLite: An Efficient Graph Query Language Based
  on Datalog}.
\newblock \bibinfo{journal}{\emph{{IEEE} Trans. Knowl. Data Eng.}}
  \bibinfo{volume}{27}, \bibinfo{number}{7} (\bibinfo{year}{2015}),
  \bibinfo{pages}{1824--1837}.
\newblock
\urldef\tempurl%
\url{https://doi.org/10.1109/TKDE.2015.2405562}
\showDOI{\tempurl}


\bibitem[\protect\citeauthoryear{Shkapsky, Yang, Interlandi, Chiu, Condie, and
  Zaniolo}{Shkapsky et~al\mbox{.}}{2016a}]%
        {BigDatalog}
\bibfield{author}{\bibinfo{person}{Alexander Shkapsky}, \bibinfo{person}{Mohan
  Yang}, \bibinfo{person}{Matteo Interlandi}, \bibinfo{person}{Hsuan Chiu},
  \bibinfo{person}{Tyson Condie}, {and} \bibinfo{person}{Carlo Zaniolo}.}
  \bibinfo{year}{2016}\natexlab{a}.
\newblock \showarticletitle{Big Data Analytics with Datalog Queries on Spark}.
  In \bibinfo{booktitle}{\emph{Proceedings of the 2016 International Conference
  on Management of Data}} (San Francisco, California, USA)
  \emph{(\bibinfo{series}{SIGMOD '16})}. \bibinfo{publisher}{Association for
  Computing Machinery}, \bibinfo{address}{New York, NY, USA},
  \bibinfo{pages}{1135–1149}.
\newblock
\showISBNx{9781450335317}
\urldef\tempurl%
\url{https://doi.org/10.1145/2882903.2915229}
\showDOI{\tempurl}


\bibitem[\protect\citeauthoryear{Shkapsky, Yang, Interlandi, Chiu, Condie, and
  Zaniolo}{Shkapsky et~al\mbox{.}}{2016b}]%
        {DBLP:conf/sigmod/ShkapskyYICCZ16}
\bibfield{author}{\bibinfo{person}{Alexander Shkapsky}, \bibinfo{person}{Mohan
  Yang}, \bibinfo{person}{Matteo Interlandi}, \bibinfo{person}{Hsuan Chiu},
  \bibinfo{person}{Tyson Condie}, {and} \bibinfo{person}{Carlo Zaniolo}.}
  \bibinfo{year}{2016}\natexlab{b}.
\newblock \showarticletitle{Big Data Analytics with Datalog Queries on Spark}.
  In \bibinfo{booktitle}{\emph{Proceedings of the 2016 International Conference
  on Management of Data, {SIGMOD} Conference 2016, San Francisco, CA, USA, June
  26 - July 01, 2016}}, \bibfield{editor}{\bibinfo{person}{Fatma {\"{O}}zcan},
  \bibinfo{person}{Georgia Koutrika}, {and} \bibinfo{person}{Sam Madden}}
  (Eds.). \bibinfo{publisher}{{ACM}}, \bibinfo{pages}{1135--1149}.
\newblock
\showISBNx{978-1-4503-3531-7}
\urldef\tempurl%
\url{https://doi.org/10.1145/2882903.2915229}
\showDOI{\tempurl}


\bibitem[\protect\citeauthoryear{Shkapsky, Yang, and Zaniolo}{Shkapsky
  et~al\mbox{.}}{2015}]%
        {DBLP:conf/icde/ShkapskyYZ15}
\bibfield{author}{\bibinfo{person}{Alexander Shkapsky}, \bibinfo{person}{Mohan
  Yang}, {and} \bibinfo{person}{Carlo Zaniolo}.}
  \bibinfo{year}{2015}\natexlab{}.
\newblock \showarticletitle{Optimizing recursive queries with monotonic
  aggregates in DeALS}. In \bibinfo{booktitle}{\emph{31st {IEEE} International
  Conference on Data Engineering, {ICDE} 2015, Seoul, South Korea, April 13-17,
  2015}}, \bibfield{editor}{\bibinfo{person}{Johannes Gehrke},
  \bibinfo{person}{Wolfgang Lehner}, \bibinfo{person}{Kyuseok Shim},
  \bibinfo{person}{Sang~Kyun Cha}, {and} \bibinfo{person}{Guy~M. Lohman}}
  (Eds.). \bibinfo{publisher}{{IEEE} Computer Society},
  \bibinfo{pages}{867--878}.
\newblock
\showISBNx{978-1-4799-7964-6}
\urldef\tempurl%
\url{https://doi.org/10.1109/ICDE.2015.7113340}
\showDOI{\tempurl}


\bibitem[\protect\citeauthoryear{Si, Lee, Zhang, Albarghouthi, Koutris, and
  Naik}{Si et~al\mbox{.}}{2018}]%
        {DBLP:conf/sigsoft/SiLZAKN18}
\bibfield{author}{\bibinfo{person}{Xujie Si}, \bibinfo{person}{Woosuk Lee},
  \bibinfo{person}{Richard Zhang}, \bibinfo{person}{Aws Albarghouthi},
  \bibinfo{person}{Paraschos Koutris}, {and} \bibinfo{person}{Mayur Naik}.}
  \bibinfo{year}{2018}\natexlab{}.
\newblock \showarticletitle{Syntax-guided synthesis of Datalog programs}. In
  \bibinfo{booktitle}{\emph{Proceedings of the 2018 {ACM} Joint Meeting on
  European Software Engineering Conference and Symposium on the Foundations of
  Software Engineering, {ESEC/SIGSOFT} {FSE} 2018, Lake Buena Vista, FL, USA,
  November 04-09, 2018}}, \bibfield{editor}{\bibinfo{person}{Gary~T. Leavens},
  \bibinfo{person}{Alessandro Garcia}, {and} \bibinfo{person}{Corina~S.
  Pasareanu}} (Eds.). \bibinfo{publisher}{{ACM}}, \bibinfo{pages}{515--527}.
\newblock
\showISBNx{978-1-4503-5573-5}
\urldef\tempurl%
\url{https://doi.org/10.1145/3236024.3236034}
\showDOI{\tempurl}


\bibitem[\protect\citeauthoryear{Si, Raghothaman, Heo, and Naik}{Si
  et~al\mbox{.}}{2019}]%
        {DBLP:conf/ijcai/SiRHN19}
\bibfield{author}{\bibinfo{person}{Xujie Si}, \bibinfo{person}{Mukund
  Raghothaman}, \bibinfo{person}{Kihong Heo}, {and} \bibinfo{person}{Mayur
  Naik}.} \bibinfo{year}{2019}\natexlab{}.
\newblock \showarticletitle{Synthesizing Datalog Programs using Numerical
  Relaxation}. In \bibinfo{booktitle}{\emph{Proceedings of the Twenty-Eighth
  International Joint Conference on Artificial Intelligence, {IJCAI} 2019,
  Macao, China, August 10-16, 2019}}, \bibfield{editor}{\bibinfo{person}{Sarit
  Kraus}} (Ed.). \bibinfo{publisher}{ijcai.org}, \bibinfo{pages}{6117--6124}.
\newblock
\urldef\tempurl%
\url{https://doi.org/10.24963/ijcai.2019/847}
\showDOI{\tempurl}


\bibitem[\protect\citeauthoryear{Solar{-}Lezama, Tancau, Bod{\'{\i}}k, Seshia,
  and Saraswat}{Solar{-}Lezama et~al\mbox{.}}{2006}]%
        {DBLP:conf/asplos/Solar-LezamaTBSS06}
\bibfield{author}{\bibinfo{person}{Armando Solar{-}Lezama},
  \bibinfo{person}{Liviu Tancau}, \bibinfo{person}{Rastislav Bod{\'{\i}}k},
  \bibinfo{person}{Sanjit~A. Seshia}, {and} \bibinfo{person}{Vijay~A.
  Saraswat}.} \bibinfo{year}{2006}\natexlab{}.
\newblock \showarticletitle{Combinatorial sketching for finite programs}. In
  \bibinfo{booktitle}{\emph{Proceedings of the 12th International Conference on
  Architectural Support for Programming Languages and Operating Systems,
  {ASPLOS} 2006, San Jose, CA, USA, October 21-25, 2006}},
  \bibfield{editor}{\bibinfo{person}{John~Paul Shen} {and}
  \bibinfo{person}{Margaret Martonosi}} (Eds.). \bibinfo{publisher}{{ACM}},
  \bibinfo{pages}{404--415}.
\newblock
\urldef\tempurl%
\url{https://doi.org/10.1145/1168857.1168907}
\showDOI{\tempurl}


\bibitem[\protect\citeauthoryear{Tekle and Liu}{Tekle and Liu}{2019}]%
        {DBLP:journals/corr/abs-1909-08246}
\bibfield{author}{\bibinfo{person}{K.~Tuncay Tekle} {and}
  \bibinfo{person}{Yanhong~A. Liu}.} \bibinfo{year}{2019}\natexlab{}.
\newblock \showarticletitle{Extended Magic for Negation: Efficient
  Demand-Driven Evaluation of Stratified Datalog with Precise Complexity
  Guarantees}. In \bibinfo{booktitle}{\emph{Proceedings 35th International
  Conference on Logic Programming (Technical Communications), {ICLP} 2019
  Technical Communications, Las Cruces, NM, USA, September 20-25, 2019}}
  \emph{(\bibinfo{series}{{EPTCS}}, Vol.~\bibinfo{volume}{306})},
  \bibfield{editor}{\bibinfo{person}{Bart Bogaerts}, \bibinfo{person}{Esra
  Erdem}, \bibinfo{person}{Paul Fodor}, \bibinfo{person}{Andrea Formisano},
  \bibinfo{person}{Giovambattista Ianni}, \bibinfo{person}{Daniela Inclezan},
  \bibinfo{person}{Germ{\'{a}}n Vidal}, \bibinfo{person}{Alicia Villanueva},
  \bibinfo{person}{Marina~De Vos}, {and} \bibinfo{person}{Fangkai Yang}}
  (Eds.). \bibinfo{pages}{241--254}.
\newblock
\urldef\tempurl%
\url{https://doi.org/10.4204/EPTCS.306.28}
\showDOI{\tempurl}


\bibitem[\protect\citeauthoryear{Torlak and Bod{\'{\i}}k}{Torlak and
  Bod{\'{\i}}k}{2013}]%
        {DBLP:conf/oopsla/TorlakB13}
\bibfield{author}{\bibinfo{person}{Emina Torlak} {and}
  \bibinfo{person}{Rastislav Bod{\'{\i}}k}.} \bibinfo{year}{2013}\natexlab{}.
\newblock \showarticletitle{Growing solver-aided languages with rosette}. In
  \bibinfo{booktitle}{\emph{{ACM} Symposium on New Ideas in Programming and
  Reflections on Software, Onward! 2013, part of {SPLASH} '13, Indianapolis,
  IN, USA, October 26-31, 2013}}, \bibfield{editor}{\bibinfo{person}{Antony~L.
  Hosking}, \bibinfo{person}{Patrick~Th. Eugster}, {and}
  \bibinfo{person}{Robert Hirschfeld}} (Eds.). \bibinfo{publisher}{{ACM}},
  \bibinfo{pages}{135--152}.
\newblock
\urldef\tempurl%
\url{https://doi.org/10.1145/2509578.2509586}
\showDOI{\tempurl}


\bibitem[\protect\citeauthoryear{Torlak and Jackson}{Torlak and
  Jackson}{2007}]%
        {DBLP:conf/tacas/TorlakJ07}
\bibfield{author}{\bibinfo{person}{Emina Torlak} {and} \bibinfo{person}{Daniel
  Jackson}.} \bibinfo{year}{2007}\natexlab{}.
\newblock \showarticletitle{Kodkod: {A} Relational Model Finder}. In
  \bibinfo{booktitle}{\emph{Tools and Algorithms for the Construction and
  Analysis of Systems, 13th International Conference, {TACAS} 2007, Held as
  Part of the Joint European Conferences on Theory and Practice of Software,
  {ETAPS} 2007 Braga, Portugal, March 24 - April 1, 2007, Proceedings}}
  \emph{(\bibinfo{series}{Lecture Notes in Computer Science},
  Vol.~\bibinfo{volume}{4424})}, \bibfield{editor}{\bibinfo{person}{Orna
  Grumberg} {and} \bibinfo{person}{Michael Huth}} (Eds.).
  \bibinfo{publisher}{Springer}, \bibinfo{pages}{632--647}.
\newblock
\urldef\tempurl%
\url{https://doi.org/10.1007/978-3-540-71209-1\_49}
\showDOI{\tempurl}


\bibitem[\protect\citeauthoryear{Veanes, Grigorenko, de~Halleux, and
  Tillmann}{Veanes et~al\mbox{.}}{2009}]%
        {DBLP:conf/icfem/VeanesGHT09}
\bibfield{author}{\bibinfo{person}{Margus Veanes}, \bibinfo{person}{Pavel
  Grigorenko}, \bibinfo{person}{Peli de Halleux}, {and}
  \bibinfo{person}{Nikolai Tillmann}.} \bibinfo{year}{2009}\natexlab{}.
\newblock \showarticletitle{Symbolic Query Exploration}. In
  \bibinfo{booktitle}{\emph{Formal Methods and Software Engineering, 11th
  International Conference on Formal Engineering Methods, {ICFEM} 2009, Rio de
  Janeiro, Brazil, December 9-12, 2009. Proceedings}}
  \emph{(\bibinfo{series}{Lecture Notes in Computer Science},
  Vol.~\bibinfo{volume}{5885})}, \bibfield{editor}{\bibinfo{person}{Karin~K.
  Breitman} {and} \bibinfo{person}{Ana Cavalcanti}} (Eds.).
  \bibinfo{publisher}{Springer}, \bibinfo{pages}{49--68}.
\newblock
\showISBNx{978-3-642-10372-8}
\urldef\tempurl%
\url{https://doi.org/10.1007/978-3-642-10373-5\_3}
\showDOI{\tempurl}


\bibitem[\protect\citeauthoryear{Vianu}{Vianu}{2021}]%
        {10.1145/3452021.3458815}
\bibfield{author}{\bibinfo{person}{Victor Vianu}.}
  \bibinfo{year}{2021}\natexlab{}.
\newblock \showarticletitle{Datalog Unchained}. In
  \bibinfo{booktitle}{\emph{Proceedings of the 40th ACM SIGMOD-SIGACT-SIGAI
  Symposium on Principles of Database Systems}} (Virtual Event, China)
  \emph{(\bibinfo{series}{PODS'21})}. \bibinfo{publisher}{Association for
  Computing Machinery}, \bibinfo{address}{New York, NY, USA},
  \bibinfo{pages}{57–69}.
\newblock
\showISBNx{9781450383813}
\urldef\tempurl%
\url{https://doi.org/10.1145/3452021.3458815}
\showDOI{\tempurl}


\bibitem[\protect\citeauthoryear{Wang, Balazinska, and Halperin}{Wang
  et~al\mbox{.}}{2015}]%
        {10.14778/2824032.2824052}
\bibfield{author}{\bibinfo{person}{Jingjing Wang}, \bibinfo{person}{Magdalena
  Balazinska}, {and} \bibinfo{person}{Daniel Halperin}.}
  \bibinfo{year}{2015}\natexlab{}.
\newblock \showarticletitle{Asynchronous and Fault-Tolerant Recursive Datalog
  Evaluation in Shared-Nothing Engines}.
\newblock \bibinfo{journal}{\emph{Proc. VLDB Endow.}} \bibinfo{volume}{8},
  \bibinfo{number}{12} (\bibinfo{date}{Aug.} \bibinfo{year}{2015}),
  \bibinfo{pages}{1542–1553}.
\newblock
\showISSN{2150-8097}
\urldef\tempurl%
\url{https://doi.org/10.14778/2824032.2824052}
\showDOI{\tempurl}


\bibitem[\protect\citeauthoryear{Wang, Dillig, Lahiri, and Cook}{Wang
  et~al\mbox{.}}{2018}]%
        {DBLP:journals/pacmpl/0001DLC18}
\bibfield{author}{\bibinfo{person}{Yuepeng Wang}, \bibinfo{person}{Isil
  Dillig}, \bibinfo{person}{Shuvendu~K. Lahiri}, {and}
  \bibinfo{person}{William~R. Cook}.} \bibinfo{year}{2018}\natexlab{}.
\newblock \showarticletitle{Verifying equivalence of database-driven
  applications}.
\newblock \bibinfo{journal}{\emph{Proc. {ACM} Program. Lang.}}
  \bibinfo{volume}{2}, \bibinfo{number}{{POPL}} (\bibinfo{year}{2018}),
  \bibinfo{pages}{56:1--56:29}.
\newblock
\urldef\tempurl%
\url{https://doi.org/10.1145/3158144}
\showDOI{\tempurl}


\bibitem[\protect\citeauthoryear{Wang, Shah, Criswell, Pan, and Dillig}{Wang
  et~al\mbox{.}}{2020b}]%
        {DBLP:journals/pvldb/WangSCPD20}
\bibfield{author}{\bibinfo{person}{Yuepeng Wang}, \bibinfo{person}{Rushi Shah},
  \bibinfo{person}{Abby Criswell}, \bibinfo{person}{Rong Pan}, {and}
  \bibinfo{person}{Isil Dillig}.} \bibinfo{year}{2020}\natexlab{b}.
\newblock \showarticletitle{Data Migration using Datalog Program Synthesis}.
\newblock \bibinfo{journal}{\emph{Proc. {VLDB} Endow.}} \bibinfo{volume}{13},
  \bibinfo{number}{7} (\bibinfo{year}{2020}), \bibinfo{pages}{1006--1019}.
\newblock
\urldef\tempurl%
\url{https://doi.org/10.14778/3384345.3384350}
\showDOI{\tempurl}


\bibitem[\protect\citeauthoryear{Wang, Hutchison, Suciu, Howe, and Leang}{Wang
  et~al\mbox{.}}{2020a}]%
        {DBLP:journals/pvldb/WangHSHL20}
\bibfield{author}{\bibinfo{person}{Yisu~Remy Wang}, \bibinfo{person}{Shana
  Hutchison}, \bibinfo{person}{Dan Suciu}, \bibinfo{person}{Bill Howe}, {and}
  \bibinfo{person}{Jonathan Leang}.} \bibinfo{year}{2020}\natexlab{a}.
\newblock \showarticletitle{{SPORES:} Sum-Product Optimization via Relational
  Equality Saturation for Large Scale Linear Algebra}.
\newblock \bibinfo{journal}{\emph{Proc. {VLDB} Endow.}} \bibinfo{volume}{13},
  \bibinfo{number}{11} (\bibinfo{year}{2020}), \bibinfo{pages}{1919--1932}.
\newblock
\urldef\tempurl%
\url{http://www.vldb.org/pvldb/vol13/p1919-wang.pdf}
\showURL{%
\tempurl}


\bibitem[\protect\citeauthoryear{Willsey, Nandi, Wang, Flatt, Tatlock, and
  Panchekha}{Willsey et~al\mbox{.}}{2021}]%
        {DBLP:journals/pacmpl/WillseyNWFTP21}
\bibfield{author}{\bibinfo{person}{Max Willsey}, \bibinfo{person}{Chandrakana
  Nandi}, \bibinfo{person}{Yisu~Remy Wang}, \bibinfo{person}{Oliver Flatt},
  \bibinfo{person}{Zachary Tatlock}, {and} \bibinfo{person}{Pavel Panchekha}.}
  \bibinfo{year}{2021}\natexlab{}.
\newblock \showarticletitle{egg: Fast and extensible equality saturation}.
\newblock \bibinfo{journal}{\emph{Proc. {ACM} Program. Lang.}}
  \bibinfo{volume}{5}, \bibinfo{number}{{POPL}} (\bibinfo{year}{2021}),
  \bibinfo{pages}{1--29}.
\newblock
\urldef\tempurl%
\url{https://doi.org/10.1145/3434304}
\showDOI{\tempurl}


\bibitem[\protect\citeauthoryear{Zaniolo, Yang, Das, Shkapsky, Condie, and
  Interlandi}{Zaniolo et~al\mbox{.}}{2017}]%
        {DBLP:journals/tplp/ZanioloYDSCI17}
\bibfield{author}{\bibinfo{person}{Carlo Zaniolo}, \bibinfo{person}{Mohan
  Yang}, \bibinfo{person}{Ariyam Das}, \bibinfo{person}{Alexander Shkapsky},
  \bibinfo{person}{Tyson Condie}, {and} \bibinfo{person}{Matteo Interlandi}.}
  \bibinfo{year}{2017}\natexlab{}.
\newblock \showarticletitle{Fixpoint semantics and optimization of recursive
  Datalog programs with aggregates}.
\newblock \bibinfo{journal}{\emph{Theory Pract. Log. Program.}}
  \bibinfo{volume}{17}, \bibinfo{number}{5-6} (\bibinfo{year}{2017}),
  \bibinfo{pages}{1048--1065}.
\newblock
\urldef\tempurl%
\url{https://doi.org/10.1017/S1471068417000436}
\showDOI{\tempurl}


\bibitem[\protect\citeauthoryear{Zaniolo, Yang, Interlandi, Das, Shkapsky, and
  Condie}{Zaniolo et~al\mbox{.}}{2018}]%
        {DBLP:conf/amw/ZanioloYIDSC18}
\bibfield{author}{\bibinfo{person}{Carlo Zaniolo}, \bibinfo{person}{Mohan
  Yang}, \bibinfo{person}{Matteo Interlandi}, \bibinfo{person}{Ariyam Das},
  \bibinfo{person}{Alexander Shkapsky}, {and} \bibinfo{person}{Tyson Condie}.}
  \bibinfo{year}{2018}\natexlab{}.
\newblock \showarticletitle{Declarative BigData Algorithms via Aggregates and
  Relational Database Dependencies}. In \bibinfo{booktitle}{\emph{Proceedings
  of the 12th Alberto Mendelzon International Workshop on Foundations of Data
  Management, Cali, Colombia, May 21-25, 2018}} \emph{(\bibinfo{series}{{CEUR}
  Workshop Proceedings}, Vol.~\bibinfo{volume}{2100})},
  \bibfield{editor}{\bibinfo{person}{Dan Olteanu} {and}
  \bibinfo{person}{Barbara Poblete}} (Eds.). \bibinfo{publisher}{CEUR-WS.org}.
\newblock
\urldef\tempurl%
\url{http://ceur-ws.org/Vol-2100/paper2.pdf}
\showURL{%
\tempurl}


\end{thebibliography}

\newpage
\appendix

\onecolumn
\section{Grammar Refinements}
As discussed in~\Cref{sec:grammar}, we implement further
refinements of the grammar in~\Cref{fig:grammar} to 
limit the search space. 
First, we allow the user to specify a type for each attribute
of a relation. 
For example, a weighted edge relation 
$E(x: \text{ID}, y: \text{ID}, z: \text{int})$
has 2 attributes of type ID and the last attribute 
typed int. 
Although they all share the same concrete type 
(think ``machine type''), the abstract types help
guide the synthesizer to never use the same variables
for attributes of different types.
Second, we allow the user to define helper functions.
One example is the definition of $D$ in~\Cref{fig:bench:bc}
which computes the distance between two vertices.
$D$ is used repeatedly in the remaining definition 
in~\Cref{fig:bench:bc}, and it is a common practice 
for programmers to abstract out such recurring patterns.
We leverage good practice like this to aid synthesis:
if a user defines a helper function with head relation 
$R$, we will include $R$ as a base relation in our grammar
because it is likely that the helper function is also helpful 
in the optimized query.
When generating a candidate, we simply inline the definition
of the helper function, thereby including an entire 
sub-expression of the original query. 
Concretely, the refinements augment the grammar 
in~\Cref{fig:grammar} with types and adds the following 
cases to the production rule of $Q_0$: 
\[Q_0 \rightarrow \cdots \mid Z_{int} \mid F(Z, Z, \ldots, Z)\]
where $Z_{int}$ are variables of type int and $F$ is
a user-defined helper function.
\section{Benchmark Programs}

\Cref{fig:bench:ws,fig:bench:bc,fig:bench:mlm,fig:bench:r,fig:bench:bm,fig:bench:cc,fig:bench:sssp} contain the benchmark queries used in
the experiments.
In each query, $V$ is the set of input vertices and 
$E$ the set of input edges. 
A binary edge relation is unweighted, and 
a ternary edge relation is weighted with the weight 
in the third position.
The domain of all relations is integers.
Each query outputs the head relation of the last rule.
\begin{figure}[H]
  \begin{align*}
  TC(x, y) & = V(x) \wedge [x=y] \vee \exists t(TC(x, t) \wedge E(t, y))\\
  R(y) & = TC(\texttt{a}, y)
  \end{align*}
  \caption{Beyond Magic (BM). \texttt{a} is a constant 
  vertex ID, chosen uniformly at random during 
  experiments.}
  \label{fig:bench:bm}
\end{figure}
\begin{figure}[H]
  \begin{align*}
    TC(x, y) & = V(x) \wedge [x=y] \vee \exists t(TC(x, t) \wedge E(t, y))\\
    SCC[x] &= \min_v \setof{v}{TC(x, v)}
  \end{align*}
  \caption{Connected Components (CC). Note the vertex ID 
  itself is used as the label for that vertex, instead of
  $L(v)$ in~\cref{fig:cc}.}
  \label{fig:bench:cc}
\end{figure}
\begin{figure}[H]
  \begin{align*}
  D(x,d) & = [x=\texttt{a}] \wedge [d=0] \\
  & \vee \exists(y,d_1, d_2: D(y, d_1) \wedge E(y, x, d_2) \wedge [d=d_1+d_2]) \\
  SP[x] & = \min_d \setof{d}{D(x,d)}
  \end{align*}
  \caption{Single-source Shortest Paths (SSSP).
  \texttt{a} is a constant 
  vertex ID, chosen uniformly at random during 
  experiments.}
  \label{fig:bench:sssp}
\end{figure}
\begin{figure}[H]
  \begin{align*}
  W(t,j,w) &= A(j,w) \wedge [t=j] \\
  &\vee \exists(s : [t=s+1] \wedge W(s,j,w) \wedge [1 \leq j < t]) \\
  P[t] &= \sum_{j,w} \setof{w}{W(t,j,w)} \\
  S[t] &= P[t] - P[t-10]
  \end{align*}
  \caption{Window Sum (WS) with a window size of 10.
  $A$ is an array, and $A(i, v)$
is true when $A$ holds $v$ at index $i$, and 
arrays are 1-indexed.}
  \label{fig:bench:ws}
\end{figure}
\begin{figure}[H]
  \begin{align*}
  D(s,t,k) &= V(x) \wedge [s=t] \wedge [k=0]\\
  &\vee [k=1+\min_{v,l} \setof{l}{E(v, t) \wedge [s \neq t] \wedge D(s,v,l)}] \\
  \sigma(s, t, n) &= V(s) \wedge [s=t]  \wedge [n=1] \\
  & \vee E(v, t) \wedge D(s,v,d_{sv}) \wedge D(s,t,d_{st}) \\
  &\wedge [d_{sv} = d_{st} + 1] \wedge [s \neq t] \wedge \sigma(s,v,m)\\
  B[v] &= \sum_{s, t, b} 
  \{\sigma_{sv} \times \sigma_{vt} / \sigma_{st} \mid
  [s\neq t] \wedge [s\neq v] \wedge [t\neq v] \\
  &\wedge D(s, t, d_{st}) \wedge 
  D(s, v, d_{sv}) \wedge D(v, t, d_{vt}) 
  \wedge [d_{st} = d_{sv} + d_{vt}] \\
  & \wedge \sigma(s, v, \sigma_{sv}) \wedge \sigma(v, t, \sigma_{vt}) \wedge 
  \sigma(s, t, \sigma_{st})\}
  \end{align*}
  \caption{Betweenness Centrality (BC).
  Intuitively, $D$ computes the distance between two 
  vertices, and $\sigma$ computes the number of 
  shortest paths between two vertices.}
  \label{fig:bench:bc}
\end{figure}
\begin{figure}[H]
  \begin{align*}
  TC(x, y, w) &= V(x) \wedge [x=y] \wedge [w=0] \\
  &\vee \exists(z, w_1: TC(x, z, w_1) \wedge E(z, y) \wedge [w = w_1 + 1]) \\
  SP[x, y] &= \min_w \setof{w}{ TC(x, y, w)} \\
  R[x] &= \max_y \setof{SP[x,y]}{}
  \end{align*}
  \caption{Graph Radius (R).
  $R[x]$ computes the length of the longest shortest-path 
  between $x$ and any other vertex in the graph. Intuitively, it is the diameter with one vertex fixed.}
  \label{fig:bench:r}
\end{figure}
\begin{figure}[H]
  \begin{align*}
  TC(x, y) &= V(x) \wedge [x=y] \vee \exists(z : TC(x, z) \wedge E(z, y)) \\
  M[x] &= \sum_v \setof{v}{TC(x, v)}
  \end{align*}
  \caption{Multi-level Marketing (MLM).
  Intuitively, each vertex $v$ represents a participant 
  who makes $v$ amount of profit; the query $M[x]$
  computes the total profit of the sub-network under 
  participant $x$.}
  \label{fig:bench:mlm}
\end{figure}
\section{Magic Set Optimization}
\label{app:magic:sets}

We describe here a general form of magic set optimization and show
that its correctness can be proven by a verifier using only three
rules: the FGH rule, the Stratification Rule, and the Fixpoint rule,
described below.  Each of the rewrite rules can be proven by our
verifier, but our synthesizer cannot synthesize the magic rewritings
in general; it is currently restricted to relatively simple
rewritings, like those illustrated earlier in the paper.

\paragraph*{\bf Notations} In this section we restrict the discussion
to the Boolean semiring and monotone functions.  When we write $F(X)$
we assume that $X$ is a tuple of IDB relations, e.g. $X = (R, S, T)$,
and $F(X)$ returns a tuple of relations of the same arities.  A
datalog program has the form $X \cd F(X)$, and we denote by $\bar X$
its least fixpoint.  Given two tuples of relations $X,Y$ of the same
type, i.e. the same number of relations and of the same arities, we
write $X \Rightarrow Y$ to mean component-wise set inclusion.  For
example, if $X = (R_1, S_1, T_1)$ and $Y = (R_2, S_2, T_2)$ then
$X \Rightarrow Y$ means $R_1 \subseteq R_2$, $S_1 \subseteq S_2$,
$T_1 \subseteq T_2$.

\subsection{The Three Rules}

\paragraph*{\bf Simplified FGH Rule} We consider the following
simplified version of the FGH rule.  Given two datalog programs:
\begin{align*}
  \Pi_1:\ X \cd & F(X) & \Pi_2:\ Y \cd H(Y)
\end{align*}
a {\em homomorphism} from $\Pi_1$ to $\Pi_2$ is a function $G$
satisfying $G(\emptyset)=\emptyset$ and $G(F(X))=H(G(X))$ for all $X$.
If such a homomorphism exists, then the FGH-rule in
Theorem~\ref{th:fgh} implies the following:
%
%
%
\begin{align}
  G(\bar X) = & \bar Y \label{eq:fgh:simplified}
\end{align}
The only difference from the FGH-rule in Theorem~\ref{th:fgh} is that
we do not ask for $\Pi_1$ to return $G(X)$.

\paragraph*{\bf Stratification Rule} Assume $K(Z), F(Z,X)$ are two
monotone functions over tuples of relations, such that the output of
$K$ has the same type as $Z$, and the output of $F$ has the same type
as $X$.  Then the following two programs are equivalent:
\begin{align*}
  &
    \begin{array}[t]{rrl}
      \Pi:\ & Z \cd & K(Z) \\
            & X \cd & F(Z,X)
    \end{array} &&
                   \begin{array}[t]{rrll}
                     \Pi': & Z \cd & K(Z) &\mbox{// Stratum 0: let $\bar Z$ be its fixpoint}\\ \cline{2-3}
                           & X \cd & F(\bar Z,X) & \mbox{// Stratum 1: note the use of $\bar Z$}
                   \end{array}
\end{align*}
The program $\Pi$ computes both sets of rules $K, F$ in a single
stratum.  Program $\Pi'$ separates them into two strata: first it
computes the fixpoint $\bar Z$ of $K$, then uses it as an EDB to
compute the fixpoint of $F(\bar Z, X)$.  This stratification rule is
well known for monotone datalog, and we omit the proof.  A formal
statement asserts that, for any $\omega$-continuous functions $K$,
$F$, denoting $L(Z,X) \defeq (K(Z), F(Z,X))$, we have
$\lfp(L) = (\lfp(K), \lfp(\lambda X.F(\lfp(K),X)))$.


\paragraph*{\bf Fixpoint Rule} Consider a datalog program, and suppose
that its stratum $s$ is the following:
\begin{align*}
  X \cd & F(X)
\end{align*}
Let $\bar X$ be the fixpoint of stratum $s$.  Then the following
constraint holds in all strata $s' > s$:
\begin{align}
  F(\bar X) \Rightarrow & \bar X \label{eq:constraint:at:fixpoint}
\end{align}

\subsection{Running Example}

Throughout this section we will illustrate using the following example.

\begin{example} \label{ex:magic:running} We show below a program $\Pi$
  and its magic set optimized program $\Pi_O$.
  \begin{align*}
    &
      \begin{array}[b]{rrl}
    \Pi: & R(x) \cd & V(x) \\
         & R(x) \cd & T(x,y,z) \wedge R(y) \wedge R(z) \\
         & Q(x) \cd & G(x) \wedge R(x)
      \end{array}
&
  \begin{array}[b]{rrl}
    \Pi_O: & Q_0'() \cd & \\
           & R'_O(y) \cd & R'_O(x) \wedge T(x,y,z)\\
           & R'_O(z) \cd & R'_O(x) \wedge T(x,y,z) \wedge R_O(y)\\
           & R_O'(x) \cd & Q'_O() \wedge G(x) \\
           & R_O(x) \cd & R'_O(x) \wedge V(x) \\
           & R_O(x) \cd & R'_O(x) \wedge T(x,y,z) \wedge R(y) \wedge R(z) \\
           & Q_O(x) \cd & Q'_O() \wedge G(x) \wedge R_O(x)
  \end{array}
  \end{align*}
  $\Pi$ computes an IDB $R(x)$, then returns $Q(x)$ which is a
  restriction of $R$.  The optimized program $\Pi_O$ computes the IDB
  $R_O(x)$ only on a subset of the nodes $x$ that are sufficient to
  answer $Q$, namely the set defined by $R'_O(x)$.  The predicate
  $R'_O(x)$ is called the {\em magic predicate}.
\end{example}


\subsection{Definition of Magic Set Rewriting}

We use the elegant definition of magic set rewriting by Mascellani and
Pedreschi~\cite{DBLP:conf/birthday/MascellaniP02}.

\paragraph*{\bf Notation}
An atom $A$ is a predicate symbol followed by variables and/or
constants, e.g. $A$ can be $R('a','b',x,y)$.  Single atoms are denoted
$A,B, \dots $ and sequences (possibly empty) of atoms by
$\bA, \bB, \dots$ Each datalog rule has the form $A \cd \bA$.
Following the convention used in \cite{DBLP:books/aw/AbiteboulHV95}, a
query is given by a datalog program $\Pi$ and a query predicate $Q$
not occurring in $\Pi$, such that $Q$ is defined by a single rule
$r_Q$ of the form $r_Q : Q(\bv) \cd \bA$, where $\bA$ does not contain
the predicate $Q$.

\paragraph*{\bf Modes}
For an $n$-ary relation symbol $R$, a {\em mode\/} is a string
$\set{+,-}^n$.  Intuitively a $+$ represents an input, and a $-$
represents an output.  
%
%
%
Given a datalog program $\Pi$, we fix a moding for each relational
symbol $R$.  The moding can be arbitrary, with a single restriction:
the mode of the output predicate $Q$ must be $(-,-,\cdots,-)$,
i.e. all its positions are output positions.  To simplify the
notations, we will assume w.l.o.g. that, for each relational symbol
$R$ the input positions precede the output positions, i.e. its mode is
$+\cdots+-\cdots-$.  Hence, an atom $A$ of the form
$R(\mathbf{u}, \mathbf{v})$ has input arguments $\mathbf{u}$ and
output arguments~$\mathbf{v}$.

%

\paragraph*{\bf Magic set transformation}
To each relational symbol $R$ we associate two new symbols.  A {\em
  magic symbol} $R'_O$, whose arity is the number of input positions
in the mode of $R$, and an {\em optimized symbol} $R_O$, of the same
arity as $R$.  If $A$ is the atom $R(\mathbf{u}, \mathbf{v})$, then we
denote by $A'_O \defeq R'(\mathbf{u})$ and
$A_O \defeq R_O(\bm u, \bm v)$ the atoms with the corresponding
symbols $R'_O$ and $R_O$ respectively.  Similarly, if $\bm A$ is a
sequence of atoms, then we denote by $\bm A'_O$ and $\bm A_O$ the
corresponding sequences of atoms.  If $R$ is an EDB, then we define
$R_O$ to be the same EDB, and we will often remove the subscript $O$.

\begin{definition}
\label{def:magicsets}
\cite[Definition~3]{DBLP:conf/birthday/MascellaniP02} Let $\Pi$ be a
datalog program with an output predicate $Q$.  The magic set
transformation is the program $\Pi_O$ obtained from $\Pi$ by the
following transformation steps:

\begin{enumerate}
\item \label{item:def:magicsets:1} For every rule in $\Pi$ fix an
  order of the atoms in its body, i.e. the rule becomes:
  \begin{align*}
    r \colon\  A  \cd & B_1\wedge B_2 \wedge \cdots \wedge B_k
  \end{align*}
  For every atom $B_\ell$ above add the new rule:
  \begin{align}
    r'_\ell \colon\ B'_{\ell,O} \cd & A'_O \wedge B_{1,O} \wedge \cdots \wedge B_{\ell-1,O}\label{eq:def:of:r:prime}
  \end{align}
\item \label{item:def:magicsets:2} Add the following rule with an
  empty body (i.e. the body is \texttt{true}): $Q'_O() \cd$.
\item \label{item:def:magicsets:3} Replace each original rule
  $A \cd \bA$ in $\Pi$ by the new rule $A_O \cd A'_O, \bA_O$.
\end{enumerate}
\end{definition}

\begin{example} \label{ex:magic:running:2} In
  Example~\ref{ex:magic:running} we use the modes $R(+)$, and $Q(-)$.
  We associate to the atoms $R(x), Q(x)$ the magic atoms
  $R_O'(x), Q_O'()$ and the optimized atoms $R_O(x), Q_O(x)$.  The
  modes for the EDBs $V, T, G$ can be arbitrary, since the magic
  symbols $V'_O, T'_O, G'_O$ are never used in any rule, and hence
  they were omitted from $\Pi_O$; they are useful for us only to
  simplify the statement of Lemma~\ref{lemma:magic:sets:main:property}
  below, and for that reason we illustrate them here, assuming that
  their modes are $-,-,-$.  Then, according to
  item~\ref{item:def:magicsets:1} of Def.~\ref{def:magicsets}, the
  optimized program should include these rules:
  \begin{align*}
    V'_O() \cd & R'_O(x) \\
    T'_O() \cd & R'_O(x) \\
    G'_O() \cd & Q'()
  \end{align*}
  If we add these rules to the program $\Pi_O$ in
  Example~\ref{ex:magic:running}, then we observe that, at fixpoint,
  $\bar R_O' \neq \emptyset$ (assuming $G \neq \emptyset$), and
  therefore,
  $\bar V'_O() = \bar T'_O() = \bar G'_O() = \texttt{true}$.
\end{example}

\paragraph*{\bf Correctness} We restate here the theorem
from~\cite[Theorem 4]{DBLP:conf/birthday/MascellaniP02}:

\begin{theorem}
\label{theo:magicsets}
Let $\Pi$ be a datalog program with query predicate $Q$.  Fix any
moding of its symbols, and let $\Pi_O$, be the corresponding magic
program. Then, at fixpoint, the IDB $\bar Q$ computed by $\Pi$ equals
the IDB $\bar Q_O$ computed by $\Pi_O$.
\end{theorem}

We will re-prove the theorem by showing, importantly, that the
equivalence of the two programs follows from the three rules, FGH,
Stratification, and Fixpoint, and therefore can be checked
automatically by a verifier.

\subsection{\property}

To prove the equivalence of $\Pi$ and $\Pi_O$ we only need one
property of the optimized program, which we call \property.  We prove
here that the specific rewriting in Def.~\ref{def:magicsets} ensures
that the resulting program $\Pi_0$ satisfies \property; later we will
show that any program $\Pi_O$ satisfying \property\ is equivalent to
$\Pi$.  Our proof here consists of several applications of the chase
procedure, which we briefly review here.

\paragraph*{The Chase} Suppose that the following constraint holds:
$\forall x (\Phi(x) \Rightarrow \Psi(x))$.  Then the following
equivalence holds:
\begin{align*}
  \forall x (\Phi(x) \wedge \Gamma(x) \equiv & \Phi(x) \wedge \Psi(x) \wedge \Gamma(x))
\end{align*}
By ``applying the chase'' we mean rewriting the formula
$\Phi(x) \wedge \Gamma(x)$ to
$\Phi(x) \wedge \Psi(x) \wedge \Gamma(x)$.  The ``back-chase''
proceeds in reverse, i.e. it removes $\Psi(x)$.  Both the chase and
the back-chase can be encoded in an \eqsat\ system (see
Sec.~\ref{sec:semantic-opt}), and therefore proofs based on
chase/back-chase can be derived and checked automatically.

Consider a program $\Pi$ and its magic-set rewriting $\Pi_O$ in
Definition~\ref{def:magicsets}.  Consider the least fixpoint of
$\Pi_O$; as usual we denote by $\bar R_O', \bar R_O$ the instances in
this least fixpoint; this notation extends to the case when $R$ is
an EDB symbol, then simply $\bar R_O\defeq R$.

\begin{definition}[Boundedness] \label{def:boundedness}
  Let $R$ be an IDB symbol occurring in the program $\Pi$.  We say
  that an instance $R$ is {\em bounded} w.r.t. $\Pi_O$ if it
  satisfies:
  \begin{align}
    \forall u \forall v (\bar R'_O(u) \wedge R(u,v) \Rightarrow & \bar R_O(u,v)) \label{eq:assumption:magic:sets:main:property}
  \end{align}
\end{definition}
This definition holds trivially for each EDB relation $R$: each such
relation is bounded, because $\bar R_O = R$ by definition.

\begin{lemma}[\property] \label{lemma:magic:sets:main:property}
  Consider an instance of all IDBs of the program $\Pi$ that is
  bounded w.r.t. $\Pi_O$.  Then, for every rule $r$ in the original
  program $\Pi$:
  \begin{align*}
    r:\  A \cd & \bm B
  \end{align*}
  the following equivalence holds:
  \begin{align}
    \bar A'_O \wedge \bm B \equiv & \bar A'_O \wedge \bar{\bm B}_O' \wedge \bm B \label{eq:equivalence:needed:for:magic:proof}
  \end{align}
\end{lemma}

To help the reader parsing
Eq.~\eqref{eq:equivalence:needed:for:magic:proof}, we note that the
atoms in the sequence $\bm B$ refer to the bounded relational instance
(the symbols appearing in $\Pi$), while the atoms 
$\bar{A}'_O$ and
$\bar{\bm B}_O'$ refer to the fixpoint of the program $\Pi_O$.

\begin{proof}
  This proof is the place where we use the magic
  rule~\eqref{eq:def:of:r:prime} in item~\ref{item:def:magicsets:1} of
  Definition~\ref{def:magicsets}.  Specifically, we write
  $\bm B = B_1 \wedge \cdots \wedge B_k$, where the order of the atoms
  is that chosen in Def.~\ref{def:magicsets}
  item~\ref{item:def:magicsets:1}.  At fixpoint, for each $\ell=1,k$,
  the magic rule~\eqref{eq:def:of:r:prime} becomes the following implication by the Fixpoint rule  
  (see Eq.~\eqref{eq:constraint:at:fixpoint}):
  \begin{align}
    \bar A'_O \wedge \bigwedge_{i=1,\ell-1} \bar B_{i,O} \Rightarrow & \bar B_{\ell,O}' \label{eq:r:prime:constraint}
  \end{align}
  We will chase repeatedly the LHS
  of~\eqref{eq:equivalence:needed:for:magic:proof} with the
  implications~\eqref{eq:r:prime:constraint}
  and~\eqref{eq:assumption:magic:sets:main:property} to arrive at the
  RHS.

  For each $\ell=0,k$, denote by $\Phi_{\ell}$ the following
  sentence:
  \begin{align*}
    \Phi_\ell \defeq & \bar A' \wedge \bigwedge_{i=1,\ell}\left(\bar B_{i,O}'\wedge  B_i\right)\wedge \bigwedge_{i=\ell+1,k}  B_i
  \end{align*}
  Eq.~\eqref{eq:equivalence:needed:for:magic:proof} asserts that
  $\Phi_0 \equiv \Phi_k$, and we prove it by showing that, for every
  $\ell=1,k$, the following holds:
  \begin{align*}
    \Phi_{\ell-1} \equiv & \Phi_{\ell}
  \end{align*}
  This follows from the following chase steps:
  \begin{align*}
    \bar A'_O \wedge \bigwedge_{i=1,\ell-1} \left(\bar B_{i,O}' \wedge  B_i\right) \wedge  B_\ell  \equiv &\bar A'_O \wedge \bigwedge_{i=1,\ell-1} \left(\bar B_{i,O}' \wedge \bar B_{i,O} \wedge  B_i\right) \wedge  B_\ell  &&\mbox{Chase with~\eqref{eq:assumption:magic:sets:main:property}}\\
    \equiv & \bar A'_O \wedge \bigwedge_{i=1,\ell-1} \left(\bar B_{i,O}' \wedge \bar B_{i,O} \wedge  B_i\right)\wedge \bar B_{\ell,O}' \wedge  B_\ell  &&\mbox{Chase with~\eqref{eq:r:prime:constraint}}\\
    \equiv & \bar A'_O \wedge \bigwedge_{i=1,\ell} \left(\bar B_{i,O}' \wedge  B_i\right)  &&\mbox{Back-chase with~\eqref{eq:assumption:magic:sets:main:property}}
  \end{align*}
  By conjoining both sides of the equivalence above with
  $\bigwedge_{i=\ell+1,k} B_i$ we obtain
  $\Phi_{\ell-1}\equiv \Phi_\ell$, as required.
\end{proof}

\begin{example}
  We describe the \property\ for the running
  Example~\ref{ex:magic:running}.  Let $\bar R_O', \bar Q_O'$ be the
  outputs of the magic predicates of the program $\Pi_O$.  Notice that
  $Q_O'() \equiv \texttt{true}$.  The \property\ asserts the following:
  if $R, Q$ are two {\em bounded} instances, meaning that they satisfy:
  \begin{align}
    \forall x (\bar R_O'(x) \wedge R(x) \Rightarrow & \bar R_O(x))
&   \forall x (\bar Q_O'() \wedge Q(x) \Rightarrow & \bar Q_O(x)) \label{eq:constraint:for:chasing}
  \end{align}
  then the following holds (one constraint for each of the 3 rules of
  $\Pi$):
  \begin{align*}
    \forall x (\bar R'_O(x) \wedge V(x) \equiv & \bar R'_O(x) \wedge \bar V'_O() \wedge V(x))\\
    \forall x (\bar R'_O(x) \wedge T(x,y,z) \wedge R(y) \wedge R(z) \equiv & \bar R'_O(x) \wedge \bar T_O'() \wedge T(x,y,z) \wedge \bar R_O'(y) \wedge R(y) \wedge \bar R'_O(z) \wedge R(z))\\
    \forall x (\bar Q_O'() \wedge G(x) \wedge R(x) \equiv & \bar Q_O'() \wedge \bar G_O'() \wedge G(x) \wedge \bar R_O'(x) \wedge R(x))
  \end{align*}
  We saw in Example~\ref{ex:magic:running:2} that
  $\bar V'_O() \equiv \bar T'_O() \equiv \bar G'_O() \equiv
  \texttt{true}$, hence to check \property\ it suffices to check
  only two equivalences:
  \begin{align*}
    \forall x (\bar R'_O(x) \wedge T(x,y,z) \wedge R(y) \wedge R(z) \equiv & \bar R'_O(x) \wedge  T(x,y,z) \wedge \bar R_O'(y) \wedge R(y) \wedge \bar R'_O(z) \wedge R(z))\\
    \forall x (\bar Q_O'() \wedge G(x) \wedge R(x) \equiv & \bar Q_O'() \wedge G(x) \wedge \bar R_O'(x) \wedge R(x))
  \end{align*}
  We leave it up to the reader to check that both constraints can be
  derived by repeated chase and back-chase using the
  constraints~\eqref{eq:constraint:for:chasing} and the following
  constraints derived from the property of the fixpoint of $\Pi_O$
  (see Eq.~\eqref{eq:constraint:at:fixpoint}):
  \begin{align*}
    \forall x (\bar R_O'(x) \wedge T(x,y,z) \Rightarrow & \bar R_O'(y))\\
    \forall x (\bar R_O'(x) \wedge T(x,y,z) \wedge \bar R_O(y) \Rightarrow & \bar R_O'(z))
  \end{align*}
\end{example}

\subsection{Correctness Proof of Magic Set Rewriting}

We prove a stronger claim than Theorem~\ref{theo:magicsets}: we prove
the correctness of magic set rewriting using only the FGH, the
Stratification, and the Fixpoint rules.  While at a high level our
proof is inspired by Mascellani and
Pedreschi~\cite{DBLP:conf/birthday/MascellaniP02}, it differs in that
we use the least fixpoint semantics of a datalog program rather than
the minimal model semantics.  This allows us to prove the equivalence
$\Pi \equiv \Pi_O$ using only the three rules, FGH, Stratification,
and Fixpoint.  Moreover, our proof here is independent of the
particular definition of the magic set rewriting used to define
$\Pi_O$ (Def.~\ref{def:magicsets}), and, instead, applies to any
program $\Pi_O$ that satisfies \property.

Let $X$ denote the tuple of IDBs of the original program $\Pi$, and
fix a moding.  Then we let $X'$ denote the tuple of the magic IDBs, in
other words:
\begin{align*}
  X = & (R_1(\bm u_1, \bm v_1), R_2(\bm u_2, \bm v_2), \ldots) \\
  X' = & (R'_1(\bm u_1), R'_2(\bm u_2), \ldots)
\end{align*}
We denote by $X' \wedge X$ their pairwise conjunction:
\begin{align*}
  X' \wedge X \defeq & (R'_1(\bm u_1)\wedge R_1(\bm u_1, \bm v_1), R'_2(\bm u_2)\wedge R_2(\bm u_2, \bm v_2), \ldots)
\end{align*}
%


Consider two programs $\Pi$ and $\Pi_O$:
\begin{align}
  \Pi: \ X \cd &F(X) & \Pi_O:\ X'_O \cd & F'(X_O',X_O) \label{eq:magic:f} \\
               && X_O \cd & X_O' \wedge F(X_O) \nonumber
\end{align}
where $F, F'$ are two monotone functions. The second line in $\Pi_O$
corresponds directly to the optimized rules in
item~\ref{item:def:magicsets:3} of Definition~\ref{def:magicsets}.
The first line corresponds to the magic predicates.  For the moment we
allow $F'$, which
should not be confused with a derivative\footnote{We use the notation
  $F'$ to follow the convention
  in~\cite{DBLP:conf/birthday/MascellaniP02} where the magic predicate
  for $R$ is denoted $R'$.}, to  be any monotone function.

Let $\bar X_O', \bar X_O$ denote the fixpoint of the program $\Pi_O$.
We generalize Definition~\ref{def:boundedness}:

\begin{definition}[Boundedness] We say that $X$ is {\em bounded}
  w.r.t. $\Pi_O$ if it satisfies:
  \begin{align*}
    \bar X'_O \wedge X \Rightarrow & \bar X_O
  \end{align*}
\end{definition}

\begin{definition}[\property]
  We say that $\Pi_O$ satisfies \property, if every bounded $X$
  satisfies:
  \begin{align}
    \bar X'_O \wedge F(X) = & \bar X'_O \wedge F(\bar X'_O \wedge X)\label{eq:consequence:magic:sets:main:property}
  \end{align}
\end{definition}

\begin{theorem} \label{th:magic:fgh} Consider two programs
  $\Pi, \Pi_O$ where $\Pi_O$ satisfies \property.  Then the following
  holds:
  \begin{align}
     \bar X_O' \wedge \bar X = & \bar X_O \label{eq:origina:magic:relationship}
  \end{align}
  where $\bar X$ is the fixpoint of $\Pi$ and $\bar X_O', \bar X_O$ is
  the fixpoint of $\Pi_O$.
\end{theorem}

Theorem~\ref{th:magic:fgh} immediately implies
Theorem~\ref{theo:magicsets}.  Indeed, if $F(X)$ is the ICO of the
original datalog program and $F'$ defines the magic predicates, as per
items~\ref{item:def:magicsets:1} and~\ref{item:def:magicsets:2} of
Def.~\ref{def:magicsets}, then the optimized program satisfies
\property, by Lemma~\ref{lemma:magic:sets:main:property}.  Thus, the
identity~\eqref{eq:origina:magic:relationship} holds, and, in
particular, the following for the query predicate $Q$:
\begin{align*}
  Q'_O() \wedge Q(\bar u) = & Q_O(\bar u)
\end{align*}
Theorem~\ref{theo:magicsets} follows from the fact that
$Q'_O() \equiv \texttt{true}$; this is the only place where we need
item~\ref{item:def:magicsets:2} of Definition~\ref{def:magicsets}.

In the rest of this subsection we prove Theorem~\ref{th:magic:fgh}.

\begin{proof} [Proof of Theorem~\ref{th:magic:fgh}]  We use four steps.

  \paragraph*{\bf Step 1: FGH rule for $\Pi_O \equiv \Pi_{\text{copy}}$.}  We start
  by using the FGH-rule to prove that the program $\Pi_O$ is
  equivalent to the following program:
  \begin{align*}
    \Pi_{\text{copy}}:\ \ \ \ \ X'_O \cd & F'(X_O',X_{\text{copy}}) \\
                     X_{\texttt{copy}} \cd & X_O' \wedge F(X_{\text{copy}}) \\
                     X_O \cd & X_O' \wedge F(X_O)
  \end{align*}
  The new program creates a copy $X_{\text{copy}}$ of $X_O$.
  Intuitively, it is obvious that the new program computes the same
  IDBs as the original program.  Formally, one can check that the
  following function $G$ is a homomorphism mapping the state
  $(X_O', X_O)$ of $\Pi_O$ to the state $(X_O', X_{\text{copy}}, X_O)$
  of $\Pi_{\text{copy}}$: $G(X_O', X_O) \defeq (X_O', X_O, X_O)$.  If
  the fixpoint of $\Pi_O$ is $(\bar X'_O, \bar X_O)$, then the
  homomorphism implies that the fixpoint of $\Pi_{\text{copy}}$ is
  $(\bar X_O', \bar X_O, \bar X_O)$.

  \paragraph*{\bf Step 2: Stratification rule for  $\Pi_{\text{copy}} \equiv \Pi_1$.} We apply
  the stratification rule to $\Pi_{\text{copy}}$, and write it as:
  \begin{align}
    &   \begin{array}[b]{rrl}
          \Pi_0: & X'_O \cd & F'(X_O',X_{\text{copy}}) \\
                 & X_{\texttt{copy}} \cd & X_O' \wedge F(X_{\text{copy}}) \\
          \\ \hline
          \\
          \Pi_1: & X_O \cd & \bar X_O' \wedge F(X_O)
    \end{array}
\label{eq:magic:program:pi:o:pime}
  \end{align}
  We denote by $\bar X_O', \bar X_{\text{copy}} 
   (= \bar X_O)$ the
  fixpoint of the first stratum $\Pi_0$.  Importantly, the second
  stratum $\Pi_1$ uses the fixpoint $\bar X_O'$ as an EDB.

  \paragraph*{\bf Step 3: Fixpoint rule for the invariant $\Phi(X)$.}
  Next, we prove that the state $X$ of the original program $\Pi$
  satisfies the following invariant:
  \begin{align*}
    \Phi(X) \equiv & (\bar X_O' \wedge X \Rightarrow \bar X_O)
  \end{align*}
  The invariant holds trivially when $X=\emptyset$.  Assuming that it
  holds for $X$, we check that it also holds for $F(X)$:
  \begin{align*}
    \bar X_O' \wedge F(X) \equiv & \bar X_O' \wedge F(\bar X_O'\wedge X) && \mbox{By \property,~\eqref{eq:consequence:magic:sets:main:property}}   \\ 
    \Rightarrow & \bar X_O' \wedge F(\bar X_O) && \mbox{Induction hypothesis $\Phi(X)$} \\
    \Rightarrow & \bar X_O && \mbox{Fixpoint rule: $\bar X_O$ is the least fixpoint of $\Pi_O$}
  \end{align*}

  \paragraph*{\bf Step 4: FGH-rule for $\Pi \equiv \Pi_1$.}
  Consider now the original program $\Pi$ and the program $\Pi_1$:
  their states are $X$ and $X_O$ respectively.  We claim that, under
  the invariant $\Phi(X)$, the following function $G$ is a
  homomorphism from $\Pi$ to $\Pi_O$:
  \begin{align*}
    G(X) \defeq & \bar X_O' \wedge X
  \end{align*}
  We prove $G(F(X)) = H(G(X))$ where $H$ is the ICO of the
  program $\Pi_1$, Eq.~\eqref{eq:magic:program:pi:o:pime}.  We expand both sides:
  \begin{align*}
    G(F(X)) = & \bar X_O' \wedge F(X) & H(G(X)) = & \bar X_O' \wedge F(\bar X_O' \wedge X)
  \end{align*}
  Their equality follows immediately from
  \property~\eqref{eq:consequence:magic:sets:main:property}.

  By the FGH-rule~\eqref{eq:fgh:simplified}, it follows that
  $\bar X_O' \wedge \bar X = \bar X_O$, where $\bar X, \bar X_O$ are
  the fixpoints of $\Pi$ and $\Pi_1$ respectively.  We have already
  shown that the fixpoint of $\Pi_1$ is equal to that of the magic
  optimized program $\Pi_O$, and this completes the proof.
\end{proof}

\subsection{Discussion}

\paragraph*{\bf Necessity of Stratification} A question is whether the
stratification rule is redundant, more precisely whether the
correctness proof of the magic set rewriting could be completed using
only the FGH rule.  The answer is {\em no}.  To see this, observe that
if two datalog programs are proven equivalent by the FGH rule, then
the number of iterations needed by the second program to reach a
fixpoint is at most equal to that needed by the first program, see
Corollary~\ref{cor:number:iterations}.  However, in some cases, the
magic-set optimized program may require a significantly larger number
of iterations, proving that the FGH-rule is insufficient to prove
their equivalence.

For a concrete example, consider the programs $\Pi$ and $\Pi_O$ in
Example~\ref{ex:magic:running}, and assume that $T(x,y,z)$ is a
complete binary tree of depth $n$ and with $2^n$ nodes, where $x$ is
the parent and $y$ and $z$ are the two children.  The original program
$\Pi$ reaches its fixpoint after $n$ iterations, while the optimized
program $\Pi_O$ may require $2^n$ iterations, since it performs a
left-deep traversal of the tree.

\paragraph*{\bf The Flexibility of Moding} A nice feature of
the framework
introduced by Mascellani and
Pedreschi~\cite{DBLP:conf/birthday/MascellaniP02} is that it decouples
the correctness proof from the performance consideration.  In
practice, the moding, which is also called {\em adornment} or {\em
  binding pattern}, is determined by a {\em Sideway Information
  Passing} (SIP) algorithm.  However, the correctness proof holds for
any moding, even if it is not the result of a SIP.  We illustrate this
decoupling with a classic example.

\begin{example}\label{ex:sg}
  The {\em same generation} program, and its magic rewriting, are the following:
  \begin{align*}
    &
      \begin{array}[b]{rrl}
        \Pi:& S(x,y) \cd & H(x,y) \\
            & S(x,y) \cd & U(x,p) \wedge S(p,q) \wedge D(q,y)\\
            & Q(y) \cd & S(a,y)
      \end{array}
                        \begin{array}[b]{rrl}
                          \Pi_O:& Q'_O() \cd & \\
                                & S'_O(p) \cd & S'_O(x) \wedge U(x,p) \\
                                & S'_O(a) \cd & Q'_O() \\
                                & S_O(x,y) \cd & S_O'(x) \wedge H(x,y) \\
                                & S_O(x,y) \cd & S_O'(x) \wedge U(x,p) \wedge S_O(p,q) \wedge D(q,y)\\
                                & Q_O(y) \cd & Q_O'() \wedge S_O(a,y)
                        \end{array}
  \end{align*}
  The EDBs $U, D, H$ stand for ``up'', ``down'', and ``horizontal''.

  The SIP-based magic rewriting will adorn $S$ with $+-$, because
  $Q(y) = S(a,y)$ where $a$ is a constant, and will order the atoms
  in the rule $S(x,y) \cd U(x,p) \wedge S(p,q) \wedge D(q,y)$ as
  shown, so as to facilitate sideways information passing.  This
  leads to the optimized program $\Pi_O$ shown above.

  What happens if we chose a different order in the rule for $S(x,y)$?
  Assuming the order is $S(x,y) \cd U(x,p) \wedge D(q,y) \wedge
  S(p,q)$.  The new magic set rewriting will have a modified rule for $S'_O(p)$:
  \begin{align*}
    S'_O(p) \cd & S'_O(x) \wedge U(x,p) \wedge D(q,y)
  \end{align*}
  We have introduced a redundant cartesian product with $D(q,y)$: the
  new magic program is still correct, but less efficient.

  What happens if we choose the adornment $-+$ for $S$?  Also, assume
  that we reorder the atoms in the second rule for $S$ to: $S(x,y) \cd
  D(q,y) \wedge S(p,q) \wedge U(x,p)$.  Then the magic rewriting
  becomes
  \begin{align*}
    Q'_O() \cd & \\
    S'_O(q) \cd & S'(y) \wedge D(q,y) \\
    S'_O(y) \cd & Q'_O() \\
    S_O(x,y) \cd & S_O'(y) \wedge H(x,y) \\
    S_O(x,y) \cd & S_O'(y) \wedge D(q,y), \wedge S_O(p,q) \wedge U(x,p)\\
    Q_O(y) \cd & Q_O'() \wedge S_O(a,y)
  \end{align*}
%
%
  The third rule above\footnote{Strictly speaking the rule is
    unsafe.  We allow it here for illustration.}, $S'_O(y) \cd
  Q'_O()$, defines $S'_O$ as the entire domain.  The ``optimized''
  program is still correct, but less efficient than the original one.
\end{example}

\paragraph*{\bf Multiple Modings} 
Finally, we explain how to circumvent an apparent limitation of the
framework of Mascellani and
Pedreschi~\cite{DBLP:conf/birthday/MascellaniP02}: the fact that each
predicate symbol $R$ can have a single mode.  When multiple modes are
needed, then this can be achieved by making copies of the IDBs and
moding them differently.  We illustrate this with another classic
example of magic set rewriting.

\begin{example}\label{ex:reverse:sg}
  Consider the reverse-same-generation program:
  \begin{align*}
    \Pi:\ \ \ S(x,y) \cd & H(x,y) \\
              S(x,y) \cd & U(x,p) \wedge S(q,p) \wedge D(q,y) \\
              Q_1(y) \cd & S(a,y)
  \end{align*}
  The only change is that $S(p,q)$ is replaced by $S(q,p)$ in the
  second rule.  The SIP algorithm requires us to adorn $S$ in two
  ways, both $+-$ and $-+$.  To achieve that it suffices to create two
  copies of $S$, the left $Sl$ and the right $Sr$.  More precisely,
  consider the program:
  \begin{align*}
    \Pi':\ \ \ Sl(x,y) \cd & H(x,y) \\
               Sr(x,y) \cd & H(x,y) \\
               Sl(x,y) \cd & U(x,p)\wedge Sr(q,p) \wedge D(q,y) \\
               Sr(x,y) \cd & D(q,y) \wedge Sl(q,p) \wedge U(x,p) \\
               Q(y) \cd & Sl(a,y)
  \end{align*}
  The FGH-rule proves formally that $\Pi$ and $\Pi'$ are equivalent.
  More precisely, consider the following function $G$ mapping the
  state $S$ of $\Pi$ to the state $Sl, Sr$ or $\Pi'$:
  $G(S) \defeq (S, S)$.  One can check immediately that $G$ is a
  homomorphism.  It follows that $\Pi$ and $\Pi'$ compute the same IDB
  $S = Sl = Sr$. 

  Following SIP, we define the following modings for $\Pi'$,
  $Sl(+-), Sr(-+), Q(-)$, and also use the ordering of the rules as
  show above, where, in the rule for $Sr$, we have switched the order
  of $D$ and $U$.  Then, the magic set transformation in
  Def.~\ref{def:magicsets} produces the following optimized program:
\begin{align*}
  \Pi_O:\ \ \    Q'_O() \cd &\\
                 Sr'_O(p) \cd & Sl'_O(x)\wedge U(x,p) \\
                 Sl'_O(q) \cd & Sr'_O(y)\wedge D(q,y) \\
                 Sl'_O(a) \cd & Q'_O()\\
                 Sl_O(x,y) \cd & Sl'_O(x)\wedge H(x,y) \\
                 Sr_O(x,y) \cd & Sr'_O(y) \wedge H(x,y) \\
                 Sl_O(x,y) \cd &Sl'_O(x)\wedge U(x,p)\wedge Sr_O(q,p) \wedge D(q,y) \\
                 Sr_O(x,y) \cd &Sr'_O(y)\wedge D(q,y) \wedge Sl_O(q,p) \wedge U(x,p) \\
                 Q_O(y) \cd & Q'_O(),Sl_O(a,y)
\end{align*}
\end{example}




\end{document}